\documentclass[final
,onefignum,onetabnum]{siamonline171218}

\usepackage{subfigure}

\usepackage{lipsum}
\usepackage{amsfonts}
\usepackage{graphicx}
\usepackage{epstopdf}
\usepackage{algorithmic}

\graphicspath{{figures/}}

\usepackage{enumitem}
\setlist[enumerate]{leftmargin=.5in}
\setlist[itemize]{leftmargin=.5in}

\newsiamremark{remark}{Remark}
\newsiamremark{hypothesis}{Hypothesis}
\crefname{hypothesis}{Hypothesis}{Hypotheses}
\newsiamthm{claim}{Claim}

\headers{Growing stripes, with and without wrinkles}{M. Avery, R. Goh, O. Goodloe, A. Milewski, and A. Scheel}

\title{Growing stripes, with and without wrinkles\thanks{
\funding{This work was partially supported by NSF grant DMS--1612441.}}}

\author{M. Avery\thanks{University of Minnesota, Minneapolis, MN 
  (\email{avery142@umn.edu})}
\and R. Goh\thanks{Boston University, Boston, MA 
  (\email{rgoh@bu.edu})}
\and O. Goodloe\thanks{Arizona State University, Tempe, AZ
  (\email{ogoodloe@asu.edu})}
\and A. Milewski\thanks{University of Bristol, Bristol, UK 
  (\email{am16053@my.bristol.ac.uk})}
\and A. Scheel\thanks{University of Minnesota, Minneapolis, MN 
  (\email{scheel@umn.edu},\url{www.umn.edu/\~scheel})}}

\usepackage{amsopn}

\usepackage{enumitem,amssymb}
\newlist{todolist}{itemize}{2}
\setlist[todolist]{label=$\square$}

 \newcommand{\R}{\mathbb{R}}

\newcommand{\Z}{\mathbb{Z}}

\newcommand{\rmd}{\mathrm{d}}
\newcommand{\rme}{\mathrm{e}}
\newcommand{\rmi}{\mathrm{i}}
\newcommand{\rmO}{\mathrm{O}}

\newcommand{\sign}{\mathrm{sign}}
\newcommand{\eps}{\varepsilon}

\newcommand{\rr}{\rho}

\renewcommand{\Re}{\mathrm{Re}\,}

\ifpdf
\hypersetup{
  pdftitle={Growing stripes, with and without wrinkles},
  pdfauthor={M. Avery, R. Goh, O. Goodloe, A. Milewski, and A. Scheel}
}
\fi

\begin{document}

\maketitle

\begin{abstract}
 We present results on stripe formation in the Swift-Hohenberg equation with a directional quenching term. Stripes are ``grown'' in the wake of a moving parameter step line, and we analyze how the orientation of stripes changes depending on the speed of the quenching line and on a lateral aspect ratio. We observe stripes perpendicular to the quenching line, but also stripes created at oblique angles, as well as periodic wrinkles created in an otherwise oblique stripe pattern. Technically, we study stripe formation as traveling-wave solutions in the Swift-Hohenberg equation and in reduced Cahn-Hilliard and Newell-Whitehead-Segel  models, analytically, through numerical continuation, and in direct simulations. \end{abstract}



\section{Introduction}

Striped phases appear in a plethora of contexts, from sand \cite{sand} and icicle ripples \cite{icicle}, to convection roll  \cite{bodenschatz} and precipitation patterns \cite{thomas}, to bacterial colony growth \cite{eshel} or the formation of presomites in early development \cite{digit}, and to ion-beam milling \cite{bradley}, dip-coating \cite{dipstripe}, lamellar crystal growth \cite{double,zigzageutectic}, and water jet cutting \cite{friedrich}. Simple understanding of such patterns is often based on the weak instability of a trivial, spatially constant state against perturbations that are periodic in space with wavenumbers close to a critical wavenumber $k_\mathrm{c}>0$. While the selection of this specific wavenumber may be due to quite different and complex physical mechanisms, the resulting phenomena, at least for weak instabilities, often bear a striking resemblance. In many cases, phenomena are well captured by simple, universal models such as the Swift-Hohenberg equation \cite{sh},
\begin{equation}
u_t = -(1+\Delta)^2u + \rr u - u^3 \label{e:sh}
\end{equation}
posed on $(x,y)\in \R^2$, with $0<\rr\ll 1$. The linear part will enhance Fourier modes $\rme^{\rmi( k_x x + k_y y)}$ with $k=\sqrt{k_x^2+k_y^2}\sim k_\mathrm{c}= 1$, and the nonlinear part leads to (local) saturation and competition of these modes, such that the locally dominant observed pattern is of the form $\rme^{\rmi( k_x x + k_y y)}+c.c.$, $k\sim 1$, for some orientation of the wave vector $(k_x,k_y)$ depending on space $(x,y)$. Indeed, starting the system with random initial conditions in a large domain leads to quite complex, incoherent structures; see Figure \ref{f:1}. One mostly observes such sinusoidal stripe patterns with some locally chosen orientation, but these local stripe domains are bordered by a plethora of defects, including grain boundaries, dislocations, and disclinations \cite{pismen}. 

\begin{figure}\centering
\includegraphics[width=.9\textwidth]{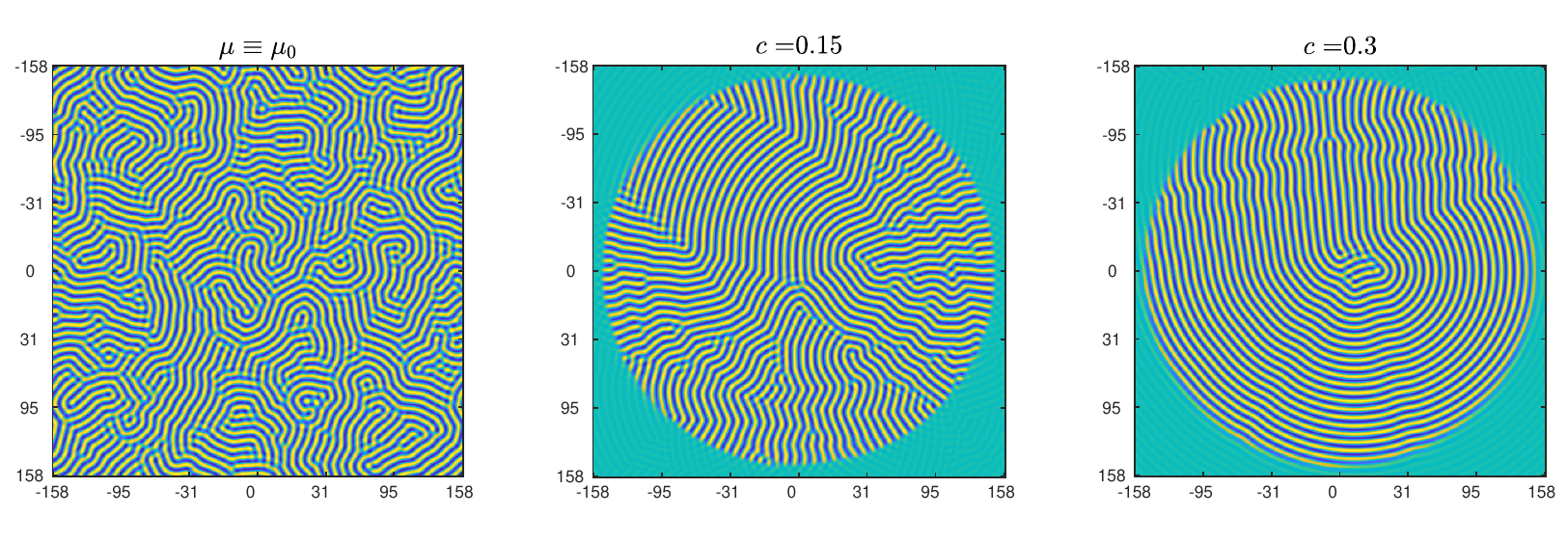}
\caption{Patterns in the Swift-Hohenberg equation \eqref{e:sh} from random initial conditions with constant $\rr\equiv 0.25$ (left), compared to patterns resulting from directional quenching $\rr=\mu \sign(c t-|(x,y)|)$, $\mu=0.25$ (center, left). Notice the predominant parallel orientation  of stripes relative to the quenching boundary for large speeds and the wrinkly structures created for smaller speeds.}\label{f:1}
\end{figure}

\begin{figure}[h!]\centering
	\subfigure{
		\includegraphics[width=.48\textwidth]{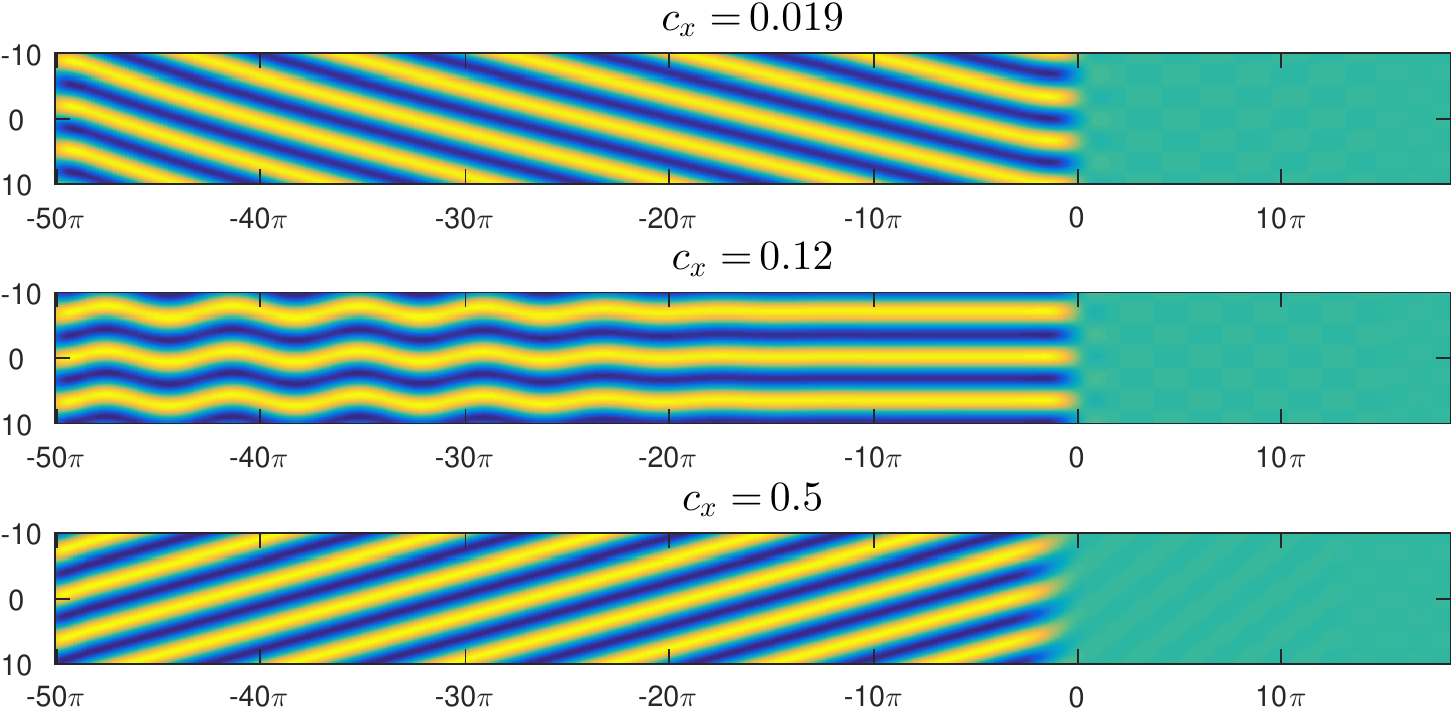}}
	\hfill
	\subfigure{
		\includegraphics[width=.48\textwidth]{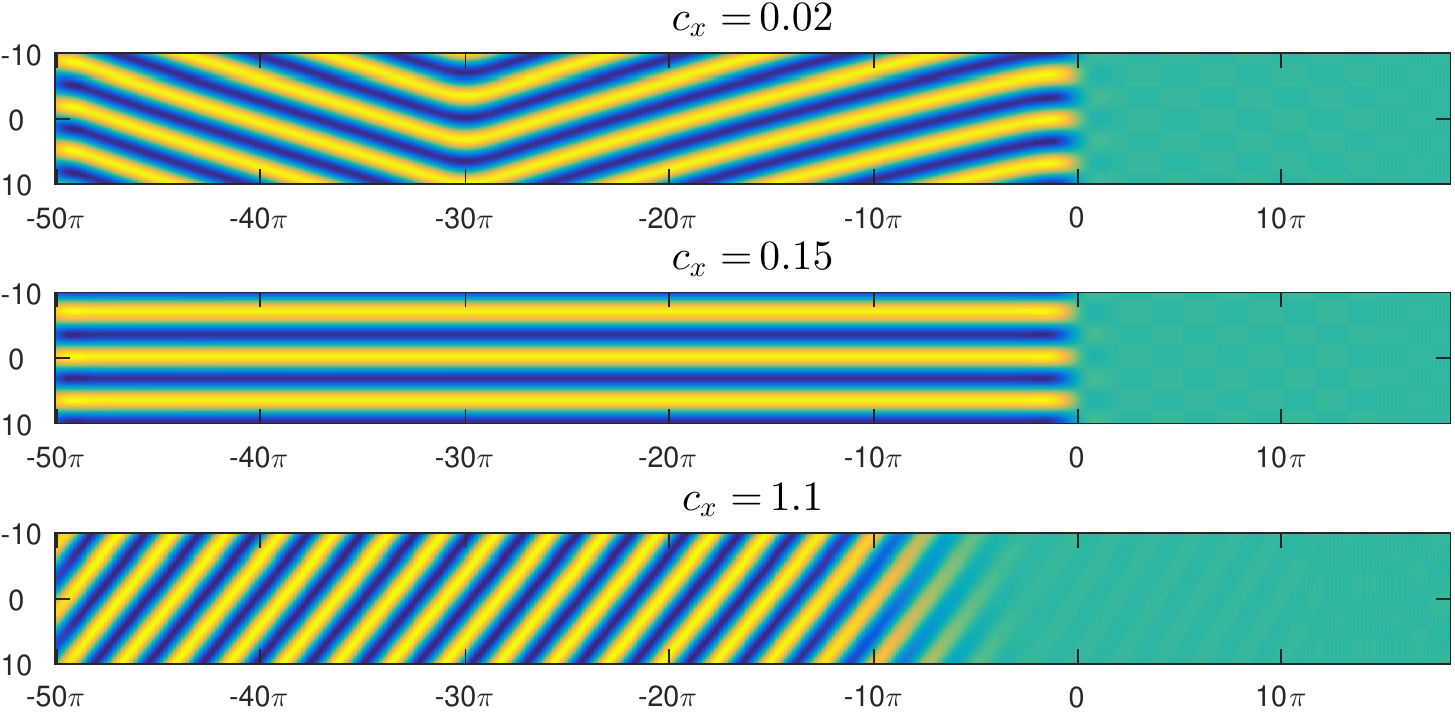}}
	\hfill
	\caption{Transitions in the formation of striped patterns in the Swift-Hohenberg equation as analyzed and predicted, here, for $\rr=0.25$, $k_y=0.95$. Simulations carried out in a co-moving frame with quenching line at $x=0$. Transitions observed (top to bottom, left to right), are from oblique to zigzag to straight, back to oblique, until stripe formation detaches.}\label{f:2}
\end{figure}

\paragraph{Directional quenching} In many of the physical contexts listed above, patterns do not arise through the type of quenching that is captured by the scenario of small random initial perturbations of an unstable state, but one rather observes that the domain in which patterns form grows in time. There has been quite some interest in the interplay between  growth processes and pattern formation mechanisms, in particular since the phenomenology of both the growth process and the pattern formation mechanism can quite dramatically influence each other. We shall focus here on the effect of the growth mechanism on the pattern formation, and neglect the reverse effect, as a first approximation. We therefore assume that the parameter $\rr$ in \eqref{e:sh} is time- and space-dependent $\rr=\rr(t,x,y)$ with $\rr(t,x,y)\equiv\mu>0$ in a time-dependent region $(x,y)\in\Omega_t$ and  $\rr(t,x,y)\equiv-\mu<0$ in the complement $(x,y)\not\in\Omega_t$. We illustrate the striking ``regularity''
of the resulting patterns in Figure \ref{f:1} where $\Omega_t=\{|(x,y)|<c_x t\}$ for some $c_x>0$. 

Clearly, one can envision exploiting the geometry of $\partial\Omega_t$ in order to ``manufacture'' particular crystalline configurations. The first ingredient for a systematic prediction of the resulting patterns is an understanding of patterns arising from approximately planar interfaces, where $\Omega_t=\{x<c_x t\}$. A second simplification is the assumption of transverse periodicity $y\in \R/L_y\Z$ exploited in the simulations of Figure \ref{f:2}. 

The focus of this paper is the parameter regime where stripes are perpendicular, or almost perpendicular to what we shall refer to as the quenching line $\partial\Omega_t$. In this regime, $k_x\sim 0$ and therefore $k_y\sim 1$. We think of $k_y=2\pi/L_y$ and $c_x$ as control parameters, and of $k_x$ as a parameter that is selected by the process, thus dependent on $k_y$ and $c_x$. 

Direct simulations confirm a somewhat delicate behavior of the selection process when $k_y\sim 1$ and $c_x$ is gradually increased. Figure \ref{f:2} shows how, for $k_y\lesssim 1$, one observes transitions from oblique to zigzagging to perpendicular stripes, that will eventually become oblique again before they detach from the quenching line for larger speeds.

\begin{figure}[h]\centering
\includegraphics[width=1\textwidth]{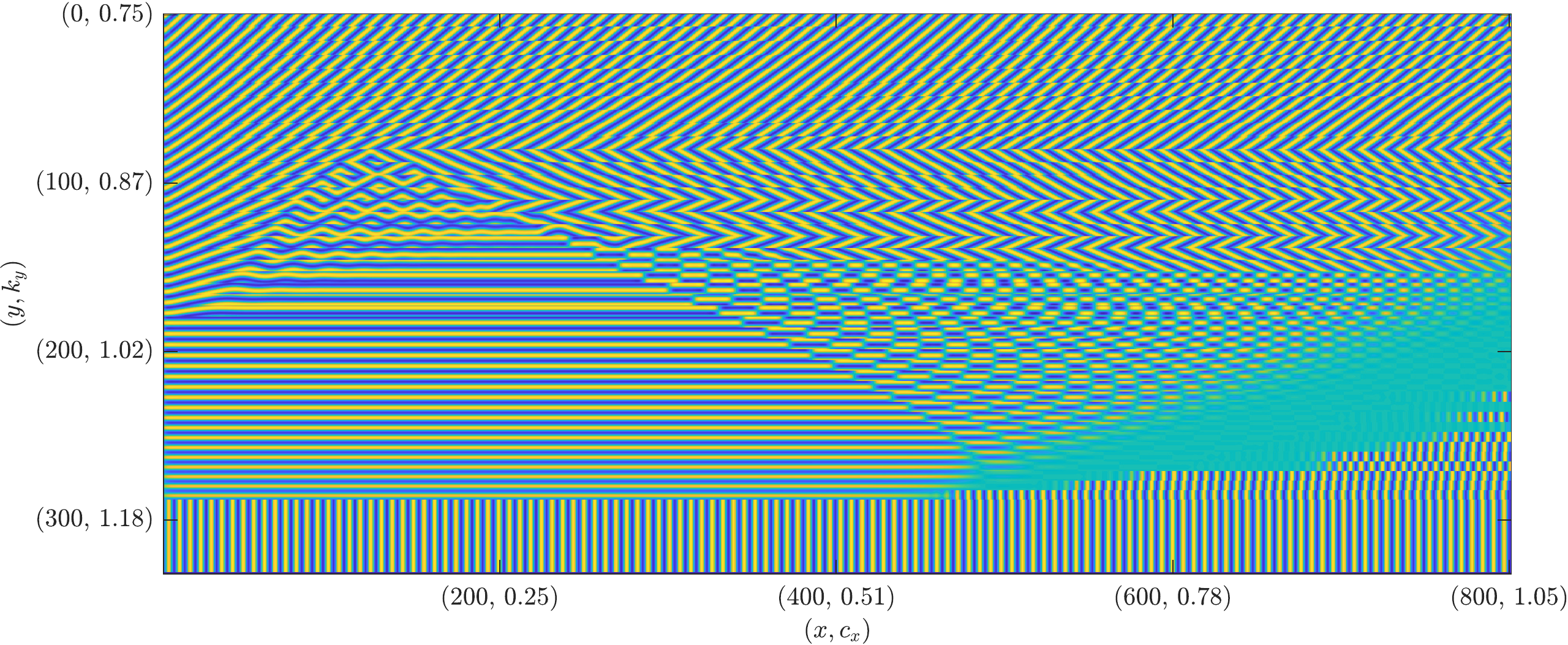}\\
\includegraphics[width=1\textwidth]{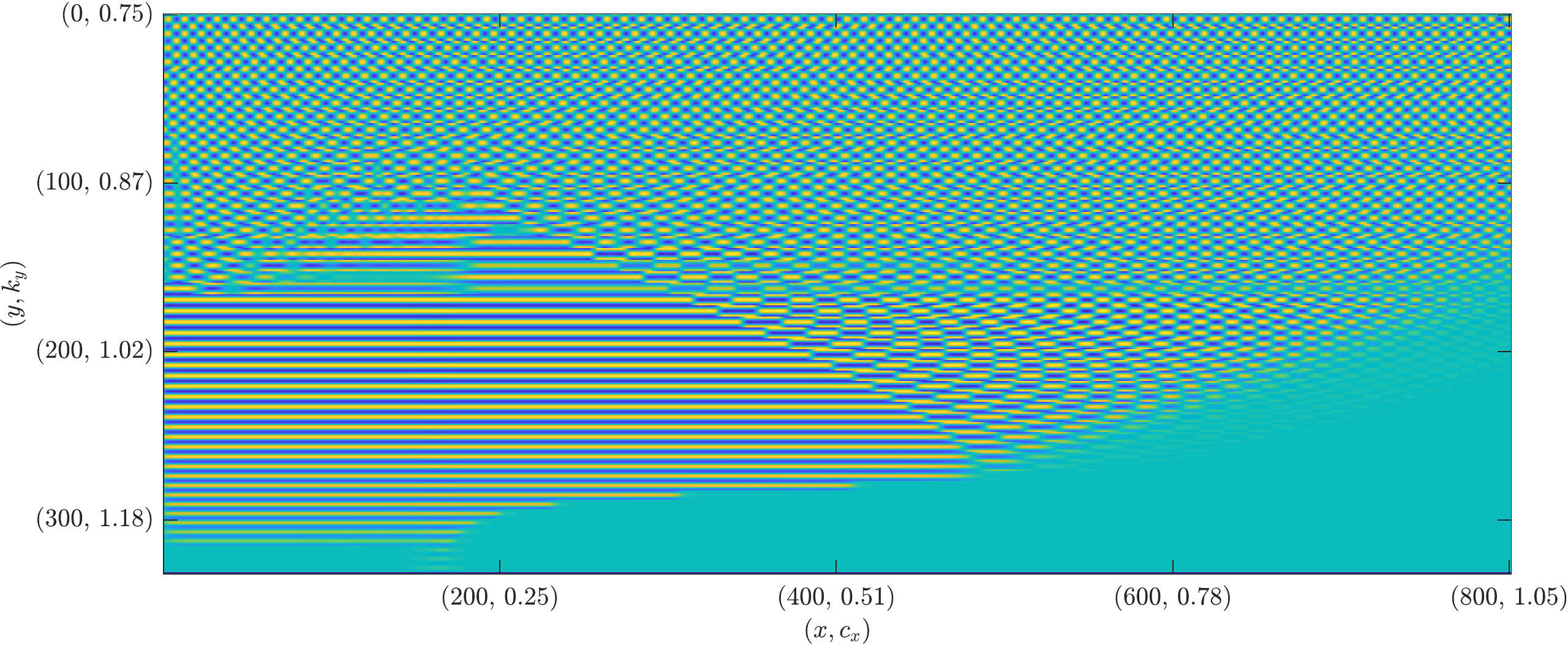}
\caption{Parameter landscape of stripe formation: Simulations with $k_y$ varying from  $0.75$ to $1.24$ in increments of $0.01$ in a strip of length $300\pi$ (only part shown), with quenching step $\rho=-\mu\sign(x-\xi(t))$ and exponential speed  $\xi(t)=9\cdot 10^{-5}(\rme^{0.13 t}-1)$. Dynamics are frozen at $x<\xi(t)-10$, such that the pattern observed at position $x$ encodes the possibly transient pattern formed at the corresponding speed $c_x\sim \xi'(t)$. Top picture for generic initial conditions, bottom pictures for solutions with even-odd symmetry $u(x,y)=u(x,-y)=-u(x,\frac{\pi}{k_y}-y)$; compare Figure \ref{f:end} for theoretical predictions. }\label{f:2a}
\end{figure}

\paragraph{Summary of phenomena} Figure \ref{f:2a} presents much of the phenomenology in a sequence of simulation results for $\mu = 1/4$. We characterize patterns created at the quenching interface depending on quenching rate $c_x$ and the lateral wavenumber $k_y$. The Figure contains 50 simulations with $k_y\in [0.75,1.24]$, increasing in increments of $0.01$, stacked vertically. In each simulation, $c_x$ increased exponentially with very small rate and the dynamics are frozen at a distance 10 behind the quenching line such that the pattern gives a good illustration of the dynamics in the $(k_y,c_x)$-parameter plane. We observe striped patterns with different orientations relative to the quenching interface,
\begin{itemize}
 \item \emph{perpendicular stripes} (horizontal in the picture),
 \item \emph{oblique stripes},
 \item \emph{parallel stripes} (vertical in the picture),
\end{itemize}
and striped patterns with defects, oblique and perpendicular, respectively,
\begin{itemize}
 \item \emph{zigzagging stripes},
 \item \emph{spotted stripes}.
\end{itemize}

Perpendicular and oblique stripes dominate in this parameter regime; parallel stripes take over for large wavenumbers and large speeds. One notices roughly two qualitatively different transitions,
\begin{itemize}
 \item \emph{oblique detachment}: oblique stripes form for small speeds $c_x<c_x^\mathrm{osn}(k_y)$, $0.85\lesssim c_x\lesssim 1$, until zigzagging and, for $k_y\gtrsim 0.9$, perpendicular stripes take over;
 \item \emph{perpendicular detachment}: for $0.85\lesssim k_y\lesssim 1.17$, perpendicular stripes detach at $c_x^\mathrm{psn}(k_y)$ and give rise to  either oblique or spotted stripes. 
\end{itemize}
Details of these transitions are in fact more subtle and the parameter regions visible in this figure only qualitatively reflect the more accurate analysis presented below. Figure \ref{f:2a} also contains the equivalent effective bifurcation diagram when dynamics are restricted to the subspace of even solutions that are also odd with respect to reflection at a line $y=\pi/(2k_y)$. This \emph{even-odd} subspace also corresponds to solutions satisfying Neumann boundary conditions $u_{yyy}=u_y=0$  at $y=0$ and Dirichlet boundary conditions $u_{yy}=u=0$ at $y=\pi/(2k_y)$. It eliminates both zigzagging stripes and parallel stripes and therefore presents a somewhat clearer picture of the perpendicular stripe existence region. 

Our approach here is, as a consequence, threefold. We first analyze
oblique detachment, then perpendicular detachment, and finally present a more comprehensive view of all dynamics. We analyze the detachment bifurcation in universal modulation equations, expecting that those results are more widely applicable.

\paragraph{Oblique detachment and Cahn-Hilliard on $x<0$}
We show that phenomena for small $c_x\gtrsim 0$ and $1-k_y=\eps\sim 0$ can be captured by a Cahn-Hilliard equation
\begin{equation}\label{e:ch}
\psi_t=-(\psi_{xx}+\eps\psi-\psi^3)_{xx}+c_x \psi_x,\ x<0,\qquad \left.\psi=\psi_{xx}=0\right|_{x=0},
\end{equation}
after suitable scalings, using either amplitude equation formalism as in \cite[\S8.3]{hoyle} or spatial dynamics \cite{hs}. We construct heteroclinic and homoclinic  orbits for $c_x=0$ that correspond to slanted stripes, compatible with the parameter jump (boundary condition) at $x=0$ and analyze the singular perturbation that yields slanted stripes for $c_x\gtrsim 0$. Heteroclinic and homoclinic orbits can be understood as parts of grain boundaries, constructed in \cite{hs}, and correspond in this sense quite literally to kinks or wrinkles in stripes. Our analysis predicts $\frac{\rmd k_x}{\rmd c_x}$ and agrees well with numerical results that we present later.
We continue the heteroclinic profiles numerically until they disappear in a saddle-node bifurcation, the \emph{oblique detachment}. The unstable branch corresponds to an oblique stripe that contains a kink. The saddle-node occurs when this kink detaches from the boundary, which leads to periodic kink shedding, the creation of zigzag patterns in the wake of the quenching line. The saddle-node bifurcation is of independent theoretical interest due to the presence of essential spectrum, and we discuss some interesting technical questions and phenomena in this context. We also exhibit a transition where the nature of the bifurcation changes to a hyperbolic homoclinic orbit, that causes  changed asymptotics and coexistence between oblique stripes and oscillating stripe angles in large finite domains.
We finally study detachment of the kink-shedding process in the Cahn-Hilliard equation, which corresponds to detachment of zigzag oscillations in Swift-Hohenberg,  near a critical speed resulting in the creation of perpendicular stripes at the quenching line. 

\paragraph{Perpendicular detachment and Newell-Whitehead-Segel}
We study the dynamics of stripes for moderate speeds in amplitude equations. One observes  yet another saddle-node bifurcation corresponding to the \emph{perpendicular detachment}, and an accompanying birth of a limit cycle. The bifurcation is accompanied by pitchfork bifurcations and several transitions from convective to absolute instabilities. For yet larger speeds, oblique stripes and eventually stripes of all orientations detach. 

\paragraph{The moduli space}
We present computational results in the Swift-Hohenberg equation that capture oblique and perpendicular stripes, using a Newton method and far-field core decomposition, in \S\ref{s:6}. This allows us to systematically track patterns through the saddle-node bifurcations and detect other instabilities. The results can be summarized in a surface, the \emph{moduli space}, in the $(k_x,k_y,c_x)$-space. The surface is surprisingly complex. Many of the phenomena discussed here are reflected in the geometry of this surface; see Figure \ref{f:end}.
 
\paragraph{Universality and similar phenomena in the literature}
Directional quenching in Turing-type systems was studied qualitatively in the context of the CIMA reaction in \cite{cima}, with qualitatively similar observations of transitions between parallel, oblique and perpendicular orientations. In the context of the Cahn-Hilliard equation as a model for phase separation, similar transitions have been studied in the literature. The most striking similarity can be found in a bifurcation study of a Langmuir-Blodgett transfer model \cite{koepf1}. Without our emphasis on a problem posed in an infinite domain, the authors observe a primary branch for small $c_x$ ($V$ in their notation), which destabilizes in a saddle-node bifurcation and then continues a snaking curve, different from our situation. Similar to our context, the authors do see a branch of periodic orbits limiting on the primary branch in a global homoclinic bifurcation (as in our situation, not always at the saddle-node but sometimes on a homoclinic to a hyperbolic equilibrium), which disappears in a steep Hopf bifurcation, that in the limit of large domains is caused by a detachment of kink-formation. A different scenario occurs in simple parameter triggers for Cahn-Hilliard \cite{krekhov}, where mass conservation forces the appearance of periodic orbits for arbitrarily small speeds. A more comprehensive numerical study based on direct simulations can be found in \cite{foard}. Also, the kink-shedding process, organized by bifurcations in Cahn-Hilliard and Newell-Whitehead-Segel equations, appears to be an organizing feature behind a number of phenomena also in reaction-diffusion processes; see for instance \cite{baer,gaffney}. 

In a more narrow sense, we expect that the first part of our discussion is a universal description of growth in systems with zigzag-instabilities, since those can universally be reduced to Cahn-Hilliard type phase-diffusion problems. The second part of our analysis relies on amplitude equations and should hold quite generally near instabilities in isotropic systems that select a finite wavenumber, and in the absence of quadratic interaction terms that would favor formation of spots over stripes. 

\paragraph{Analytical contributions}
 
Our main analytical contribution is the analysis of a heteroclinic bifurcation at $c_x=0$ which allows us to establish existence of solutions growing slanted stripes, and predict leading-order asymptotics for the angle. We also make extensively use of predictions for wavenumbers based on pinched double roots and absolute spectra \cite{holz,ssabs,rss,gs1}.

\paragraph{Outline} 
We derive the Cahn-Hilliard approximation in \S\ref{s:2} and present a heteroclinic bifurcation analysis for the case $c_x\gtrsim 0$ in \S\ref{s:3}. Numerical continuation of the heteroclinic orbit, up to and including the saddle-node bifurcation on a limit cycle where oblique stripes detach are presented in \S\ref{s:4}. Perpendicular detachment and instabilities of perpendicular stripes are discussed within amplitude equations in \S\ref{s:5}. Finally, \S\ref{s:6} shows computational results for the Swift-Hohenberg equation based on continuation and farfield-core decompositions. We conclude with a discussion of related work and some open problems.

\section{Zigzag instabilities  and the Cahn-Hilliard approximation}\label{s:2}

We first give some background on zigzag  instabilities of striped patterns in \S\ref{s:2.1}, which cause both bending and wrinkling of stripes in Figure \ref{f:2}. We briefly mention amplitude formalism and the derivation of the Cahn-Hilliard equation in \S\ref{s:2.2}, and discuss a more rigorous reduction procedure that allows one to find kinks in striped phases in \S\ref{s:2.3}. We use this formalism to motivate boundary conditions for the Cahn-Hilliard equation on a half-line in the case of directional quenching in \S\ref{s:2.4}.

\subsection{Stripes and zigzag instabilities}\label{s:2.1}
Striped solutions in the Swift-Hohenberg equation \eqref{e:sh} can be found as even solutions to the boundary-value problem 
\[
-\left(k^2\partial_{\xi\xi}+1\right)^2 u_\mathrm{p}+\mu u_\mathrm{p} - u_\mathrm{p}^3=0, \qquad u_\mathrm{p}(\xi+2\pi)=u_\mathrm{p}(\xi),
\]
such that, for fixed $\mu>0$, $u_\mathrm{p}(k x;k)$ solves \eqref{e:sh}. In fact, one can show that for small $\mu>0$ there exists a family of such solutions, parameterized by $k\sim 1$. Since \eqref{e:sh} is isotropic, we also find the associated rotated solutions 
\[
u_\mathrm{p}(k_x x + k_y y;k),\qquad k=\sqrt{k_x^2+k_y^2}.
\]
Beyond existence, one would next ask for stability of these solutions, studying the linearized operator
\[
\mathcal{L}(k_x)u=-\left(k_x^2\partial_{\xi\xi}+\partial_{yy}+1\right)^2 u+\mu u - 3 u_\mathrm{p}^2(\xi;k_x) u,\qquad y,\xi\in\R.
\]
Floquet-Bloch theory conjugates this operator to the family of operators 
\begin{equation}\label{e:fb}
\hat{\mathcal{L}}(k_x;\sigma_x,\sigma_y)u=-\left(k_x^2(\partial_{\xi}+\rmi\sigma_x)^2-\sigma_y^2+1\right)^2 u+\mu u - 3 u_\mathrm{p}^2(\xi;k_x)  u,\quad u(\xi)=u(\xi+2\pi),
\end{equation}
such that the spectrum of $\mathcal{L}$ is the union of the spectra of $\hat{\mathcal{L}}(k_x;\sigma_x,\sigma_y)$, $0\leq \sigma_x<1$, $\sigma_y\in\R$. Since the spectrum of  $\hat{\mathcal{L}}(k_x;\sigma_x,\sigma_y)$ consists of isolated, real eigenvalues of finite multiplicity, one can use regular perturbation theory to calculate expansions of eigenvalues near $\mu=\sigma_x=\sigma_y=0$. One finds that the spectrum is stable with the possible exception of a branch of eigenvalues
\[
\lambda(\sigma_x,\sigma_y;k_x)=-d_{\parallel}(k_x)\sigma_x^2-d_\perp(k_y) \sigma_y^2+\rmO(4).
\]
In particular, the spectrum is stable, $\Re\lambda\leq 0$, when effective diffusivities are positive $d_\parallel,d_\perp>0$, a region in $(k_x,\mu)$-space often referred to as the Busse balloon. The boundaries of this region are, for small $\mu$, given by the Eckhaus boundary $k_\mathrm{eck}(\mu)$ and the zigzag boundary $k_\mathrm{zz}(\mu)$, where $d_\parallel$ is negative for $k_x>k_\mathrm{eck}$ and $d_\perp$ is negative for $k_x<k_\mathrm{zz}$. It turns out that stripes with $k=k_\mathrm{zz}$ possess minimal energy density and are therefore preferred in many situations. 

Our focus here will be on systems $(x,y)\in\R\times (\R/(L_y \Z))$  where the lateral period $L_y$ is close to the critical zigzag period $L_y\sim 2\pi/k_\mathrm{zz}$. Ignoring the parameter jump at $x=0$, we see that stripes with $k_x=0$, $k_y=2\pi/L_y$ are stationary solutions in such a strip, stable only when $k_y>k_\mathrm{zz}$. In fact, choosing the lateral period such that $k_y<k_\mathrm{zz}$, we can find rotated stripes $u_\mathrm{p}(\kappa_x x + \kappa_y y;\kappa)$ with $\kappa_y=k_y$ and $\kappa_x=\sqrt{k_\mathrm{zz}^2-k_y^2}$, such that the wavelength of this rotated pattern is precisely $k_\mathrm{zz}$, thus minimizing the energy. This instability mechanism, often referred to as the zigzag instability is at the heart of much of the phenomenology in this paper. We refer to \cite{hoyle} for background on the discussion of instabilities and \cite{mielke} for a more technical discussion of the Floquet-Bloch analysis mentioned here. 

\subsection{Amplitude and phase diffusion equations}\label{s:2.2}
Striped patterns near a given orientation can be described by amplitude equations. We scale $y=k_y\tilde{y}$, and find, dropping tildes
\begin{equation}\label{e:shc}
u_t=-(\partial_{xx}+k_y^2\partial_{yy} +1)^2u+\mu u-u^3. 
\end{equation}
Substituting an Ansatz $u(t,x,y)=\rme^{\rmi y} A( t, x)+\mathrm{c.c.}$, 
assuming that $\mu$ is small and $t,x$ are slowly varying, and collecting leading orders in $\mu$, gives the Newell-Whitehead-Segel amplitude equation 
\begin{equation}\label{e:NWS}
A_t=-(\partial_{xx}+2\eps-\eps^2)^2A+\mu A -3A|A|^2,\qquad \mbox{where }\eps = 1-k_y. 
\end{equation}
Writing $A=R\rme^{\rmi\Phi}$, separating equations for $R$ and $\Phi$, and relaxing to $R=\sqrt{\mu/3}$ at leading order in $\eps$, we find after a short calculation the Cross-Newell phase-diffusion equation 
\begin{equation}\label{e:CN}
\Phi_t=-c_4\Phi_{xxxx}-c_1\eps \Phi_{xx}+c_3\Phi_x^2\Phi_{xx};\qquad c_4=1,c_1=4, c_3=6.
\end{equation}
see for instance  \cite{hoyle} for the general strategy and \cite{schneidernws} for approximation results and limits of validity. Note that the coefficients in \eqref{e:CN} are obtained from \eqref{e:NWS}, thus leading-order in $\mu$, only. One can more generally derive \eqref{e:CN} directly from the Swift-Hohenberg equation, near the zigzag instability, not necessarily at small $\mu$.  We computed coefficients $c_{1,3,4}$ numerically in this way, thus not using an asymptotic expression in $\mu$,  from the expansion of the dispersion relation. We found however very small deviations from \eqref{e:CN} of order $10^{-4}$  for $\mu=0.25$.

\subsection{Spatial dynamics, knees, and more}\label{s:2.3}
A different approach \cite{hs} focuses on stationary patterns or traveling waves of \eqref{e:shc} that remain close to horizontal striped patterns at all locations $x\in\R$, $u=u_\mathrm{p}(k_\mathrm{zz}y-\phi(x);k_\mathrm{zz}+w(x,y))$ in 
\begin{equation}\label{e:shctw}
-c_x u_x-c_y u_y=-(\partial_{xx}+k_y^2 \partial_{yy}+1)^2u + \mu u - u^3,\qquad u(x,y)=u(x,y+2\pi),
\end{equation}
with $k_y=k_\mathrm{zz}+\eps$. One casts \eqref{e:shctw} as a first-order differential equation in $x$, requiring that $w$ be orthogonal to $u_\mathrm{p}'$ and studies the resulting equations for $\phi$ and $w$ from a dynamical systems point of view. One finds a family of equilibria, that is, $x$-independent solutions, $w\equiv 0$, $\phi\equiv const$, at $\eps=0$. 
Linearizing at these equilibria gives a length-4 Jordan block at the origin, with the rest of the spectrum being bounded away form the origin. A center-manifold reduction can thus be carried out, and one finds a fourth-order differential equation for $\phi$, given at leading order through
\begin{align}
\phi_x&=\psi,\notag\\
\psi_x&=v\notag\\
v_x&=-4\eps\psi+2\psi^3+m\notag\\
m_x&=c_x\psi+c_y,\label{e:chtw}
\end{align}
 where $\eps=k_y-k_\mathrm{zz}$. Scaling 
  \begin{equation}\label{e:scalch}\psi=(2\eps)^{1/2}\tilde{\psi},\ v=\sqrt{8}\eps\tilde{v},\ m=\sqrt{32}\eps^{3/2}\tilde{m},\  \partial_x=2\eps^{1/2}\partial_{\tilde{x}},\ c_x=8\eps^{3/2}\tilde{c}_x, \  c_y={128}\sqrt\eps^{2}\tilde{c}_y, 
 \end{equation}
  eliminates $\eps$-dependence and gives the 
  the traveling-wave equation corresponding to the Cross-Newell equation \eqref{e:CN}, substituting an  Ansatz 
 $\Phi=\phi(x-c_x t, t) + c_y t$. 
\begin{figure} \centering
\includegraphics[width=1.\linewidth]{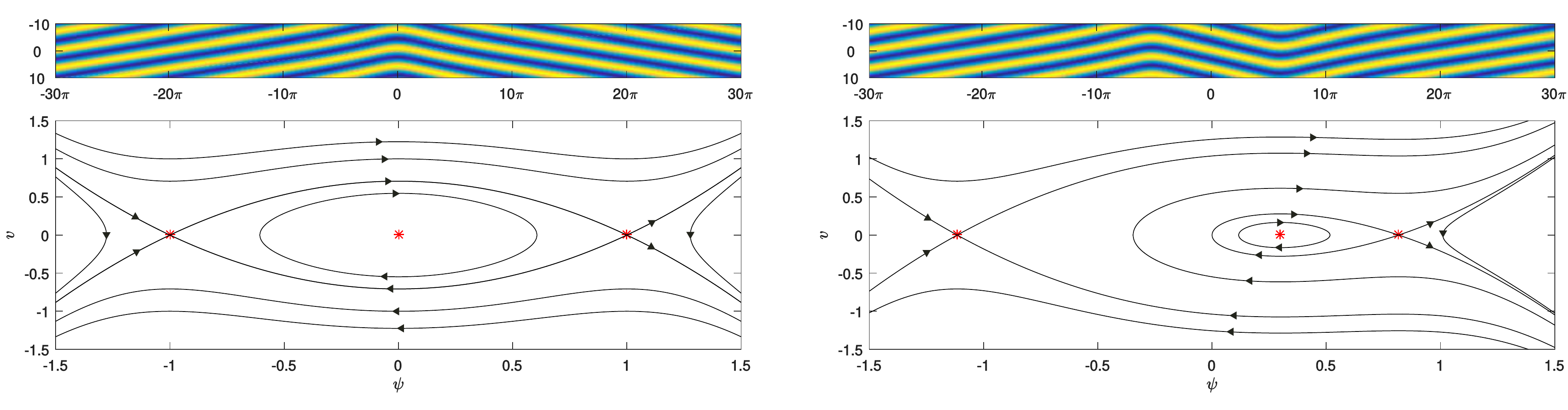}
\caption{Picture of knee solution (heteroclinic) and zigzag (homoclinic) in Swift-Hohenberg (top); phase portraits with corresponding heteroclinic and homoclinic orbits (bottom).}\label{f:knee}
\end{figure}

In the specific case of the Swift-Hohenberg equation, the stationary equations \eqref{e:shctw} for $c_x=c_y=0$ possess a Hamiltonian structure. Indeed, the system can be obtained as Euler-Lagrange equation to a variational problem with translation-invariance in $x$. Interpreting the energy as the action functional and the energy density as the Lagrangian, one then finds the Hamiltonian structure. In more detail, the energy interpreted as  action functional is 
\[
\mathcal{E}[u]=\int_{x,y} \left(\frac{1}{2}\left(\partial_{xx} u + k_y^2 \partial_{yy} u +u\right)^2 -\frac{1}{2} \eps^2 u^2 + \frac{1}{4}u^4 \right)\rmd x \rmd y.
 \]
We write \eqref{e:shctw} as a first order equation for $\underline{u}=(u,u_1,v,v_1)^T$ in the form
\begin{align}
u_x&=u_1\notag\\
u_{1,x}&=v-k_y^2 \partial_{yy} u -u\notag\\
v_x&=v_1\notag\\
v_{1,x}&=-k_y^2 \partial_{yy} v -v+ \mu u - u^3,\label{e:1st}
\end{align}
and define the Hamiltonian as 
\[
\mathcal{H}[\underline{u}]=\int_y h(\underline{u})\rmd y,\quad h(\underline{u})= -\frac{1}{2} v^2+u_1v_1 + v(k_y^2u_{yy}+u)-\frac{1}{2}u^2+\frac{1}{4}u^4,
\]
and the symplectic structure through  \cite{lloydscheel}
\[
\omega(\underline{u},\underline{\tilde{u}})=\int_y \underline{u}\cdot(\mathcal{J}\underline{\tilde{u}})\rmd y,\qquad \mathcal{J}=\begin{pmatrix}
0&0&0&1\\ 
0&0&-1&0\\
0&1&0&0\\
-1&0&0&0
\end{pmatrix}.
\]
The translation symmetry in $y$ corresponds to the conserved quantity 
\[
\mathcal{S}[\underline{u}]=\int_y s(\underline{u}),\quad s(\underline{u})=u (v_1)_y+v (u_1)_y,\quad \mathcal{J}\nabla_{L^2}S[\underline{u}]=\partial_y\underline{u},
\]
which we shall refer to as the momentum.

The tangent space to the center manifold is spanned by \cite{hs}
\begin{align*}
e_1&=
\left(\begin{array}{c}
u_\mathrm{p}'(y)\\0\\-(1+k_y^2\partial_{yy}u_\mathrm{p}')\\0
\end{array}\right),\ e_2=
\left(\begin{array}{c}
0\\u_\mathrm{p}'(y)\\0\\(1+k_y^2\partial_{yy})(1+k_y^2\partial_{yy})u_\mathrm{p}'
\end{array}\right),\\
e_3&=
\left(\begin{array}{c}
-\frac{1}{2}v_2(y)\\0\\u_\mathrm{p}'(y)-\frac{1}{2}(1+k_y^2\partial_{yy})v_2(y)\\0
\end{array}\right),\ e_4=
\left(\begin{array}{c}
0\\-\frac{1}{2}v_2(y)\\0\\ u_\mathrm{p}'(y)-\frac{1}{2}(1+k_y^2\partial_{yy})v_2(y) 
\end{array}\right).
\end{align*}
Here, $v_2(y)$ is the second derivative of the eigenvector for the Floquet-Bloch operator \eqref{e:fb} in $\sigma_y$ at $\sigma_x=0$ \cite{hs}.

One readily finds that the reduced flow \eqref{e:chtw} is obtained in this basis through $\underline{u}_\mathrm{c}=\phi e_1+\psi e_2+v e_3+m e_4$. The symplectic structure, reduced Hamiltonian, and reduced momentum are, at leading order, in coordinates $U=(\phi,\psi,v,m)^T$,
\[
\omega(U,\tilde{U})=U\cdot (J \tilde{U}), \mathcal{J}=\begin{pmatrix}
0&0&0&1\\ 
0&0&-1&0\\
0&1&0&0\\
-1&0&0&0
\end{pmatrix},\ H(U)=m \psi-\frac{1}{2}v^2-2\eps^2 \psi^2+\frac{1}{2}\psi^4, S(U)=m.
\]
We note that heteroclinic orbits connect periodic orbits that are marginally stable with respect to the zigzag instablitity \cite{lloydscheel}, due to conservation of $\mathcal{S}$. One easily verifies the $\mathcal{S}$ is also conserved for spatially inhomogeneous $\rho=\rho(x)$ that do not break the associated translation symmetry in $y$. As a consequence, the quenched system allows for oblique stripes with wavenumber $k_\mathrm{zz}$, only, when $c_x=0$.

\subsection{Spatial dynamics and effective boundary conditions}\label{s:2.4}
The considerations in the previous section provide a local description of solutions in a vicinity of the primary periodic stripe, for $x<0$, only. The picture can be complemented by a description of dynamics in $x>0$, where the origin $u=0$ is a hyperbolic equilibrium. The following discussion is kept at a somewhat informal level as it is merely meant to motivate effective boundary conditions. We consider \eqref{e:1st} now with a parameter step, replacing the constant coefficient $\mu$ by  $\rr=-\mu\sign,(x)$. 

We first define the stable manifold  $W^\mathrm{s}_+$ in $x>0$  where $\rr=-\mu<0$, as the set of initial conditions  $(u,u_x,u_{xx},u_{xxx})(y)$ at $x=0$ that give rise to solutions converging to the origin as $x\to+\infty$. Next, we define the center-unstable manifold  $W^\mathrm{cu}_-$, in $x<0$  where $\rr=-\mu<0$, as the set of initial conditions at $x=0$ that give rise to solutions  that converge to the center manifold $W^\mathrm{c}_-$ for the dynamics in $x<0$ near the stripes, as $x\to -\infty$. Solutions of interest to us lie in the intersection of $W^\mathrm{s}_+\cap W^\mathrm{cu}_-$. 

Fredholm theory shows that a transverse intersection of these two manifolds would consist of a two-dimen\-sional submanifold of $W^\mathrm{cu}_-$. Since $W^\mathrm{cu}_-$ is foliated over the 4-dimensional center manifold $W^\mathrm{c}$, we may assume that the submanifold is transverse to this foliation and then project this two-dimensional submanifold along the smooth foliation onto the center manifold $W^\mathrm{c}$, where it gives rise to a two-dimensional submanifold  $\mathcal{B}$ of $W^\mathrm{c}$. By construction, initial conditions on this two-dimensional submanifold that give rise to bounded solutions on the center manifold as $x\to -\infty$ correspond to bounded solutions on $x\in\R$, after lifting to the corresponding intersection point in the unstable foliation. 
    
The construction outlined above yields, under some transversality assumptions, the existence of \emph{effective boundary conditions}, a two-dimensional submanifold  $\mathcal{B}\subset W^\mathrm{c}$. By translation invariance with respect to the shift in $y$, the center manifold $W^\mathrm{c}$, the local flow on $W^\mathrm{c}$, and the effective boundary condition $\mathcal{B}$ are invariant under this translation, which is simply given through the additive action $\phi\mapsto \phi+\varphi$ on the circle. As a consequence, $\mathcal{B}=\{(\psi,v,m)\in B\}$ for some one-dimensional manifold $B$ which we parameterize as $(\psi_B(\sigma), v_B(\sigma),m_B(\sigma))$, $\sigma\sim 0$, with $\psi_B(0)= v_B(0)=m_B(0)=0$.

Within the center-manifold, the scaling  $\tilde{x}=\sqrt{\eps}x$, $\tilde{\psi}=\sqrt{\eps}\psi$,  $\tilde{v}={\eps}v$, $\tilde{m}={\eps}^{3/2}m$ that reduces to the $\eps$-independent Cahn-Hilliard equation \eqref{e:scalch}, eliminates the parameter $\eps$ at leading order and gives the Cahn-Hilliard steady-state equation. With this scaling, the boundary curve $\mathcal{B}$ is transformed to 
$B_\eps\sim (0,0,\sigma)$ provided that $m_B'(0)\neq 0$. In this sense, we expect a typical clamped boundary condition 
\begin{equation}\label{e:bcclamped}
\psi=\psi_x=0 \mbox{   at } x=0.
\end{equation}
Of course, these boundary conditions would be accurate only at leading order in $\eps$.

In the specific case of the Swift-Hohenberg equation with a parameter step, the boundary manifold is not ``generic'' in the sense that the tangent space at the origin is given by $\psi=\psi_{xx}=0$. This non-genericity is caused by the Hamiltonian structure of the reduced equation, or, more specifically, by the conservation of momentum. In fact, the equation with parameter jump in $x$ possesses the $y$-translation symmetry such that the momentum $S$ is conserved in $x$. Therefore, ${S}$ evaluated on the effective boundary conditions coincides with $S$ evaluated at the origin, $x=+\infty$. As a consequence, the boundary manifold $B_\eps$ is contained in $\{m=0\}$. A generic curve through the origin in the $(\psi,v)$-plane will, after scaling, reduce to the line $\psi=0$, which together with $m=0$ gives the Dirichlet boundary conditions
\begin{equation}\label{e:bcdirichlet}
\psi=\psi_{xx}=0. \mbox{   at } x=0.
\end{equation}
We would expect small non-variational effects to yield boundary conditions that interpolate between Dirichlet and clamped, and therefore also study a straight interpolation, 
\begin{equation}\label{e:bchomotopy}
\psi=\tau \psi_x+(1-\tau)\psi_{xx}=0. \mbox{   at } x=0,
\end{equation}
for $0\leq \tau\leq 1$.

Finally, we shall also consider the time-dependent, scaled  version,
\begin{equation}\label{e:cht}
\psi_t=-(\psi_{xx}+\psi-\psi^3)_{xx}+c_x\psi_x,
\end{equation}
together with the boundary conditions \eqref{e:bcclamped} and \eqref{e:bcdirichlet}. We emphasize that for any of the choices of boundary conditions, mass $\int \psi$ is not conserved at the boundary. In particular, solutions with $\psi=0$ at the boundary and $\psi\to \eta\neq 0$ for $x\to-\infty$ are possible.

\section{Slow growth: A singular heteroclinic perturbation problem}\label{s:3}

This sections contains our main analytical results. We first analyze the stationary equation $c_x=0$ in \S\ref{s:3.1} with various boundary conditions. We then set up a perturbation analysis and establish existence of oblique stripe creation for $c_x\gtrsim 0$ in \S\ref{s:3.2}. We relegate some more technical aspects of the analysis to \S\ref{s:3.3} which the reader may skip at first reading. Throughout this section, we focus on the traveling-wave equation for the Cahn-Hilliard type problem \eqref{e:chtw}. The construction here is robust and can therefore easily accommodate error terms in boundary conditions or higher-order terms in the equation.

\subsection{Oblique stripes at zero speed}\label{s:3.1}
We study \eqref{e:chtw}, with rescaling  \eqref{e:scalch}, and omitting the trivial equation for $\phi$, 
\begin{align}
\psi_x&=v\notag\\
v_x&=-\psi+\psi^3+m\notag\\
m_x&=c_x\psi+c_y.\label{e:chtw0}
\end{align}
First, set $c_x=c_y=0$. In this case $m$ is a parameter, and the resulting system in the planes $m\equiv const$ can be readily analyzed as a nonlinear pendulum equation; see Figure \ref{f:phasespace}. Clamped boundary conditions correspond to a shooting problem from the line $\psi=v=0$ to backward spatial time $x<0$. Dirichlet boundary conditions correspond to a similar shooting problem, now from the line $\psi=m=0$.

\begin{figure}[h] \centering
\includegraphics[width=0.7\textwidth]{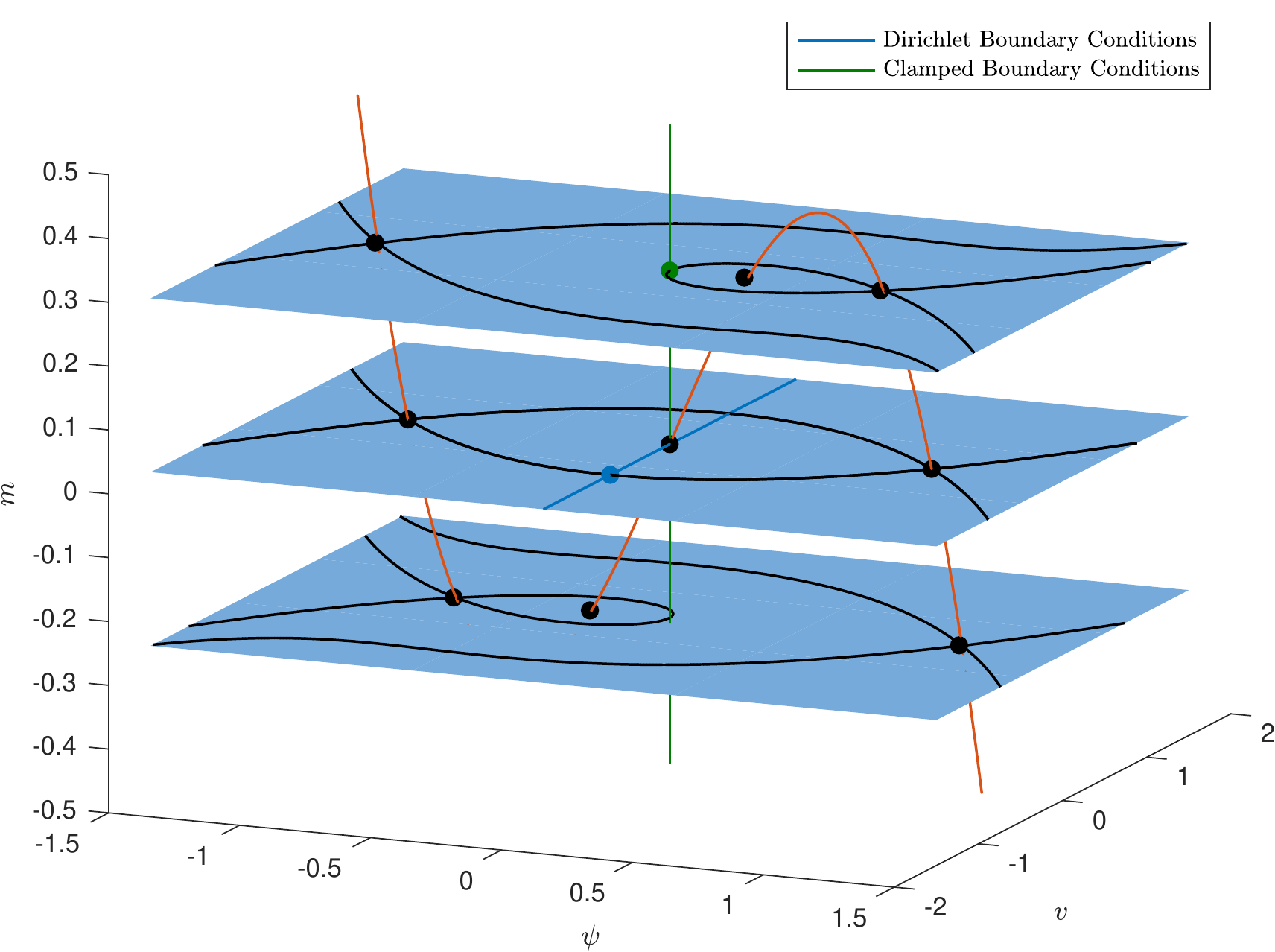}
\caption{Phase portrait for \eqref{e:chtw0} at $c_x=c_y=0$, including straight lines defined by boundary conditions at $x=0$ and the heteroclinic and homoclinic intersections with the boundary conditions. }
\label{f:phasespace}
\end{figure}

We immediately find that for Dirichlet boundary conditions, $\psi_\mathrm{d}(x)=\pm\tanh(x/\sqrt{2})$ is a bounded solution. Of course, this solution is simply half of a grain boundary solution described in the previous section. For clamped boundary conditions, we find values $\eta_* = \sqrt{2/3},$ $m_* = \eta_* - \eta_*^3$ and solutions $\psi_\mathrm{cl}(x)\to \eta_*$ for $x\to -\infty$, which are simply half of the type of step solution described in the previous section; see Figure \ref{f:knee}. For later reference, we collect explicit expressions in the case of clamped boundary conditions 
\begin{equation}\label{e:psi}
\psi_\mathrm{cl}'= (\psi_\mathrm{cl} - \eta_*) \sqrt{\frac{1}{2} \psi_\mathrm{cl} (\psi_\mathrm{cl} + 2 \eta_*)},
\end{equation}
We note that the conserved quantity $m$ leads to degenerate dynamics in $\R^3$, and perturbations with $c_x$ and $c_y$ that break this degeneracy should be viewed as singular perturbations in the sense of \cite{fen}. While one could go ahead and study this singular perturbation problem geometrically following the ideas there, we choose a somewhat more direct and possibly more self-contained approach using farfield-core decompositions.  

\subsection{Oblique stripes for $c_x\gtrsim 0$}\label{s:3.2}

Our main analytical result is as follows. 

\begin{theorem}\label{t:1}Consider \eqref{e:chtw0} with either clamped \eqref{e:bcclamped} or Dirichlet boundary conditions \eqref{e:bcdirichlet}, near the profiles $\psi_\mathrm{cl/d}$ and near $c_x=c_y=0$. 
For all $c_x$ sufficiently small, there exists a smooth function $c_y=-c_x\cdot \eta_\mathrm{cl/d}(c_x)$ and solutions $\psi_\mathrm{cl/d}(x;c_x)$ such that $\psi_\mathrm{cl/d}(x;c_x)\to \eta_\mathrm{cl/d}(c_x)$ for $x\to -\infty$ and $\psi_\mathrm{cl/d}$ satisfies clamped or Dirichlet boundary conditions at $x=0$, respectively. 

Moreover, $\psi_\mathrm{cl/d}$ and its derivatives depend smoothly on $c_x$ as smooth functions on the real line, locally uniformly. We have the expansions 
\begin{align}
\eta_\mathrm{cl}(c_x)&= \sqrt{\frac{2}{3}} - \left(\sqrt{6} - \sqrt{2} \log(2+\sqrt{3})\right) c_x+\rmO(c_x^2),\\
\eta_\mathrm{d}(c_x)&=1 - \frac{\sqrt{2} \log(2)}{2} c_x+\rmO(c_x^2).
\end{align}
\end{theorem}

We prove Theorem \ref{t:1} in the remainder of this section, up to some more technical aspects that we treat with more care in the next section. The key initial step is decompose the solution into a constant piece near infinity plus an exponentially localized perturbation. Specifically, we introduce a smooth cutoff function $\chi_-$ satisfying $0 \leq \chi_- \leq 1$, $\chi_- \equiv 1$ for sufficiently negative $x$, and $\chi_- \equiv 0$ for $x$ sufficienly close to zero, and look for solutions of the form
\begin{equation}
	\begin{pmatrix}
	\psi \\ v \\ m 
	\end{pmatrix} =
	\begin{pmatrix}
	\hat{\psi} + \chi_- \eta \\
	v \\
	\hat{m} + \chi_- (\eta - \eta^3)
	\end{pmatrix}, \label{e:loc Ansatz}
\end{equation}
where $\eta$ is a constant to be determined, and $\hat{\psi}, v,$ and $\hat{m}$ decay exponentially at $-\infty$. We then write \eqref{e:chtw0} as an equation for $\hat{\psi}, v, \hat{m}$ and $\eta$ in the form 
\begin{equation}
F(\hat{\psi}, v, \hat{m}, \eta; c_x) = 0 \label{e: IFT eqn}
\end{equation}
in an appropriately chosen exponentially weighted function space. The function $F$ is given explicitly by 
\begin{equation}
F(\hat{\psi}, v, \hat{m}, \eta; c_x) = 
\begin{pmatrix}
\hat{\psi}' + \chi_-' \eta - v \\
v' + \hat{\psi} - \hat{m} - \hat{\psi}^3 - 3 \hat{\psi}^2 \chi_- \eta - 3 \hat{\psi} (\chi_- \eta)^2 + (\chi_- - \chi_-^3) \eta^3 \\
\hat{m}' - c_x \hat{\psi} + \chi_-' (\eta-\eta^3) + c_x (1-\chi_-) \eta 
\end{pmatrix}. \label{e: F}
\end{equation}
The Ansatz \eqref{e:loc Ansatz} is chosen so that $(\psi, v, m)$ converges exponentially  to an equilibrium solution of \eqref{e:chtw0} at $-\infty$, hence \eqref{e: IFT eqn} has a trivial solution at $c_x = 0$ given by 
\begin{align}
\eta_0 &= \eta_*^\mathrm{cl/d}, \notag \\
\hat{\psi}_0 &= \psi_\mathrm{cl/d} - \chi_- \eta_*^\mathrm{cl/d}, \notag \\
v_0 &= \psi_\mathrm{cl/d}', \notag \\
\hat{m}_0 &= (1- \chi_-) \left( \eta_*^\mathrm{cl/d} - (\eta_*^\mathrm{cl/d})^3 \right) \label{e: trivial solution}
\end{align}
where $\eta_*^\mathrm{cl} = \sqrt{2/3}$ and $\eta_*^\mathrm{d} = 1$  are the limits of $\psi_\mathrm{cl/d}$ at $-\infty$. 
For ease of notation, let $u_0^\mathrm{cl/d} = (\hat{\psi_0}, v_0, \hat{m}_0, \eta_0; 0)$ denote this trivial solution. We will drop the sub-/super-scripts $\mathrm{cl/d}$ when the difference is irrelevant. 

In the following section, we use Fredholm properties to prove that the linearization of $F$ at the trivial solution is invertible in appropriately chosen spaces. The implicit function theorem then guarantees the existence of unique solutions to \eqref{e: IFT eqn} near this trivial solution as well as smooth dependence on $c_x$, proving the first part of Theorem \ref{t:1}. Uniqueness guarantees $\eta_\mathrm{cl/d} (0) = \eta_*^\mathrm{cl/d}$, giving the zeroth order terms in the asymptotics. In computing the coefficient at the next order, we make use of the following fact, obtained in the course of the existence argument.

\noindent \textit{Fact}: Let $\mathcal{L} = \partial_{\hat{\psi}, v, \hat{m}} F (u_0)$ denote the linearization of $F$ at its trivial solution with respect to its first three arguments, and let $\mathcal{L}^*$ be the adjoint of $\mathcal{L}$ with respect to the standard $L^2$ inner product. The kernel of $\mathcal{L}^*$ is one dimensional, spanned by $e_*^\mathrm{cl} = (-\psi_\mathrm{cl}'', \psi_\mathrm{cl}', -\psi_\mathrm{cl})$ for clamped boundary conditions and by $e_*^\mathrm{d} = (0, 0, 1)$ for Dirichlet boundary conditions. 

Once we have existence of solutions and smooth dependence on parameters, differentiating \eqref{e: IFT eqn}  gives, via the chain rule, 
\begin{equation}
\mathcal{L} \left(\partial_{c_x} (\hat{\psi}(c_x), v(c_x), \hat{m}(c_x))\big|_{c_x = 0}\right) + \partial_\eta F(u_0) \eta'(0) + \partial_{c_x} F(u_0) = 0. 
\end{equation}
The first term may be eliminated by projecting onto the kernel of $\mathcal{L}^*$, which is orthogonal to the range of $\mathcal{L}$. We thereby find explicit expressions for $\eta'(0)$ in terms of projections onto the adjoint kernel: 
\begin{equation}
\eta'_\mathrm{cl/d} (0) = -\frac{\langle \partial_{c_x} F(u_0^\mathrm{cl/d}), e_*^\mathrm{cl/d} \rangle}{\langle \partial_\eta F(u_0^\mathrm{cl/d}), e_*^\mathrm{cl/d} \rangle}. \label{e: d eta 0}
\end{equation}
These derivatives are given explicitly, after some simplification, by 
\begin{equation}
\partial_\eta F(u_0) = \begin{pmatrix}
	\chi_-' \\
	3 \chi_- \left((\eta_*^\mathrm{cl/d})^2-(\psi_*^\mathrm{cl/d})^2 \right) \\
	\chi_-'(1- 3 (\eta_*^\mathrm{cl/d})^2)
    \end{pmatrix}
\text{ and }
\partial_{c_x} F(u_0) = \begin{pmatrix}
0 \\
0 \\
\eta_*^\mathrm{cl/d} - \psi_\mathrm{cl/d}
\end{pmatrix}.
\end{equation}
For clamped boundary conditions, we find 
\begin{align}
\langle \partial_{c_x} F(u_0^\mathrm{cl}), e_*^\mathrm{cl} \rangle &= \langle \partial_{c_x} F(u_0^\mathrm{cl}), (-\psi_\mathrm{cl}'', \psi_\mathrm{cl}', -\psi_\mathrm{cl}) \rangle \notag \\
	&= \int_{-\infty}^0 (\eta_*^\mathrm{cl} - \psi_\mathrm{cl}) (-\psi_\mathrm{cl}) \, \rmd x \notag 
	= \int_{\eta_*^\mathrm{cl}}^0 (\psi_\mathrm{cl} - \eta_*^\mathrm{cl}) \psi_\mathrm{cl} \frac{d \psi_\mathrm{cl}}{\psi_\mathrm{cl}'} \notag \\
	&= \int_{\eta_*^\mathrm{cl}}^0 \frac{(\psi_\mathrm{cl} - \eta_*^\mathrm{cl}) \psi_\mathrm{cl}}{(\psi_\mathrm{cl}-\eta_*^\mathrm{cl}) \sqrt{\frac{1}{2} \psi_\mathrm{cl} (\psi_\mathrm{cl}+2\eta_*^\mathrm{cl})}} d \psi_\mathrm{cl} \notag 
	= \int_{\eta_*^\mathrm{cl}}^0 \frac{\sqrt{2 \psi_\mathrm{cl}}}{\sqrt{\psi_\mathrm{cl}+2\eta_\mathrm{cl}}} \, d \psi_\mathrm{cl} \notag \\
	&= \frac{2 \log(2+ \sqrt{3})}{\sqrt{3}} - 2, \label{e: M cx cl}
\end{align}
where we have used \eqref{e:psi} to write $\psi_\mathrm{cl}'$ in terms of $\psi_\mathrm{cl}$ so that we may compute the integral with respect to $\psi_\mathrm{cl}$ explicitly.  For the other factor, we find, integrating by parts, 
\begin{align}
\langle \partial_\eta F(u_0^\mathrm{cl}), e_*^\mathrm{cl} \rangle &= \langle \partial_\eta F(u_0^\mathrm{cl}), (- \psi_\mathrm{cl}'', \psi_\mathrm{cl}', -\psi_\mathrm{cl}) \rangle \notag \\
&= \int_{-\infty}^0 -\psi_\mathrm{cl}'' \chi_-' + \psi_\mathrm{cl}' (3 \chi_- ((\eta_*^\mathrm{cl})^2 - \psi_\mathrm{cl}^2)) - \psi_\mathrm{cl}\chi_-'(1-3(\eta_*^\mathrm{cl})^2) \, \rmd x \notag \\
	&= \int_{-\infty}^0 (-\psi_\mathrm{cl}'' - \psi_\mathrm{cl} (1- 3 (\eta_*^\mathrm{cl})^2) - 3 \psi_\mathrm{cl} (\eta_*^\mathrm{cl})^2 +\psi_\mathrm{cl}^3 ) \chi_-' \, \rmd x \notag\\
	&\hspace*{2.5in} + \left[ (3 \psi_\mathrm{cl} \chi_- (\eta_*^\mathrm{cl})^2) - \psi_\mathrm{cl}^3 \chi_- \right]^0_{-\infty} \notag \\
	&= \int_{-\infty}^0 -m_* \chi_-' \, \rmd x + 3 \psi_\mathrm{cl} (\eta_*^\mathrm{cl})^2 - \psi_\mathrm{cl}^3 \big|_{-\infty} \notag \\
	&= m_* - 2 \eta_*^3 
	= \eta_*^\mathrm{cl} - 3 (\eta_*^\mathrm{cl})^3
    = - \sqrt{\frac{2}{3}}. \label{e: M eta cl}
\end{align}
Inserting \eqref{e: M cx cl} and \eqref{e: M eta cl} into \eqref{e: d eta 0} gives the linear coefficient in the asymptotics presented in Theorem \ref{t:1} for clamped boundary conditions. 

For Dirichlet boundary conditions, for which $e_*^\mathrm{d} = (0,0,1)$, we instead find
\begin{align}
\langle \partial_{c_x} F(u_0^d), e_*^\mathrm{d} \rangle  &= \langle \partial_{c_x} F(u_0^d), (0, 0, 1) \rangle \notag \\
&
= \int_{-\infty}^0 \eta_*^d - \psi_d \, \rmd x \notag 
= \int_{-\infty}^0 1 + \tanh \left( \frac{x}{\sqrt{2}} \right) \, \rmd x \notag 
= \sqrt{2} \log (2)
\end{align}
and 
\begin{align}
\langle \partial_\eta F(u_0^d), e_*^\mathrm{d} \rangle &= \langle \partial_\eta F(u_0^d), (0, 0, 1) \rangle \notag \\
&= \int_{-\infty}^0 \chi_-' \left( 1-3(\eta_*^d)^2 \right) \, \rmd x \notag 
= 3 (\eta_*^d)^2 -1 \notag 
= 2,
\end{align}
giving the linear asymptotics for Dirichlet boundary conditions, which completes the proof of Theorem \ref{t:1}, up to the technical details which we present in the next section. 

\subsection{Weighted spaces, Fredholm properties, and the implicit function theorem}\label{s:3.3}
We construct function spaces $X_\mathrm{cl/d}$ and view $F$ as an operator $F: X_\mathrm{cl/d} \times \R^2 \to (L^2_\delta (\R^-))^3$, defined as follows. First,  for $\delta>0$ small, let $H^1_\delta(\R^-)$ denote the weighted Sobolev space of weakly differentiable functions on $x < 0$ with finite $H^1_\delta$ norm, given by 
\begin{equation}
||f(x)||_{H^1_\delta}^2 = ||e^{-\delta x} f(x)||_{H^1}^2 = \int_{-\infty}^0 (|f(x)|^2 + |f'(x)|^2) e^{-2\delta x} \, \rmd x.
\end{equation}
Then, we define
\begin{equation}
X_\mathrm{cl} = \{ (\hat{\psi}, v, \hat{m}) \in (H^1_\delta (\R^-))^3 : \hat{\psi}(0) = v(0) = 0  \},
\end{equation}
and
\begin{equation}
X_\mathrm{d} = \{ (\hat{\psi}, v, \hat{m}) \in (H^1_\delta (\R^-))^3 : \hat{\psi}(0) = \hat{m}(0) = 0 \}
\end{equation}
as the subspaces of $(H^1_\delta (\R^-))^3$ satisfying clamped/Dirichlet boundary conditions, respectively. The linear operator $\mathcal{L}: X_\mathrm{cl/d} \to (L^2_\delta (\R^-))^3$ defined in the previous section is the Frech\'et derivative of $F$ evaluated at the trivial solution defined in \eqref{e: trivial solution}, and is given explicitly by the expected pointwise derivative, 
\begin{equation}
\mathcal{L} \begin{pmatrix}
	\hat{\psi}_1 \\ v_1 \\ \hat{m}_1\end{pmatrix} = \begin{pmatrix}
	\hat{\psi}_1' - v_1 \\
	v_1' + (1-3\psi_\mathrm{cl/d}^2)\hat{\psi}_1 - \hat{m}_1 \\
	\hat{m}_1',
	\end{pmatrix}.
\end{equation}
Our argument is based on the Fredholm properties of $\mathcal{L}$, which determines our choice of the exponential weight $\delta$. 

\begin{lemma}[Fredholm properties]
 For $\delta > 0$ sufficiently small, $\mathcal{L}$ is a Fredholm operator with index -1, trivial kernel, and one dimensional cokernel.  
\end{lemma}
\begin{proof}
If $(\hat{\psi_1}, v_1, \hat{m}_1)$ belongs to the kernel of $\mathcal{L}$, from the third equation we must have $\hat{m_1} \equiv const.$, but constants do not belong to our weighted space, so $\hat{m}_1 \equiv 0$. The first two equations then reduce to 
\begin{equation}
\hat{\psi}_1'' + (1-3\psi_\mathrm{cl/d}^2)\hat{\psi}_1 = 0, \label{e: psi linearization} 
\end{equation}
with $v_1 = \hat{\psi}_1'$. This is the linearization of the equation solved by $\psi_{cl/d}$, about this solution.
Translation invariance of the original equation guarantees that $\hat{\psi}_1 = \psi_\mathrm{cl/d}'$ is a solution to \eqref{e: psi linearization}. Since the Wronskian is constant, a second, linearly independent solution to  \eqref{e: psi linearization} necessarily grows exponentially and we conclude that $\psi_\mathrm{cl/d}$ is the unique  solution that is bounded at $x=-\infty$. Hence, for sufficiently small $\delta > 0$, $\hat{\psi}_1 = \psi_\mathrm{cl/d}'$ is the only solution to \eqref{e: psi linearization} which is contained in our weighted space. In the case of clamped boundary conditions $\hat{\psi}_1'(0) = \psi_\mathrm{cl}''(0) \neq 0$, so the solution does not satisfy the boundary conditions at $x = 0$. For Dirichlet boundary conditions, $\hat{\psi}(0) = \psi_d'(0) \neq 0$, and again the boundary conditions are  not satisfied. Thus, the kernel of $\mathcal{L}$ is trivial. 

We find the cokernel by viewing $\mathcal{L}$ as a closed, densely defined operator on $(L^2(\R^-))^3$ and computing its adjoint $\mathcal{L}^*$ with respect to the standard $L^2$ inner product. The boundary conditions for the adjoint are the orthogonal complement to the boundary conditins for $\mathcal{L}$, i.e. the domain of $\mathcal{L}^*$ is the dense subspace of $(L^2(\R^-))^3$ defined by
\begin{equation}
Y_\mathrm{cl} = \{(\hat{\psi}_1, v_1, \hat{m}_1) \in ( H^1_{-\delta} (\R^-))^3 : \hat{m}_1 (0) = 0 \} \label{e: adj domain cl}
\end{equation}
in the clamped case and 
\begin{equation}
Y_\mathrm{d} = \{(\hat{\psi}_1, v_1, \hat{m}_1) \in ( H^1_{-\delta} (\R^-))^3 : v_1 (0) = 0 \} \label{e: adj domain d}
\end{equation}
in the Dirichlet case. In both cases, $\mathcal{L}^*$ is defined by the formula 
\begin{equation}
\mathcal{L}^* \begin{pmatrix}
	\hat{\psi}_1 \\ v_1 \\ \hat{m}_1 
	\end{pmatrix}
	= \left[ - \frac{\rmd}{\rmd x} + \begin{pmatrix}
	0 & 1-3\psi_\mathrm{cl/d}^2 & 0 \\
	-1 & 0 & 0 \\
	0 & -1 & 0
	\end{pmatrix}
	\right]
	\begin{pmatrix}
	\hat{\psi}_1 \\ v_1 \\ \hat{m}_1 
	\end{pmatrix} =
    \begin{pmatrix}
    -\hat{\psi}_1' + (1-3\psi_\mathrm{cl/d}^2)v_1 \\
    -v_1' - \hat{\psi}_1 \\
    -\hat{m}_1' - v_1
    \end{pmatrix}. 
\end{equation}
Searching for the kernel of $\mathcal{L}^*$ reduces to solving 
\begin{align}
v_1'' + (1-3\psi_\mathrm{cl/d}^2) v_1 &= 0, \notag \\
\hat{\psi}_1 &= -v_1', \notag \\
\hat{m}_1' &= - v_1
\end{align}
with appropriate boundary conditions given by \eqref{e: adj domain cl} and \eqref{e: adj domain d}, respectively. The equation for $v_1$ is again the linearization of the equation for $\psi_\mathrm{cl/d}$, hence we obtain solutions 
\[
(\hat{\psi}_1, v_1, \hat{m}_1) = (-\alpha \psi_\mathrm{cl/d}'', \alpha \psi_\mathrm{cl/d}', -\alpha \psi_\mathrm{cl/d} + \beta),
\]
for arbitrary constants $\alpha$ and $\beta$. Choosing $\delta$ sufficiently small again guarantees that these are the only possible solutions. For clamped boundary conditions \eqref{e: adj domain cl}, we obtain $\beta = 0$, hence the cokernel is spanned by $e_*^\mathrm{cl} = (-\psi_\mathrm{cl}'', \psi_\mathrm{cl}', -\psi_\mathrm{cl})$. Dirichlet boundary conditions \eqref{e: adj domain d} force $\alpha = 0$, since $\psi_d'(0) \neq 0$, so in this case the cokernel is spanned by $e_*^\mathrm{d} = (0, 0, 1)$. Note that constants and asymptotically constant functions are allowed in our space due to the exponential weight, now appearing with opposite sign $-\delta$ for the $L^2$-dual of $L^2_{\delta}$. In either case, the cokernel is one dimensional, as claimed. 

To complete the proof, one needs to verify that $\mathcal{L}$ has closed range. This follows from an abstract closed range lemma (see \cite{schwarz1993morse}) or the methods of Palmer \cite{palmer}. 
\end{proof}

In order to solve \eqref{e: IFT eqn} using the implicit function theorem, we make use of our far-field/core decomposition to treat $\eta$ as a variable. Provided the linearization $\partial_\eta F(u_0^\mathrm{cl/d})$ does not lie in the range of $\mathcal{L}$, appending $\eta$ as a variable increases the dimension of the range of the derivative of $F$ by 1, and hence the derivative becomes invertible. The implicit function theorem then gives the existence of a unique solution $(\hat{\psi} (c_x), v(c_x), \hat{m}(c_x), \eta(c_x); c_x)$ near $u_0$ in $X_\mathrm{cl/d} \times \R^2$ (in particular for sufficiently small $c_x$), depending smoothly on $c_x$. Thus, the argument is complete once we prove the following lemma:

\begin{lemma}[Transversality] \label{l:tr}
The derivative $\partial_\eta F(u_0^\mathrm{cl/d})$ does not lie in the range of $\mathcal{L}$. 
\end{lemma}
\begin{proof}
The range of $\mathcal{L}$ is orthogonal to the kernel of $\mathcal{L}^*$. In \S\ref{s:3.2}, we computed the projections of $\partial_\eta F(u_0^\mathrm{cl/d})$ onto the respective adjoint kernels, and found them to be nonzero, which proves the lemma. 
\end{proof}
\begin{remark}[Geometry]
 The fact that  $\partial_\eta F(u_0^\mathrm{cl/d})$ does not lie in the range, or, equivalently, that the scalar product with the kernel of the adjoint does not vanish, has an equivalent geometric interpretation in terms of transversality, hence the name of Lemma \ref{l:tr}. In the Dirichlet case, both unstable manifold $\eta=1$ and the subspace of solutions satisfying the boundary conditions are one-dimensional. Adding the asymptotic state $\eta$ as a parameter, we merely consider the center-unstable manifold, now two-dimensional, and show that it intersects the boundary subspace transversely. Inspecting the phase portrait, this fact is of course quite obvious. 
\end{remark}

\section{Moderate growth rates: from oblique stripes to zigzags through homoclinic bifurcations}\label{s:4}

Beyond the existence result for small speed, we analyze here the dynamics in the Cahn-Hilliard approximation \eqref{e:cht} with clamped or Dirichlet boundary conditions at $x=0$. We discuss existence of oblique stripe solutions and their stability for both clamped and Dirichlet boundary conditions, based on numerical continuation. In particular, we describe the kink-dragging bubble in \S\ref{s:4.1}. We then describe in more detail the endpoint of the bubble, where oblique stripes disappear in a saddle-node bifurcations that gives rise to time-periodic solutions, \S\ref{s:4.2}. We conclude in \S\ref{s:4.3} with a discussion of larger speeds, when the periodic kink-shedding detaches from the boundary.

\subsection{The kink-dragging bubble}\label{s:4.1}
We solved \eqref{e:chtw0} numerically setting $c_y=-\eta c_x$, 
\begin{align}
\psi_x&=v\notag\\
v_x&=-\psi+\psi^3+m\notag\\
m_x&=c_x(\psi-\eta),\label{e:chtw00}
\end{align}
with boundary conditions
\begin{equation}\label{e:bch}
\begin{array}{lll}\psi=0,&\qquad \tau v+(1-\tau) m =0,& \mbox{ at } x=0,\\
\psi=\eta,&\qquad v=0,& \mbox{ at } x=-L,
\end{array}
\end{equation}
using second-order finite differences for $L=100$ and grid spacing $dx=0.01$, observing no noticeable changes after increasing $L$ or decreasing $dx$. The parameter $\tau$ interpolates between Dirichlet boundary conditions $\psi=\psi_{xx}=0$ at $\tau=0$ and clamped boundary conditions $\psi=\psi_x=0$  at $\tau=1$.

Note that the three-dimensional ODE is then equipped with 4 boundary conditions, and the resulting overdetermined system is solved by leaving the asymptotic angle $\eta$ as a free variable, a procedure which mimics well the Fredholm analysis in \S\ref{e:chtw0}. The number of boundary conditions at $x=-L$ can also be understood as defining a one-dimensional linear subspace which approximates the one-dimensional unstable manifold of the equilibrium $\psi=\eta$ at zeroth order. 

\begin{figure}[h] \centering
\includegraphics[width=\linewidth,trim = {3.8cm .5cm 3.8cm 0},clip]{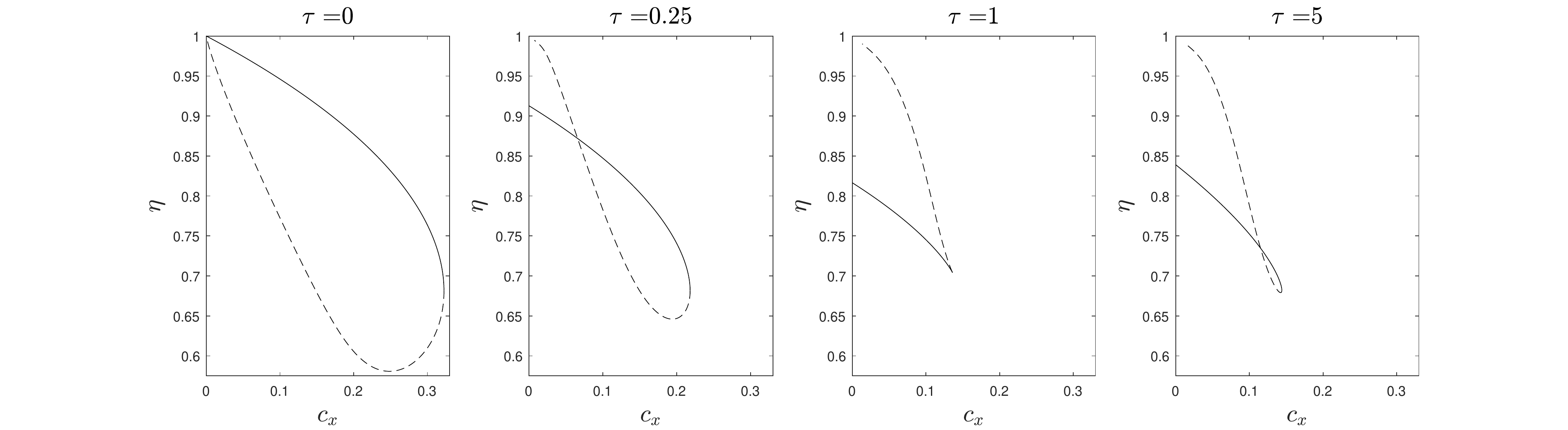}
\caption{Bifurcation diagrams of heteroclinic orbits in \eqref{e:chtw00} for several values of the boundary homotopy parameter $\tau$ \eqref{e:bch}.Shown is the value $\eta=\psi(x=-L)$ representing the angle of oblique stripes as a function of the speed $c_x$. Solid lines correspond to linearly stable, dashed to unstable solutions.}
\label{fig:teardropPlot}
\end{figure}

We used arclength continuation starting at $c_x=0$ to compute solutions for positive $c_x$; see Figure \ref{fig:teardropPlot}. The saddle-node bifurcation occurs at $c_x^\mathrm{sn}=0.136$, $\eta^\mathrm{sn}=0.704$ for clamped boundary conditions and $c_x^\mathrm{sn}=0.322$, $\eta^\mathrm{sn}=0.681$ for Dirichlet boundary conditions. Note that the saddle-node is apparently degenerate in the projection onto $\eta$ in the case of clamped boundary conditions, a fact that we corroborated by computing the kernel at the saddle-node location which exhibits a zero $\eta$-component. We illustrate how the ``folding'' at the saddle-node changes orientation near $\tau=1$ by continuing the homotopy past $\tau=1$. Incidentally, we found that $c_x^\mathrm{sn}$ is minimal at $\tau=1$.

Figure \ref{fig:teardropsols} shows selected solutions profiles and spectra of linearized operators obtained from linearizing \eqref{e:cht} at these stationary solutions. We notice that, continuing through the saddle-node, solution profiles turn non-monotone at the bifurcation point and, continuning back to $c_x=0$ on the unstable branch,  ultimately develop a kink. In an unbounded domain, solutions on the unstable branch converge locally uniformly as $c_x\searrow 0$ to the reflected solution with $\eta=-\eta(c_x=0)$, while a kink near $x=-\infty$ mediates a jump back from $\eta=-1$ to $\eta=+1$. 

The linearized spectra, computed in large domains, approximate the extended point spectrum and the absolute spectrum in the unbounded domain; see \cite{ssabs}. We computed the curves given by the absolute spectrum  of $-\partial_x^4+(1-3\eta^2)\partial_x^2+c_x\partial_x$ via continuation as outlined in \cite{rss}. We confirmed that most eigenvalues cluster on these curves, with the exception of a simple isolated real  eigenvalue that crosses the origin in the saddle-node bifurcation. 


The rightmost points of the absolute spectrum are pinched double roots which are stable as long as the selected state $\eta(c_x^\mathrm{sn})$ is convectively stable (which is true for all computed profiles, here since $\eta>1/\sqrt{3}$ is linearly stable). Note that the spectrum in the unbounded domain contains a branch of continuous spectra, inherited from the linearization at $\psi\equiv\eta$, that can be readily computed using Fourier transform as 
\[
\lambda = -k^4+(1-3\eta^2)k^2+c_x\rmi k,\qquad k\in\R.
\]
The zero mode $\lambda=k=0$ is caused by neutral mass conservation at $x=-\infty$. 

As a consequence, absent a spectral gap between the zero eigenvalue and the continuous spectrum, it is not immediately clear how this saddle-node bifurcation could be analyzed using temporal center-manifold reductions, for instance. We will nevertheless pursue such a reduction formally in the next subsection. 
\begin{figure}[h] \centering
\includegraphics[width=1\textwidth]{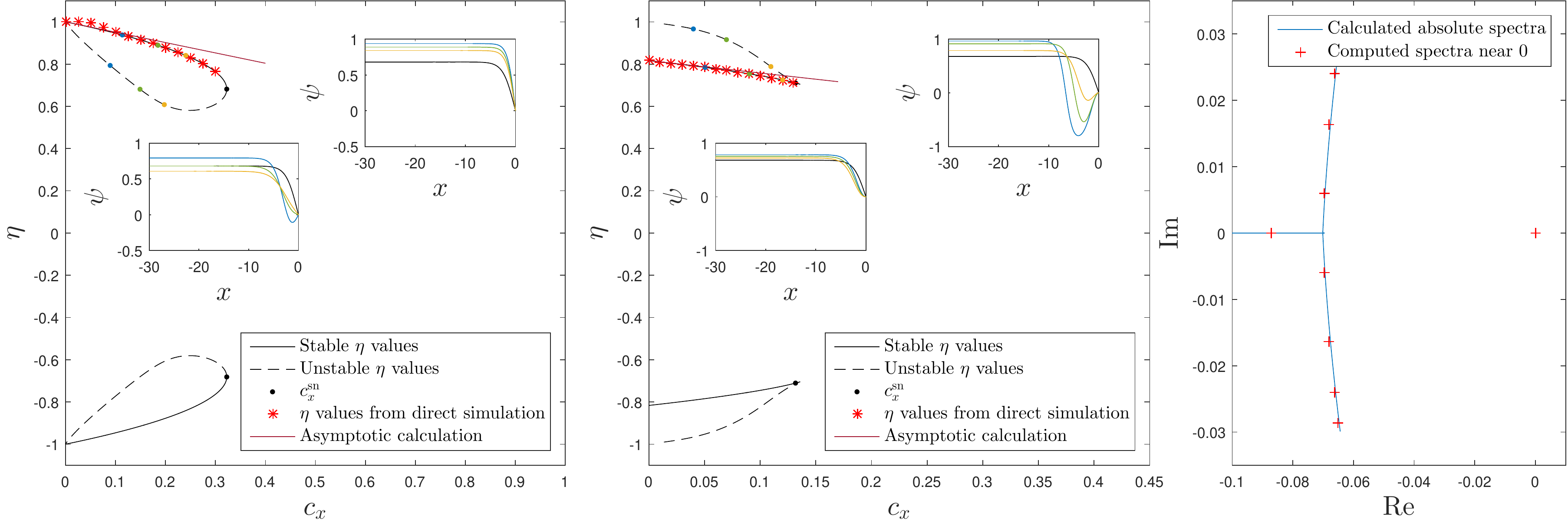}
\caption{Bifurcation diagrams for \eqref{e:chtw00}--\eqref{e:bch}, $\tau=0$ (left) and $\tau=1$ (center), with computed profiles in insets and red markers for values from direct simulations. The right figure shows the spectrum of the linearization at the critical equilibrium, demonstrating that the saddle-node bifurcation is caused by an isolated eigenvalue. Blue superimposed lines show absolute spectra; see text for details.}
\label{fig:teardropsols}
\end{figure}
In order to further demonstrate the nature of the saddle-node bifurcation, we investigated perturbations of the unstable equilibrium close to the saddle-node in direct simulations. We found the typical separation of the neighborhood of the unstable equilibrium by a codimension-one stable manifold. perturbations on either side of this manifold lead to release of a single kink and convergence to the reflected, negative, stable equilibrium, or to the stable equilibrium nearby after annihilation of the trapped kink in the unstable profile at the boundary $x=0$, respectively; see Figure  \ref{fig:kinkrelease}.
\begin{figure}[h!] \centering
\includegraphics[width=.7\linewidth]{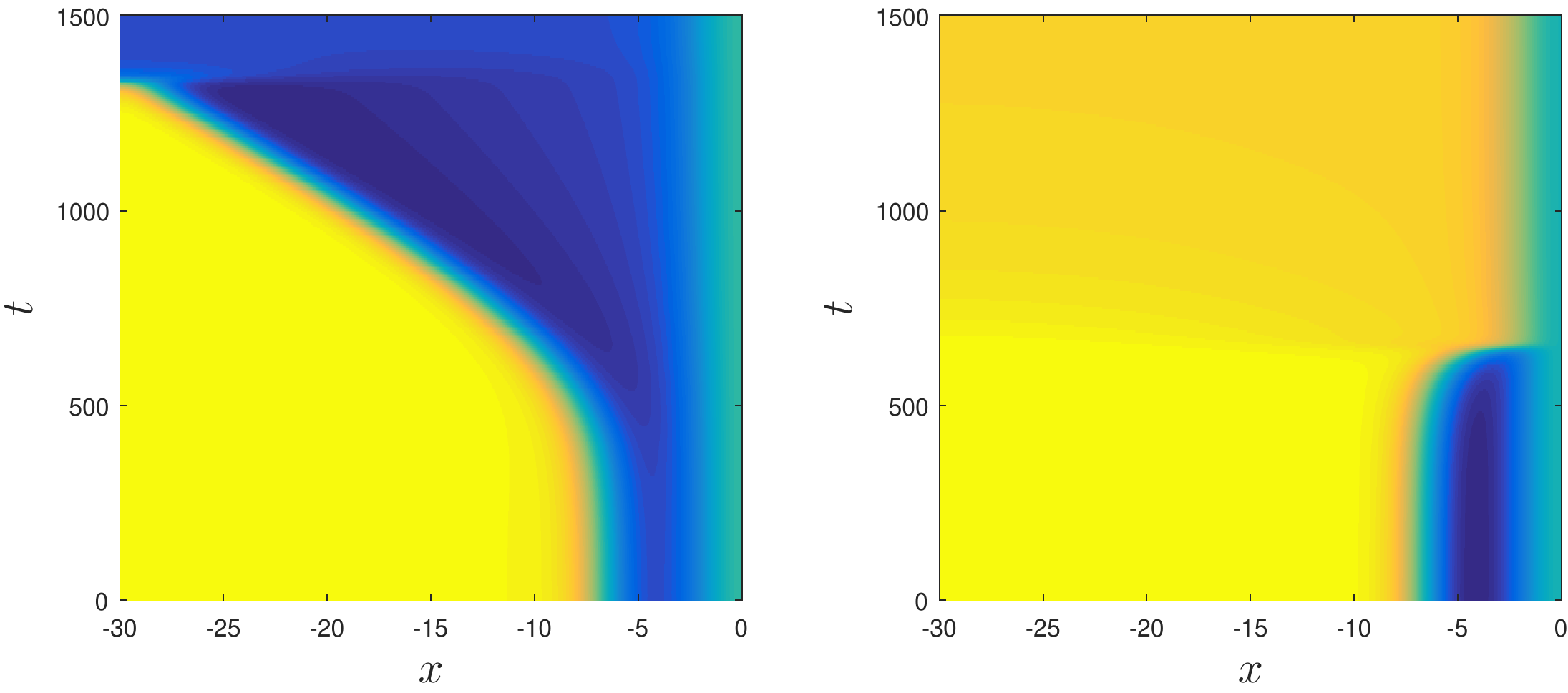}
\caption{Space-time plots of perturbations of the unstable solution profile, resulting in either release (left) or annihilation (right) of the kink. }
\label{fig:kinkrelease}
\end{figure}

\subsection{Kink-shedding ---  saddle-node on a limit cycle}\label{s:4.2}
The simulations in Figure \ref{fig:kinkrelease} demonstrate that there appears to be a heteroclinic orbit connecting the two saddle-node equilibria: small perturbations of a saddle-node equilibrium converge locally uniformly to the other saddle-node equilibtium, conjugate by reflection $\psi\to -\psi$. Reflecting the heteroclinic gives a heteroclinic loop between the saddle-node equilibria, which can be thought of as a (double, since there are two equilibria) saddle-node bifurcation on a limit cycle. One therefore expects that for parameter values $c_x$ just past the saddle-node, one observes a periodic orbit with large temporal period, due to slow passages near the region in phase space where the saddle-node was located. One can therefore infer leading-order asymptotics of the period of the periodic orbit from the leading-order expansion of dynamics on the center-manifold, only. We shall attempt to compare such predictions with periods measured in direct simulations, below.

Before calculating this expansion, we notice however a technical difficulty for the problem posed on the unbounded half line. The kink released by the perturbation from the unstable (or the saddle-node) equilibrium travels to the left from $x=0$ with speed $c_x$ but never vanishes, such that the heteroclinic solution converges to the opposite saddle-node equilibrium locally uniformly, but not in any translation-invariant norm that one may want to use to establish well-posedness of the equation. The problem is reflected in the presence of essential spectrum in the linearization at the equilibrium $\psi_*(x;c_x)$, stemming from the linearization at the constant $\eta_*$,
\[
\mathrm{spec}\,(-\partial_x^4+(1-3\eta^2)\partial_x^2+c_x\partial_x)=\{\lambda=-k^4+(1-3\eta^2)k^2+\rmi k,\ k\in\R\},
\]
which touches the origin at $\lambda=0$. Similar to the nonlinear considerations in Section \ref{s:3}, the essential spectrum can be stabilized in exponentially weighted norms \[
\|
u(x)\|_\delta=\|u(x)\rme^{-\delta x}\|_{L^2(\R^-)}, \qquad \delta\gtrsim 0,
\] 
but nonlinear analysis is typically not feasible in such norms. 

Nevertheless, we computed the eigenfunction $e$ associated with the kernel at the saddle-node and the associated adjoint eigenfunction $e^*$ to obtain an expansion for an effective equation on a center manifold,
\begin{equation}\label{e:sn}
A'=\alpha (c_x-c_x^\mathrm{sn})+\beta A^2+\rmO\left( (c_x-c_x^\mathrm{sn})^2+|c_x-c_x^\mathrm{sn}||A|+(|A|+|c_x-c_x^\mathrm{sn}|)^3\right),
\end{equation}
where 
\begin{equation}\label{e:sncoeff}	
\alpha=\int_{-\infty}^0 e^*(x) \left(\psi_\mathrm{cl/d}(x;c_x^\mathrm{sn})\right)_x\rmd x,\qquad  
\beta=\int_{-\infty}^0 e^*(x) \left(3  \psi_\mathrm{cl/d}(x;c_x^\mathrm{sn})e^2(x)\right)_{xx}\rmd x,
\end{equation}
with normalizations
\[
\int_{-\infty}^0 e^*(x) e(x)\rmd x=1,\qquad \int(e(x)\rme^{0.1x}\rmd x = 1,
\]
Since the adjoint eigenfunction $e^*$ is exponentially localized, all integrals converge, and we find values of $\alpha$ and $\beta$ for clamped and Dirichlet boundary conditions of 
\[
\alpha_\mathrm{cl}=-0.493\ldots,\qquad \beta_\mathrm{cl}=-0.0297\ldots,\qquad \alpha_\mathrm{d}=-0.959\ldots,\qquad \beta_\mathrm{d}=-0.0297\ldots .
\]

%
From the expansion, we compute a passage time near the saddle-node $T=\frac{\pi}{\sqrt{\alpha\beta (c_x-c_x^\mathrm{sn})}}$ which gives leading-order frequency $\omega$ and spacing $L$ of kinks
\begin{equation}\label{e:sns}
\omega=2\sqrt{\alpha\beta(c_x-c_x^\mathrm{sn})},\quad k=\omega/c_x,\quad L=2\pi/k.
\end{equation}
We compare the predictions with measurements in direct simulations and find good agreement, albeit not unexpectedly only for speeds $c_x$ very close to criticality; see Figure \ref{fig:wavelengthall}. Agreement is better for Dirichlet boundary conditions. For clamped boundary conditions, agreement is achieved only for values of $c_x$ extremely close to criticality. 

In the clamped case, the eigenfunction associated with the saddle-node is exponentially localized, in agreement with the fact that the saddle-node bifurcation in Figure \ref{fig:teardropPlot} is degenerate. Perturbations of the saddle-node equilibrium lead to global excursions that converge back to this equilibrium, in a leading direction not associated with this eigenfunction but with continuous spectrum reflecting the slow shedding of a kink. In this sense, the excursion can be understood as a codimension-two homoclinic orbit at a saddle-node equilibrium that enters the critical equilibrium along a direction other than the saddle-node, leading to changed asymptotics \cite{chowlin}. Unfortunately, the direction associated with this flip of the homoclinic is not hyperbolic and asymptotics for periods of periodic orbits resulting from a flip bifurcation as in \cite{chowlin} do not give better approximation results. 

On the other hand, the presence of a flip bifurcation usually marks the boundary between a homoclinic orbit to a saddle-node bifurcation on a limit cycle and a homoclinic orbit to a hyperbolic equilibrium. We noticed that for moderate domain sizes and clamped boundary conditions, the limit of periodic orbits is indeed a homoclinic orbit to the unstable equilibrium resulting from the saddle-node bifurcation\footnote{We always refer to homoclinic orbits here for simplicity when indeed the objects are homoclinic cycles, consisting of two heteroclinic orbits connecting equilibria related by $\psi\mapsto -\psi$.} In particular, one finds a small region of coexistence of periodic orbits, that is, the creation of zigzag patterns, and stable equilibria, that is, the creation of oblique stripes; see Figure \ref{fig: c_crit moderate dom}. 
 
\begin{figure}%
\centering
\subfigure[Dirichlet boundary conditions]{%
\label{fig:wavelengthDirichlet}%
\includegraphics[width=0.48\linewidth]{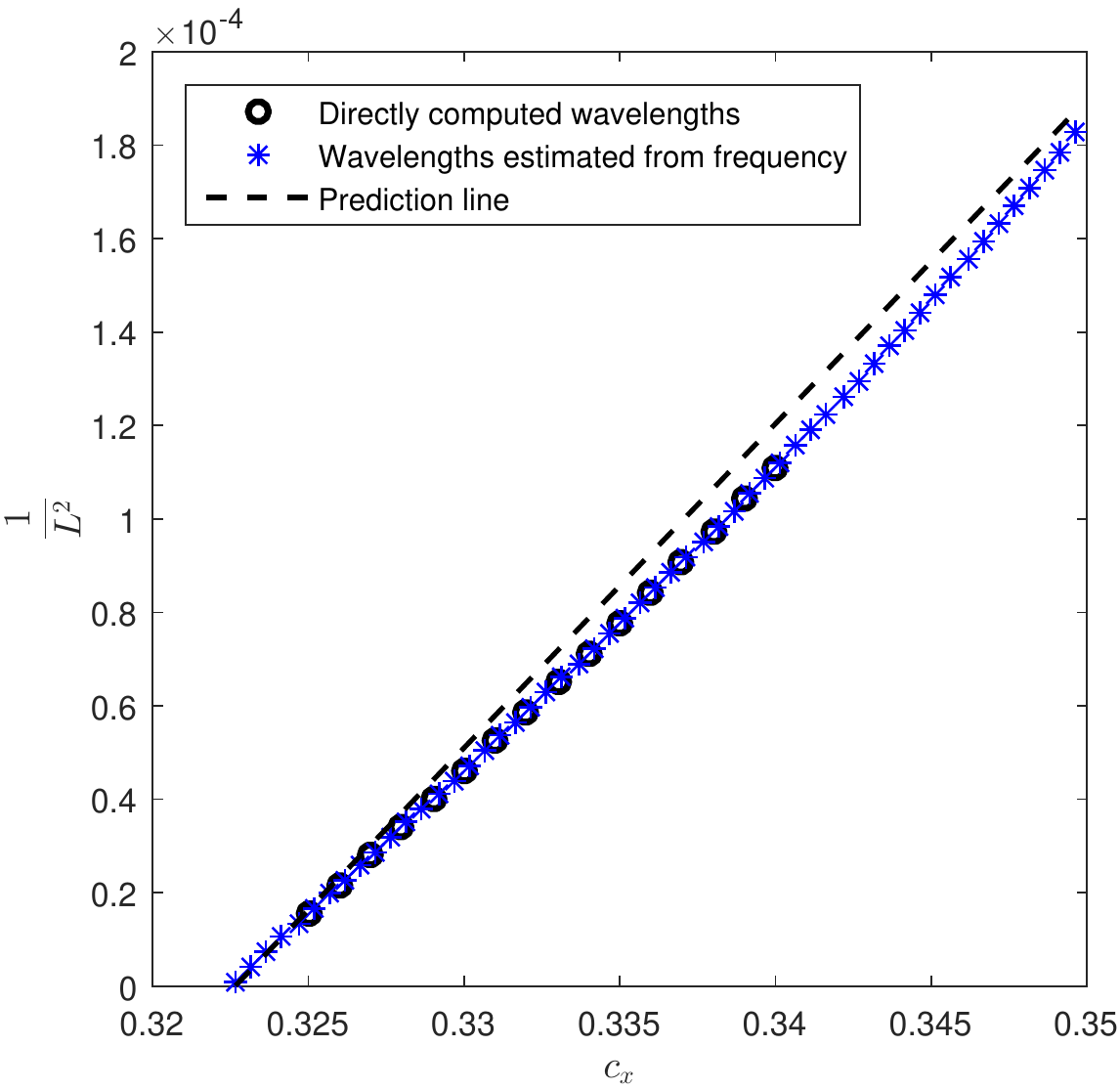}}%
\hfill
\subfigure[Clamped boundary conditions]{%
\label{fig:wavelengthClamped}%
\includegraphics[width=0.470\linewidth]{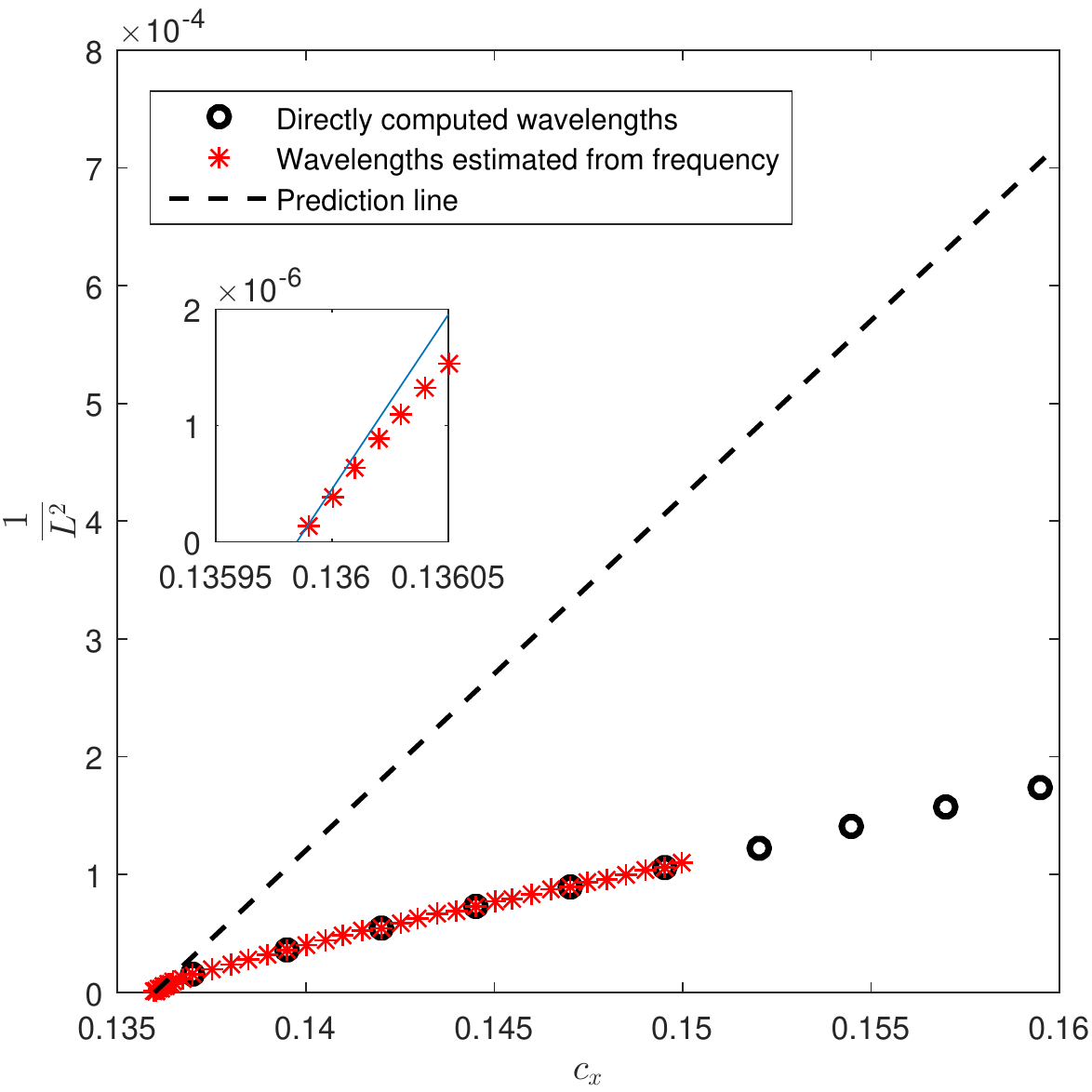}}%
\caption{Plot of inverse-square of wavelength close to $c_x^\mathrm{sn}$ according to the predicted square root scaling \eqref{e:sns} and measured frequency and wavelength data. Wavelengths from direct simulations agree well with the predicted scaling for Dirichlet boundary conditions (left), at least for large wavelengths, but agreement is limited to a very small region near the bifurcation point for clamped boundary conditions (right). Direct simulations were based on second order finite differences with domain size $250$,  $dx=0.01$ using  \textsc{matlab}'s \textsc{ode15s} for time stepping. Measurements of wavelengths were both spatial ($\circ$) and indirect through temporal periods ($*$), multiplied by the speed $c_x$.}\label{fig:wavelengthall}%
\end{figure}

Figure \ref{fig:oscProfile}, shows dependence of the $L^\infty$-norm on $c_x$ and illustrates select solution profiles. We see that the periodic solutions converge, as $x\to-\infty$, to stationary solutions of the Cahn-Hilliard equation (in the steady frame), which in turn converge to concatenations of the layer solution (or kinks)  $\psi(x)\sim \tanh(x/\sqrt{2})$ with amplitude $1$, as the wavelength tends to $\infty$. The spatial convergence of profiles is illustrated in Figure \ref{fig:oscProfile}, showing a twin-horn structure with local minima near the value of $\eta$ at the saddle-node bifurcation in Figure \ref{fig:teardropPlot}. 
\begin{figure}[h!] \centering
\raisebox{0.08in}{\includegraphics[width=0.43\textwidth]{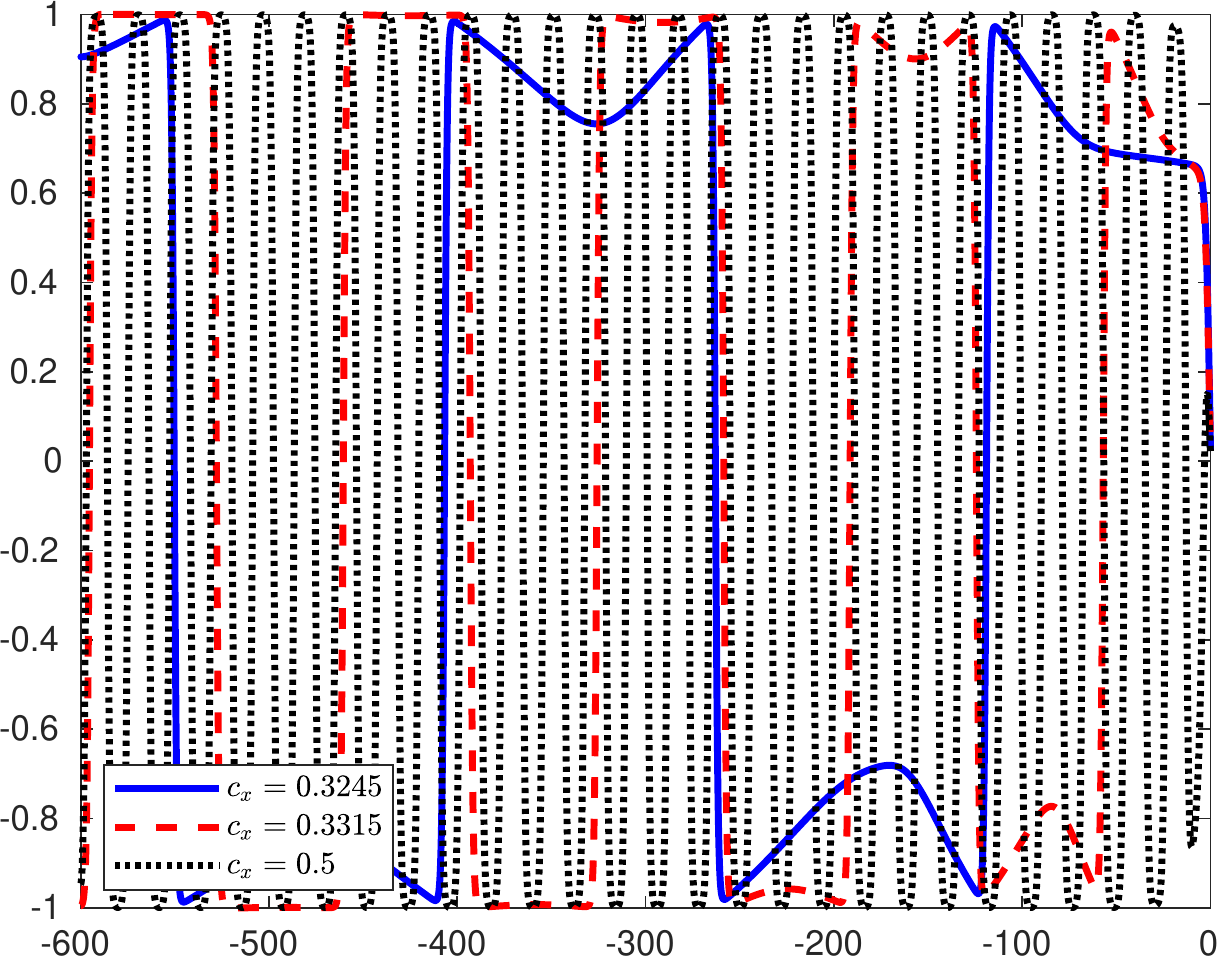}}\qquad
\includegraphics[width=0.46\textwidth]{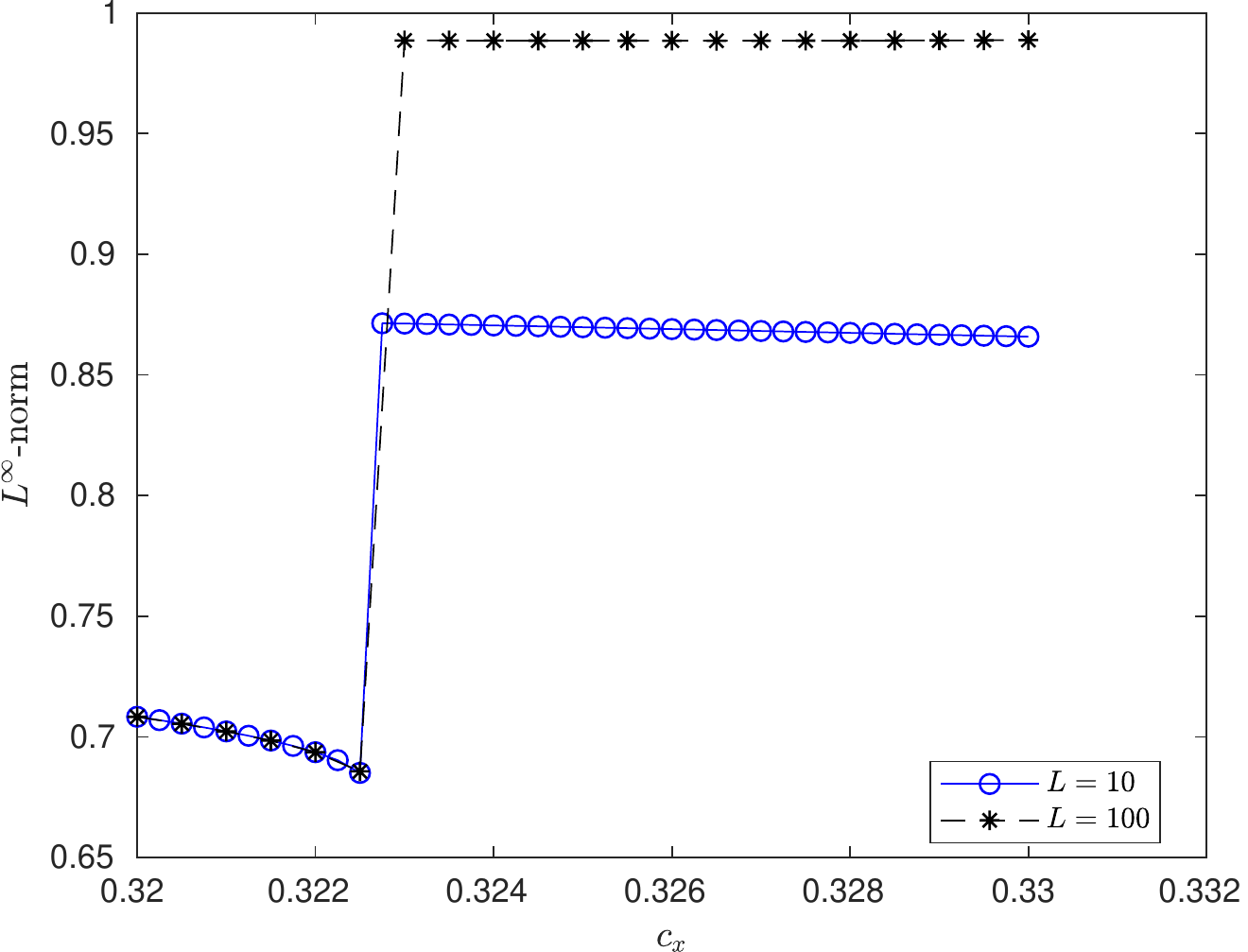}
\caption{Solution profile snapshots  from direct simulations for $c_x^\mathrm{sn} < c_x < c_x^\mathrm{lin}$ (left); we observe oscillations, creating sequences of up-down kinks. As $c_x$ increases, the period of the oscillations decreases. Large period periodic orbits develop a characteristic overshoot in the form of twin-horns, stemming from the mismatch of the selected angle $\eta$ at the saddle-node and the value $\eta=\pm1$ selected by the kink. The twin-horns develop as $x\to-\infty$, as can be seen from a plot of the $L^\infty$-norm of solutions for different domain sizes (right). }
\label{fig:oscProfile}
\end{figure}
It would clearly be interesting to analyze this bifurcation in a more precise asymptotic analysis. 

\begin{figure}[h!] \centering
\subfigure{%
\label{fig:c_sn/hom vs L}%
\includegraphics[width=0.32\linewidth]{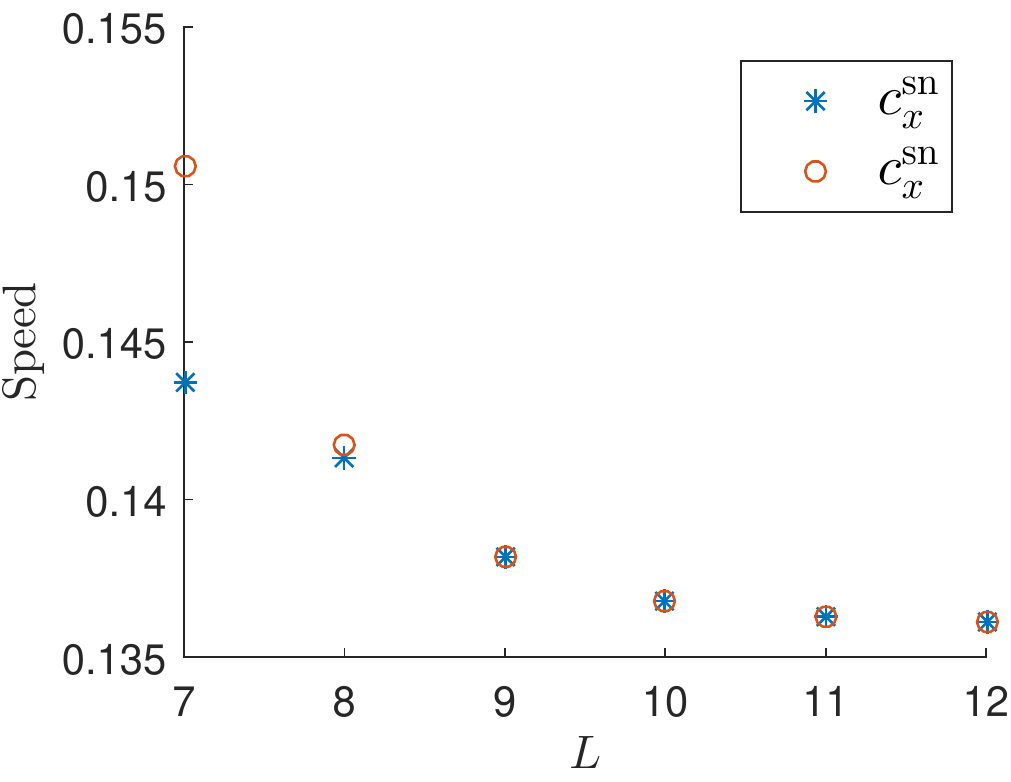}}%
\hfill
\subfigure{%
\label{fig:c_sn/hom difference}%
\includegraphics[width=0.32\linewidth]{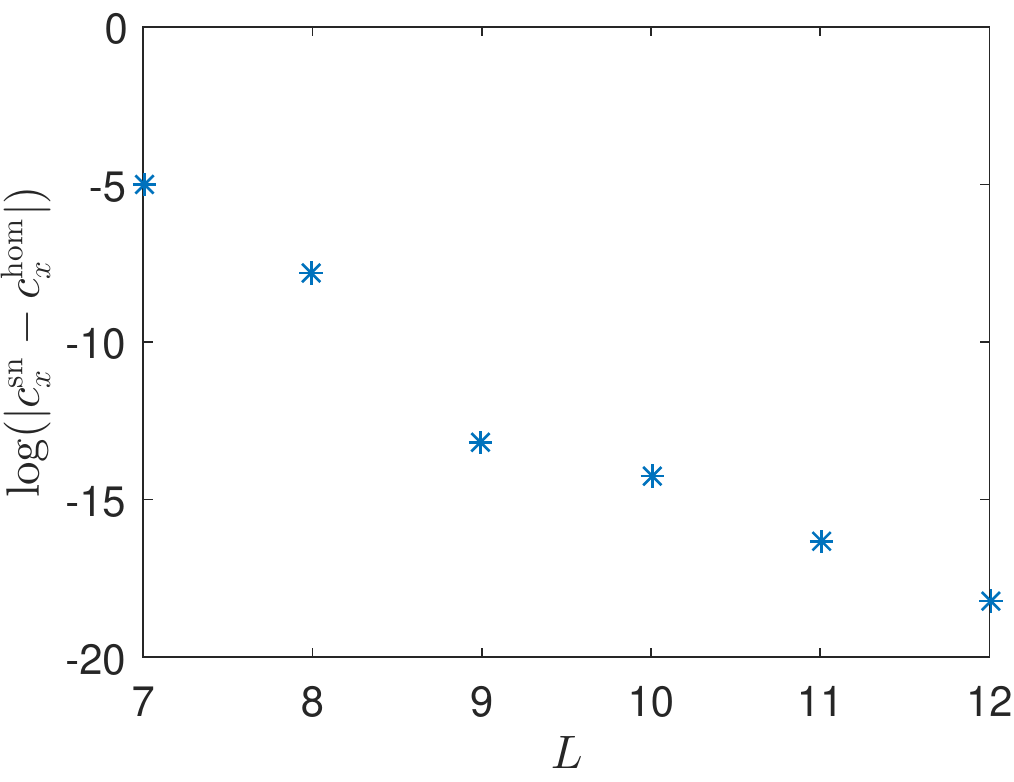}}%
\hfill
\subfigure{
\label{fig:bistability L=7}
\includegraphics[width=0.32\linewidth]{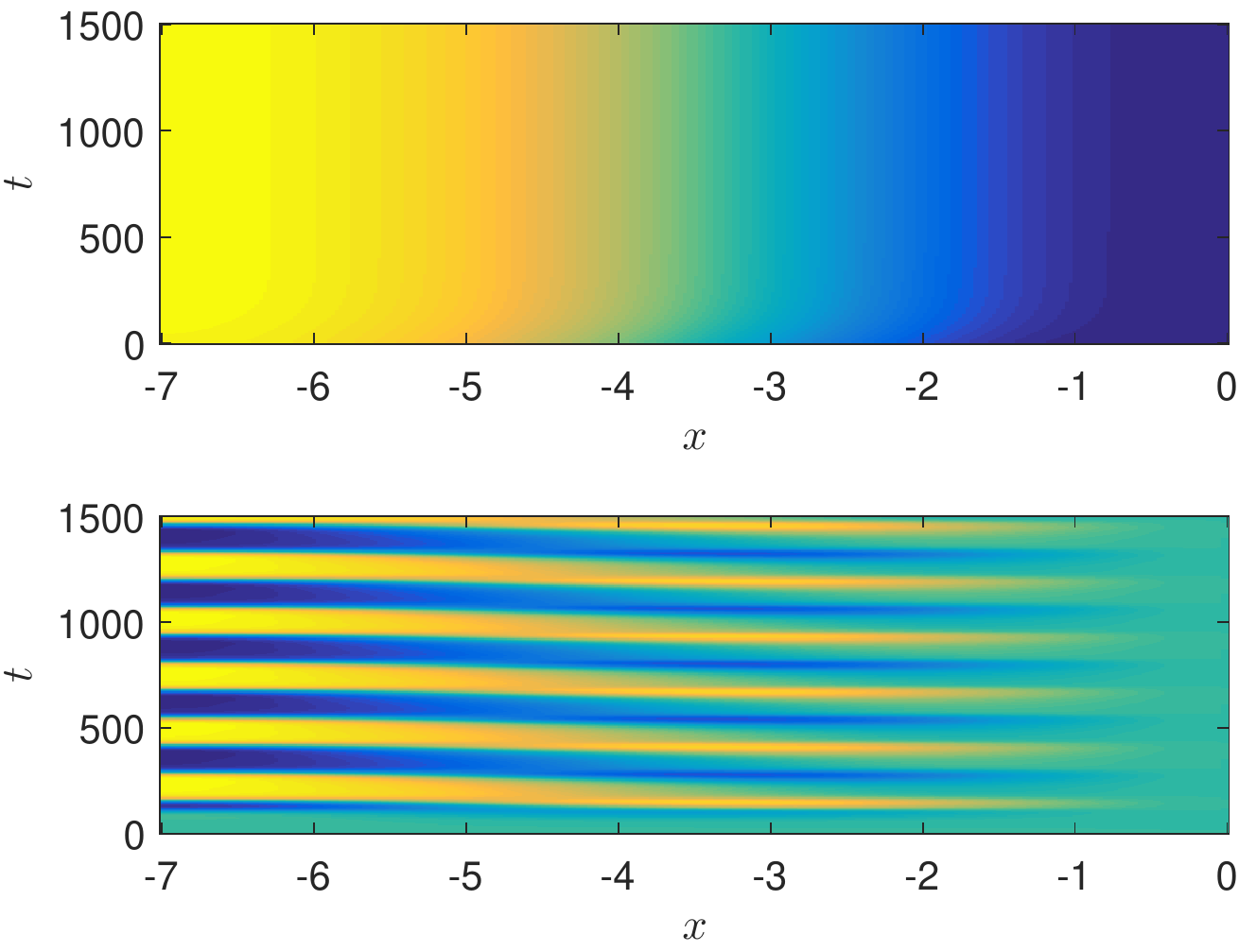}}
\caption{Critical speeds $c_x^\mathrm{sn/hom}$ over moderate sized domains, $L = 7, 8,..., 12$, for clamped boundary conditions. Left: $c_x^\mathrm{hom}$ (the speed at which the periodic orbits disappear) and $c_x^\mathrm{sn}$ (the speed at which the saddle node occurs) versus domain size. Both converge as $L \to \infty$. Middle: $\log(|c_x^\mathrm{sn} - c_x^\mathrm{hom}|)$ versus domain size, showing that this difference converges to zero exponentially in $L$. Right: spacetime demonstration of bistability for $L = 7, c_x = 0.148$. Initial data for the top right plot was a small perturbation of the equilibrium at the saddle node, while the bottom right plot used low amplitude random initial data, and converges to a stable periodic orbit.}
\label{fig: c_crit moderate dom}
\end{figure}

\subsection{Kink-shedding beyond the saddle-node and detachment}\label{s:4.3}
Beyond the immediate vicinity of the saddle-node, we find a continuously decreasing wavelength until the kink-shedding detaches and we relax to $\psi\equiv 0$ as the stable solution. In any bounded domain, this detachment transition induces a very steep bifurcation, common for transitions between convective and absolute instabilities \cite{koepf1,knobloch,ssbasin}. Speed and wavenumber converge to speed and wavenumber of the invasion front in the Cahn-Hilliard equation, given by 
\[
c_x^\mathrm{lin}=1/3 \sqrt{\frac{2}{3} \left(-1 + \sqrt{7}\right)} \left(2 + \sqrt{7}\right)\sim 1.62208, \qquad k^\mathrm{lin}=\frac{3 (3 + \sqrt{7})}{8 \sqrt{5 + \sqrt{7}}}\sim  0.765672.
\]	
The analysis in \cite{gs1}  demonstrates this limiting behavior in the case of the complex Ginzburg-Landau equation and, making conceptual assumptions on existence \cite{chfront} and transversality of the Cahn-Hilliard invasion front, should extend to the situation here; see \cite{gs2} for such a conceptual extension. Moreover, \cite{gs1} gives a first-order correction to the selected frequency near the linear front speed $c_x^\mathrm{lin}$ which is independent of the boundary conditions at $x=0$, obtained simply from the intersection of the absolute spectrum with the imaginary axis. The somewhat lengthy calculation of this intersection yields 
\begin{align}
k(c_x)&=\frac{3 (3 + \sqrt{7})}{8 \sqrt{5 + \sqrt{7}}}+ \frac{9 \sqrt{6 (2 + \sqrt{7})} (4 + \sqrt{7})}{800 + 304 \sqrt{7}}(c_x-c_x^\mathrm{lin})+\rmO\left((c_x-c_x^\mathrm{lin})^{3/2}\right)\notag\\
&\sim 0.765672 + 0.196835(c_x-c_x^\mathrm{lin}).\label{e:asych}
\end{align}
with good agreement for both clamped and Dirichlet boundary conditions; see Figure \ref{fig:oscDecay}.

\begin{figure}\centering
\includegraphics[width=0.48\textwidth]{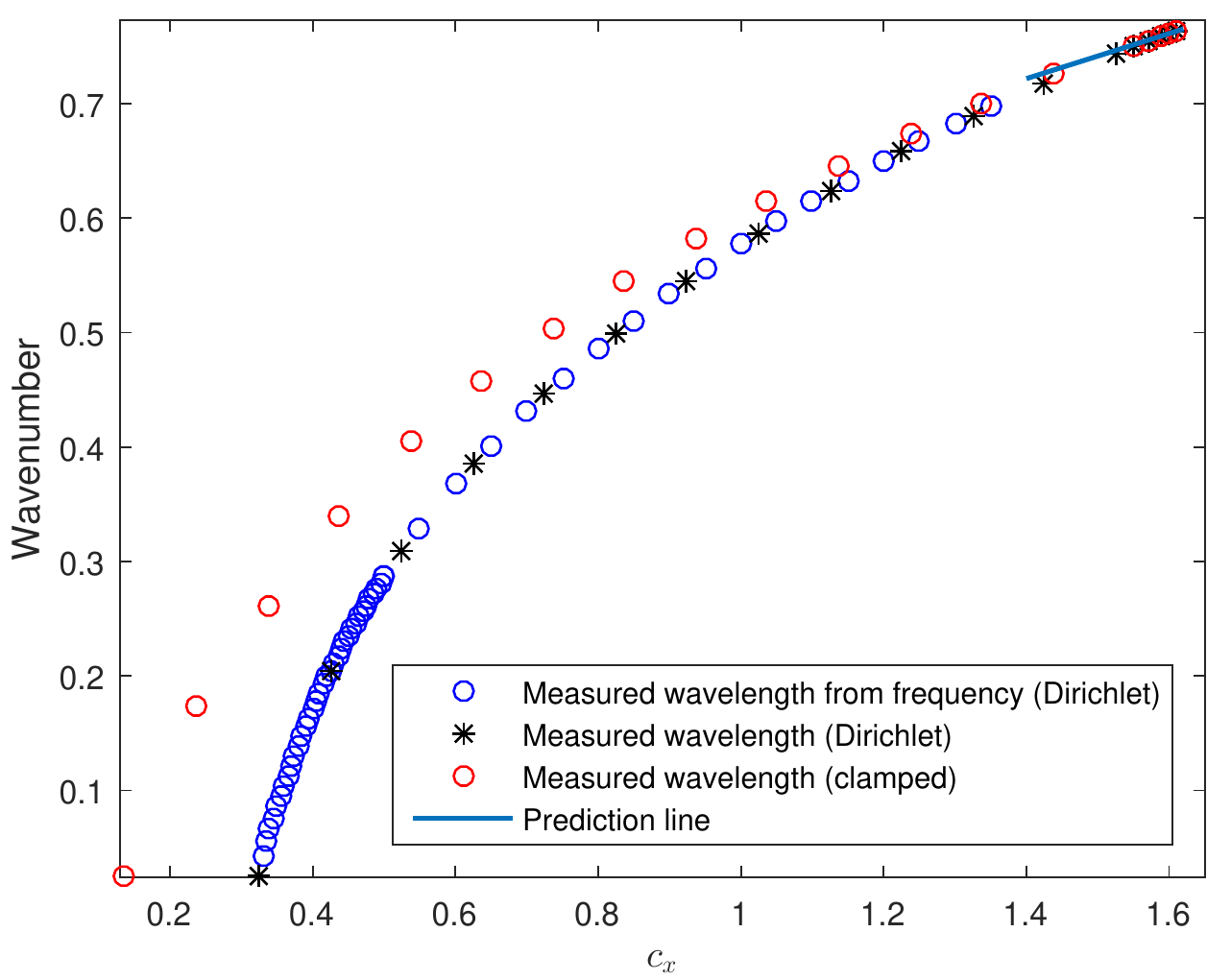}\hfill 
\raisebox{0.12in}{\includegraphics[width=0.45\textwidth]{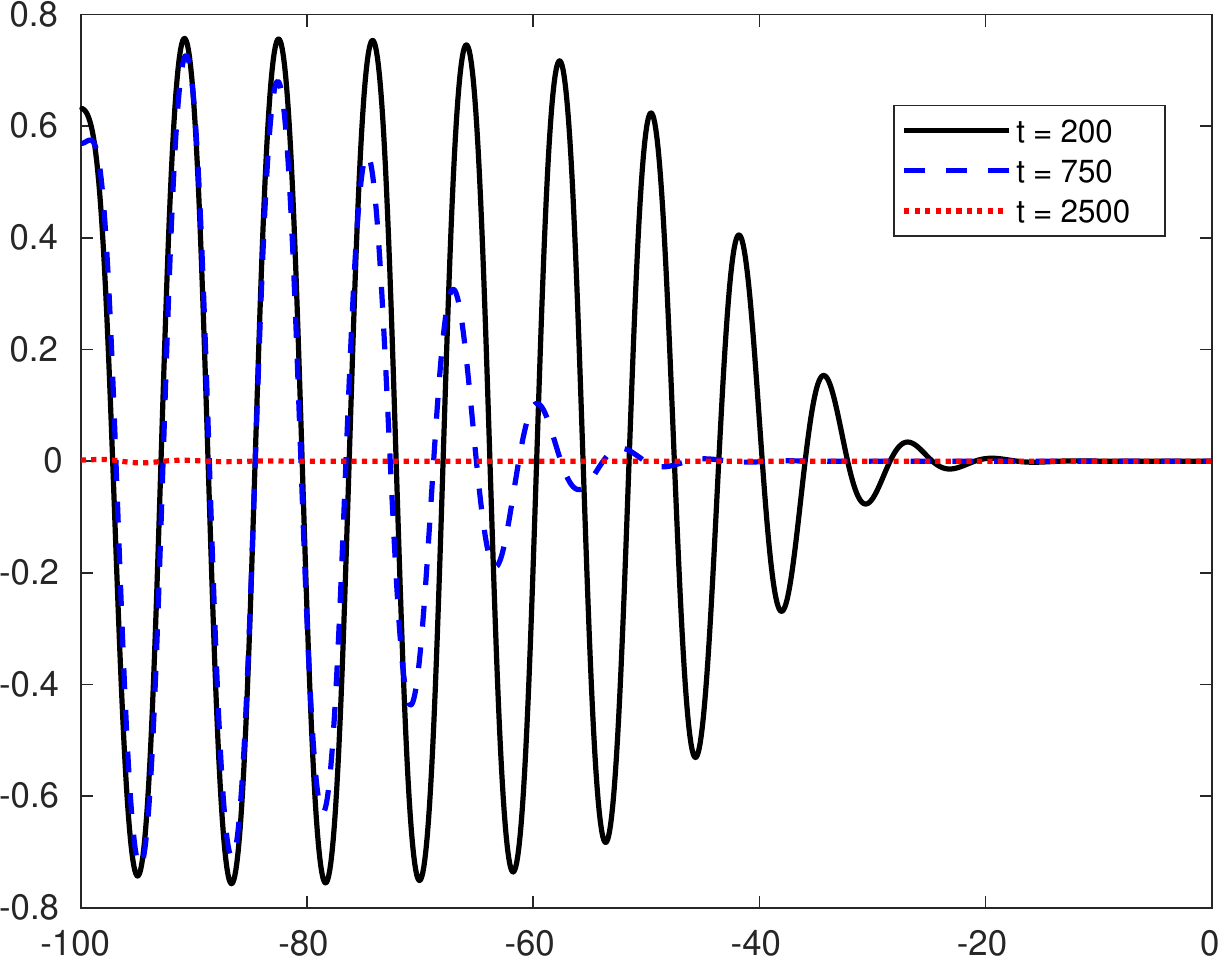}}
\caption{Measured wavenumber of periodic kink sequences  $c_x^\mathrm{sn}<c_x<c_x^\mathrm{lin}$, for both Dirichlet and clamped boundary conditions, including predictions from linear spreading speed theory as in  \eqref{e:asych} near $c_\mathrm{lin}$ (left). Sample solution profiles near detachment (right). Numerical details as in Figure \ref{fig:oscProfile}.}
\label{fig:oscDecay}
\end{figure}

\section{Perpendicular detachment, oblique reattachment, and all-stripe detachment}\label{s:5}

Thus far, we have observed and analyzed transitions for fixed wavenumber $k\lesssim k_\mathrm{zz}$, and  small speed, increasing from $c_x=0$. We first observe oblique stripes, that subsequently detach in a saddle-node on a limit cycle, giving rise to zigzagging patterns which eventually detach such that perpendicular stripes form at the quenching line. For all small parameter values, perpendicular stripes exist but are not observed, being unstable. All phenomena thus far were well described in the Cahn-Hillard phase approximation: perpendicular stripes correspond to the trivial state, subject to spinodal decomposition instability; oblique stripes correspond to the constant, spinodal stable regime;  and zigzagging stripes correspond to periodic patterns emerging from the spinodal decomposition. The Cahn-Hilliard approximation holds the inherent advantage of being a universally accurate description near a zigzag-instability, in particular for Swift-Hohenberg beyond the small-amplitude regime $\mu\ll 1$. 

For larger speeds, the phase approximation simply predicts perpendicular stripes. In order to capture the phenomena observed in the Swift-Hohenberg equation for such moderate speeds, we therefore study amplitude equations. In a quick summary, subject to many subtleties, as $c_x$ is increased we first observe a transition back to oblique stripes in the perpendicular detachment at $c_x^\mathrm{psn}(k_y)$, and subsequently a detachment and transition to parallel stripes. 

We study the transition from perpendicular stripes back to oblique stripes in \S\ref{s:5.1} and the detachment of stripes in \S\ref{s:5.2}.

\subsection{Detachment of perpendicular stripes and rattaching oblique stripes}\label{s:5.1}

We study the reattaching of oblique stripes in the Newell-Whitehead-Segel equations, derived from Swift-Hohenberg by evaluating and projecting onto the first Fourier mode $\rme^{\rmi k_y y}$,
\begin{equation}\label{e:nwsm}
A_t=-(\partial_{xx}+1-k_y^2)^2A + \rr(x) A - 3 A|A|^2 + c_x A_x.
\end{equation}
The subspaces $A\in\R$ and $A\in\rmi\R$ are invariant and correspond to solutions that are even with respect to reflections at $y=0$ and $y=\pi/k_y$, respectively. Several scalings are possible in this equation and we shall fix throughout 
\begin{equation}\label{e:mu}
\rr(x)=-\mu\sign(x),\qquad \mu=\frac{1}{4}. 
\end{equation}
Perpendicular stripes in $x<0$ correspond to solutions $A\equiv const$, oblique stripes to $A\sim \rme^{\rmi k_x x}$.

\paragraph{Continuing perpendicular stripes in $c_x$:  another saddle-node on a limit cycle}
Perpendicular stripes can be found as stationary solutions to \eqref{e:nwsm} with $A\in\R$, with boundary conditions $A(x)\to 0$ for $x\to\infty$, $A(x)\to r(k_y)$ for $x\to -\infty$, with 
\[
r^2(k_y)=\mu-(1-k_y)^2.
\]
We solved for solutions using numerical continuation and found a saddle-node bifurcation at $c_x^\mathrm{psn}(k_y)$; see Figure \ref{f:perp-bif} where the saddle-node is shown in red.  
We found that the saddle-node bifurcation curve ends at wavenumbers  $k_y = k_y^{\pm}$ with
\begin{align}
 k_y^-&\sim 0.781\ldots, &\qquad k_y^+&=\frac{\sqrt{4+\sqrt{3}}}{2}=1.19709\ldots,\notag\\
 c_x^\mathrm{psn}(k_y^-)&=0,&\qquad c_x^\mathrm{psn}(k_y^+)&=\frac{1}{\sqrt[4]{27}}=0.43869\ldots \label{e:snpol}
\end{align}
We will discuss the rationale for $c_x^\mathrm{psn}(k_y^+)$ and expansions for the saddle-node bifurcation for $k_y$ near this upper boundary in \S \ref{s:5.2} and analyze the behavior near $k_y^-$ at the end of this section\footnote{Figure \ref{f:perp-bif} also shows a light-gray curve that is the continuation of the green spreading speed for smaller values of $k_y$. For this curve, the speed is in fact complex and we plotted the real part, only. It appears to predict the saddle-node bifurcation surprisingly well, but we were not able to find any theoretical foundation for this apparent coincidence.}.

\begin{figure}
\centering\includegraphics[width=0.99\textwidth]{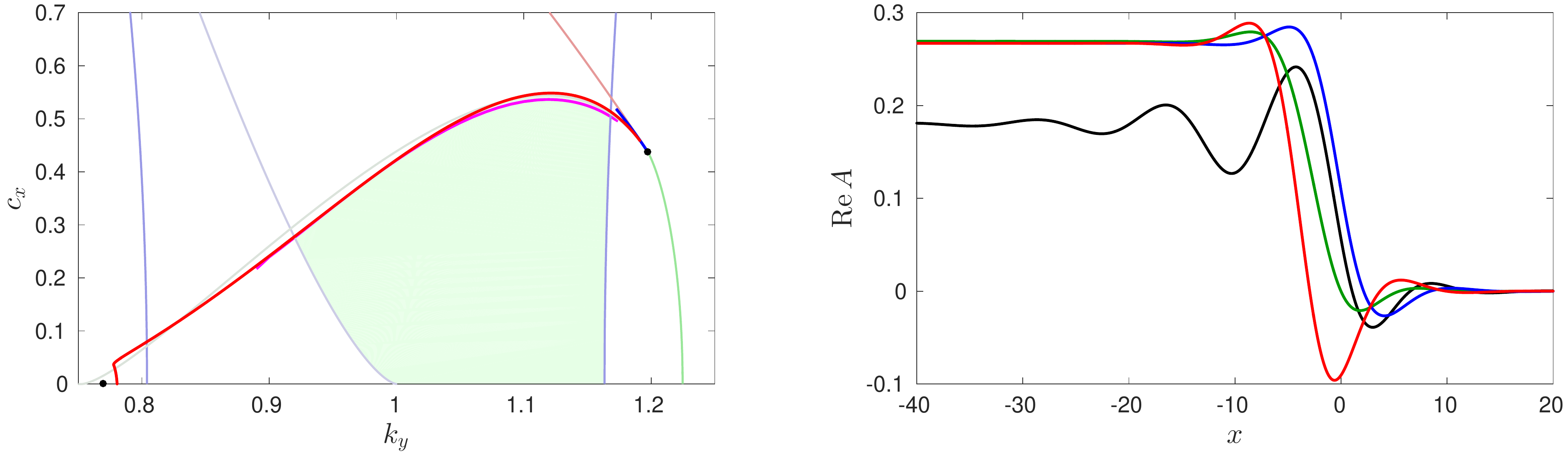}
\caption{Bifurcation diagram for perpendicular stripes in NWS, \eqref{e:nwsm}, in the $(k_y,c_x)$-plane (left). The red curve shows the saddle-node bifurcation $c_x^\mathrm{psn}(k_y)$, the green curve the detachment at the linear spreading speed \eqref{e:clinsh}, with black marker at the junction $(k_y^+,c_x(k_y^+))$ \eqref{e:snpol}. Perpendicular stripes exist in the bounded region marked by these red and green curves, $c_x\geq 0$. The blue curve near the junction marks the theoretical prediction for the transition from perpendicular to oblique stripes by triple point of the absolute spectrum \eqref{e:triple}. The magenta curve shows the pitchfork bifurcation to oblique stripes on the stable branch. Also shown in brown is the spreading speed of oblique stripes. The light gray curve shows the spreading speed of zigzags $c_x^\mathrm{zz}$ into perpendicular stripes  \eqref{e:zzlin}. The light blue curve shows the spreading speed $c_x^\mathrm{cr}$ of the cross-roll instability into perpendicular stripes \eqref{e:crlin}. The marker at $(k_y^\mathrm{t},0)$ denotes the Turing-type instability of perpendicular stripes against amplitude modulations \eqref{e:turing}. Sample profiles (right) near the left end point of the saddle-node curve $k_y=0.781$, $c_x=0$ (black) and at $k_y=0.9$, with $c_x=0.05$ (stable, blue; unstable, red) and $c_x= 0.2412$ at the saddle-node (green). }\label{f:perp-bif}
\end{figure}

Continuing through the saddle-node, one can follow the now unstable branch of perpendicular stripes decreasing $c_x$ and observe phenomena very similar to the kink shedding observed in the Cahn-Hilliard  equation, \S\ref{s:2}. The solution profile develops a kink which separates from the quenching line; see Figure \ref{f:perp-bif}. The kink typically possesses oscillatory tails, and therefore weakly locks to the quenching line, thus leading to a snaking bifurcation diagram near $c_x=0$, that is, the speed oscillates around $0$ while the distance of the kink from the quenching line increases. 

Intersting phenomena occur when, for smaller $k_y$, the saddle-node interacts with the snaking diagram. We explore this region in somewhat more detail in Section \ref{s:6.3}. We see that the snaking diagram breaks up into isolas which resemble at first figure-eight shapes with pairs of saddle-nodes. For yet smaller $k_y$, two saddle-nodes disappear in a cusp bifurcation and only an isola with two saddle-nodes remains, which eventually disappears when the two saddle-node bifurcations coalesce in a parabolic catastrophe at the minimum value of $k_y$ for which perpendicular stripes exist; see Figure \ref{fig:sh-isola}. Note that in this respect, the bifurcation diagram in Figure \ref{f:perp-bif} is rather incomplete, omitting in particular many saddle-node bifurcations  near the lower range of $k_y$-values.

Increasing $c_x$ past the saddle-node, one observes periodic kink-shedding similar to the situation in \S\ref{s:2}, with some caveats. The kink-shedding is only observed in spaces of functions that are even with respect to $y=0$ and odd with respect to $y=\pi/2$, or $y$-translates of functions in this subspace. Perturbations away from this subspace can lead to different phenomenologies, associated with destabilization and bifurcations of perpendicular stripes prior to the saddle-node in the  $(k_y,c_x)$-plane. Those stability boundaries are shown in Figure \ref{f:perp-bif} and we shall discuss them in somewhat more detail in the remainder of this section. The associated phenomenologies are illustrated in direct simulations for the Swift-Hohenberg equation in Figure \ref{f:sn}. At the end of this section, we shall return to the saddle-node curve and discuss the regime of small $k_y$.

\begin{figure}[h!]
\begin{minipage}{0.08\textwidth}
\includegraphics[width=0.89\textwidth]{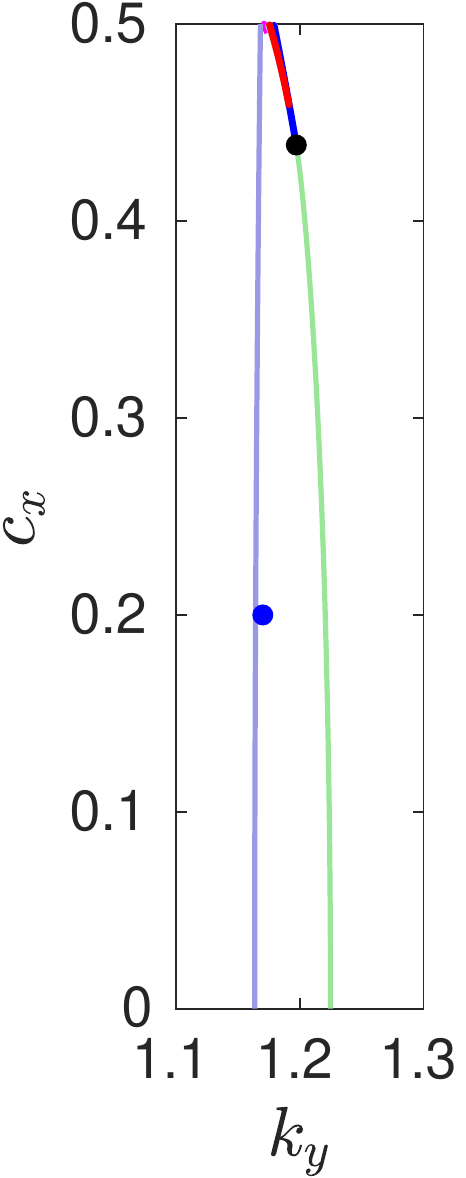}
\end{minipage}\hfill
\begin{minipage}{0.89\textwidth}\includegraphics[width=0.89\textwidth]{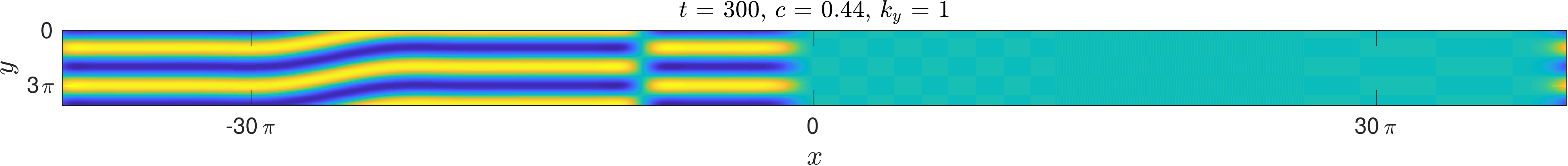}\hfill\raisebox{0.33in}{(1)}\\
\includegraphics[width=0.89\textwidth]{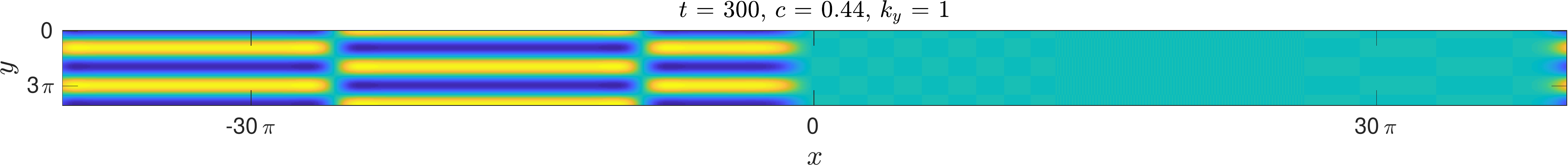}\hfill\raisebox{0.33in}{(2)}
\end{minipage}\\
\begin{minipage}{0.08\textwidth}
\includegraphics[width=0.89\textwidth]{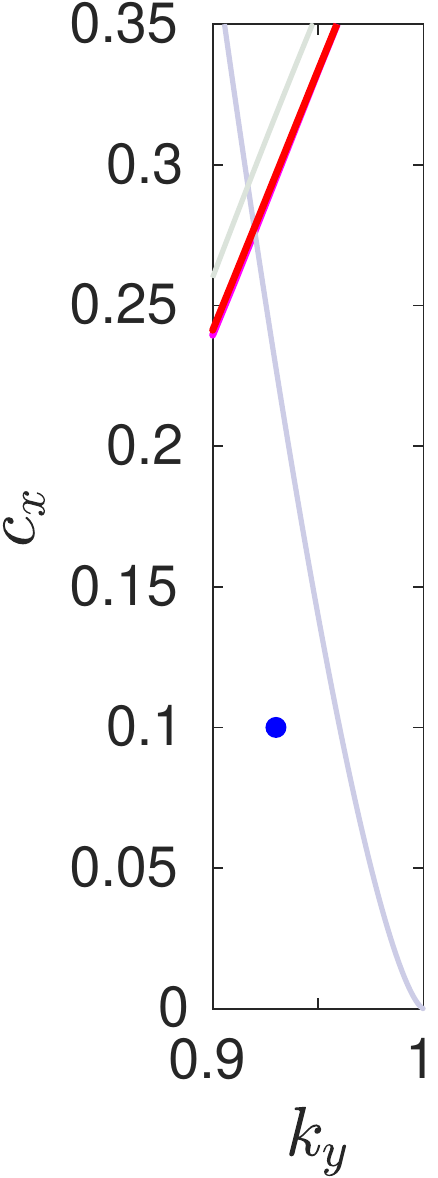}
\end{minipage}\hfill
\begin{minipage}{0.89\textwidth}\includegraphics[width=0.89\textwidth]{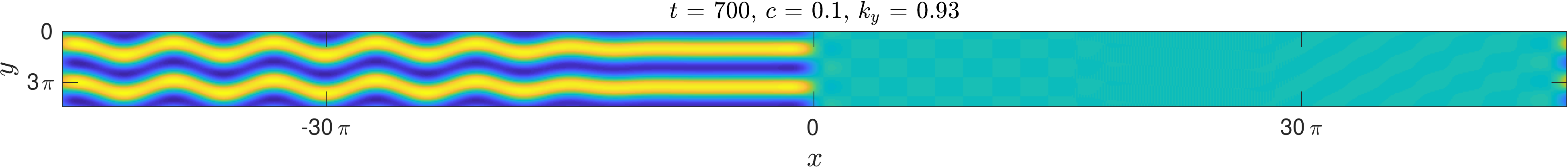}\hfill\raisebox{0.33in}{(3)}\\
\includegraphics[width=0.89\textwidth]{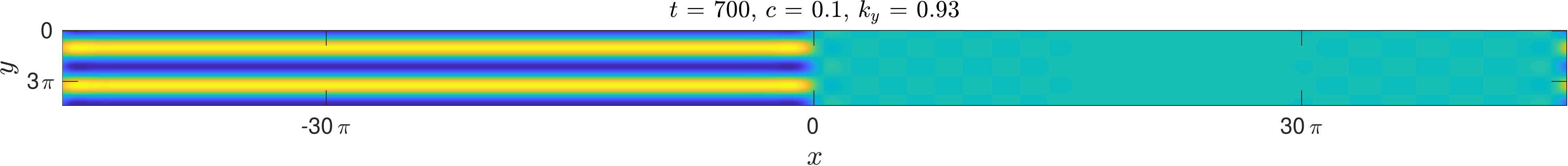}\hfill\raisebox{0.33in}{(4)}
\end{minipage}\\
\begin{minipage}{0.08\textwidth}
\includegraphics[width=0.89\textwidth]{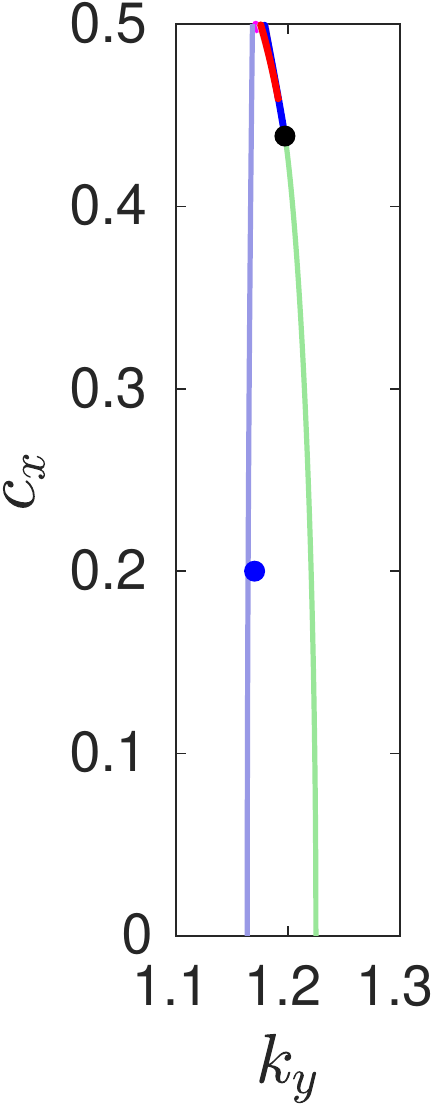}
\end{minipage}\hfill
\begin{minipage}{0.89\textwidth}\includegraphics[width=0.89\textwidth]{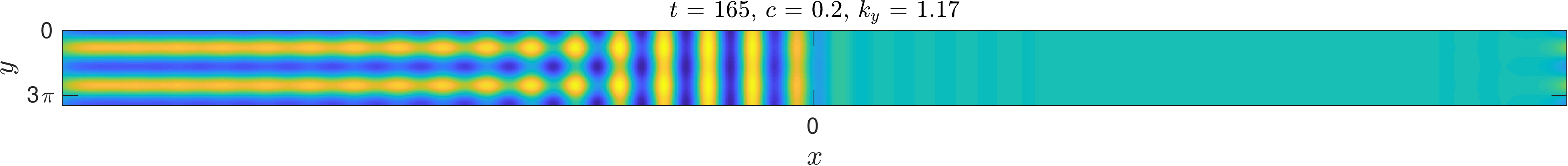}\hfill\raisebox{0.33in}{(5)}\\
\includegraphics[width=0.89\textwidth]{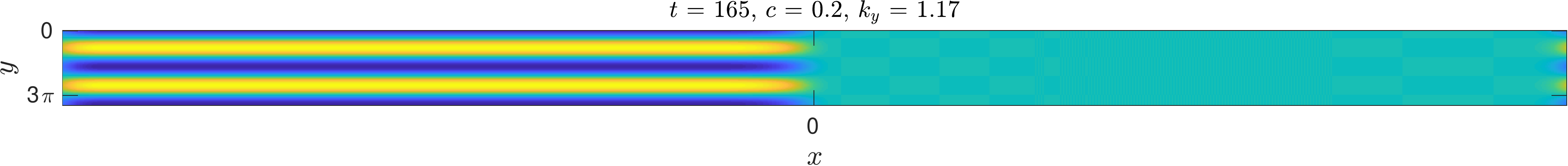}\hfill\raisebox{0.33in}{(6)}
\end{minipage}\\
\begin{minipage}{0.08\textwidth}
\includegraphics[width=0.89\textwidth]{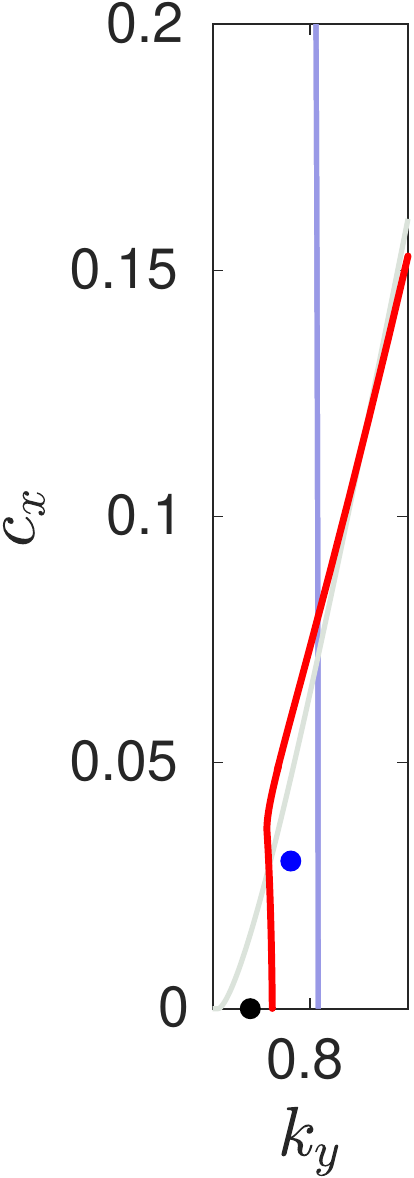}
\end{minipage}\hfill
\begin{minipage}{0.89\textwidth}\includegraphics[width=0.89\textwidth]{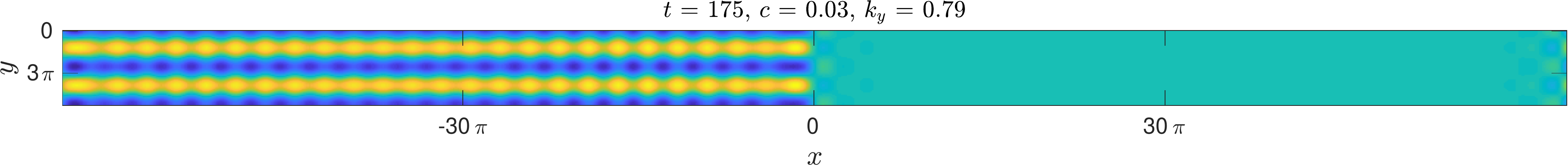}\hfill\raisebox{0.33in}{(7)}\\
\includegraphics[width=0.89\textwidth]{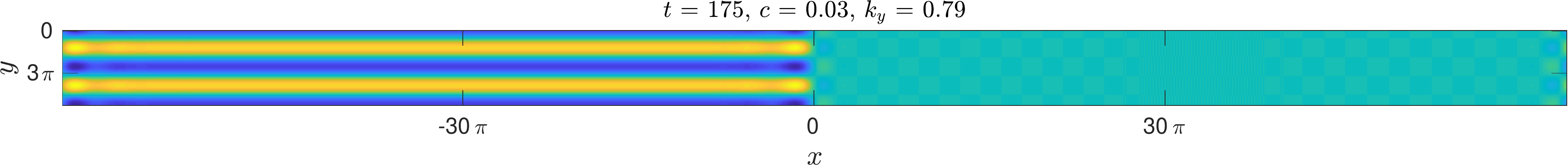}\hfill\raisebox{0.33in}{(8)}
\end{minipage}\\
\begin{minipage}{0.08\textwidth}
\includegraphics[width=0.89\textwidth]{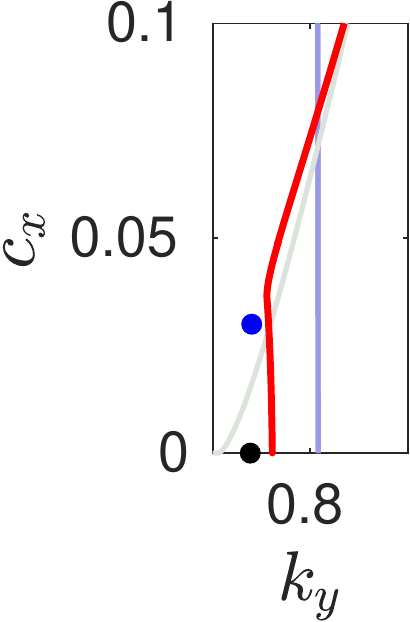}
\end{minipage}\hfill
\begin{minipage}{0.89\textwidth}
\includegraphics[width=0.89\textwidth]{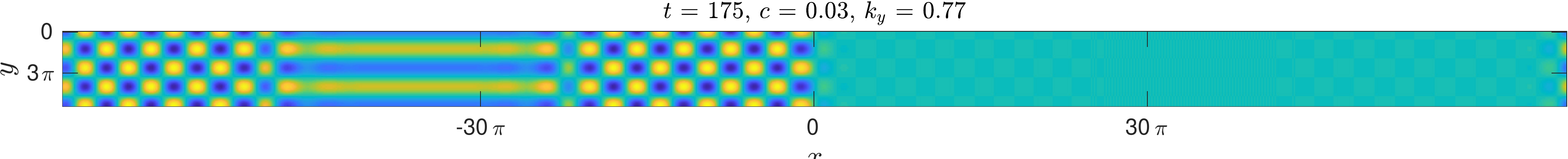}\hfill\raisebox{0.33in}{(9)}
\end{minipage}\\
\begin{minipage}{0.08\textwidth}
\includegraphics[width=0.89\textwidth]{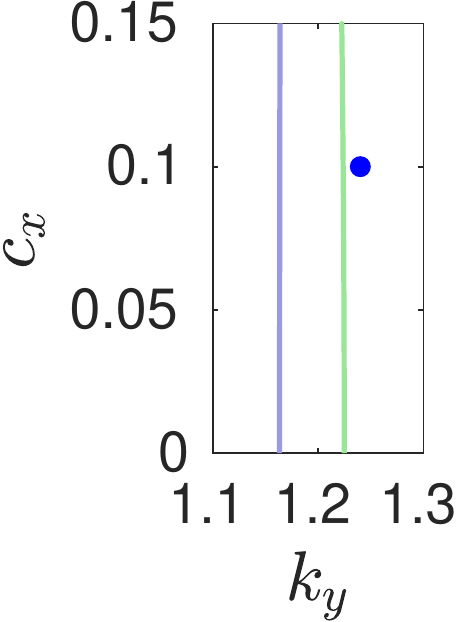}
\end{minipage}\hfill
\begin{minipage}{0.89\textwidth}
\includegraphics[width=0.89\textwidth]{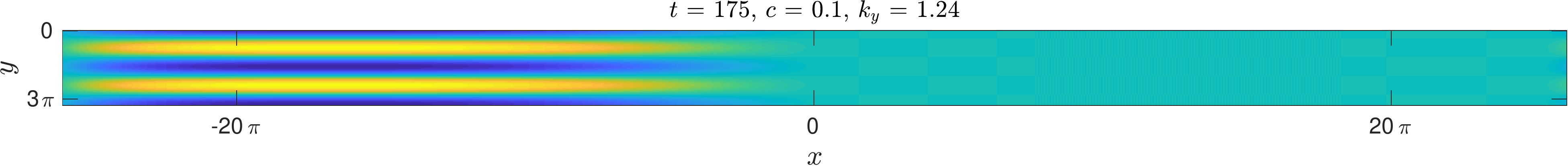}\hfill\raisebox{0.33in}{(10)}
\end{minipage}

\caption{From top to bottom: saddle-node near $k_y=1$, with random perturbations (1) and with even in $y$ random perturbations (2); crossing the zigzag boundary near $c=0.1$, zigzag modulations spread into the domain (3); the zigzag instability is suppressed for even initial conditions (4); parallel stripes just past the upper cross-roll instability at $c=0.2$ (5) are suppressed for odd initial conditions (6); 
parallel stripes past the  lower cross-roll boundary,  visible for even initial conditions (7),  
suppressed for even-odd initial conditions (8), eventually even-odd destabilized past the saddle-node (9); detachment of even-odd perturbations of stripes past the spreading speed of perpendicular stripes (10).
}\label{f:sn}
\end{figure}

\paragraph{Stability and instability of perpendicular stripes --- pitchfork to oblique stripes} 
The linearization at the quenched perpendicular stripes $A^*(x)\in\R$ exhibits a bifurcation in the direction of complex $A$, breaking the reflection symmetry in $y$. We analyzed this bifurcation by studying the linearized operator in the direction of $A_\mathrm{i}:= \mathrm{Im}\, A$, 
\[
\mathcal{L}_\mathrm{i} A_\mathrm{i}=\left[-(\partial_{xx}+1-k_y^2)^2+\rr-3(A^*)^2\right] A_\mathrm{i}.
\]
This operator possesses essential spectrum up to the origin due to the marginal stability of stripes in the far-field. For positive speeds, the essential spectrum can however be pushed into the left half plane using exponential weights $\|A\|_\delta=\sup_x (1+\exp(-\delta x))^{-1} |A(x)|$, see  \cite{fiedlerscheel}, allowing us to track possible instabilities by eigenvalues emerging near $\lambda=0$. 
By gauge invariance (from $y$-shift symmetry), the operator possesses a zero eigenvalue in this exponentially weighted space, given simply by $A^*$. Close to the saddle-node, an eigenvalue crosses the origin. At criticality, the zero eigenvalue is a Jordan block of length two and as expected the generalized eigenvector exhibits linear growth at $x\to -\infty$. Spatial dynamics methods as in \cite{ssessfront,sscl} should allow one to confirm the observed bifurcation towards oblique stripes at this parameter value. We show numerical evidence for this bifurcation in the numerical study of the full Swift-Hohenberg equation in \S\ref{s:6}. 

The corresponding bifurcation curves are shown as the magenta curve in Figure \ref{f:perp-bif}. Bifurcations happen very close to the saddle-node bifurcation except in a 	region  $k_y=1.12\pm 0.04$, where the pitchfork is located on the unstable branch just past the saddle-node bifurcation. In this region, perpendicular and oblique quenched stripes coexist, although the oblique stripes, bifurcating from the unstable branch, are unstable against the saddle-node eigenfunction. 

\paragraph{Stability and instability of perpendicular stripes -- zigzag and cross-roll spreading}

We recall from the earlier discussion in the Cahn-Hilliard equation that perpendicular stripes are unstable for small speeds due to an absolute zigzag instability. In the amplitude equation, this instability boundary can be computed from the linear spreading speed associated with the zigzag-instability. Linearizing the amplitude equations at a perpendicular stripe, we find a complex fourth-order linear equation
\begin{align*}
A_{\mathrm{r},t}&=[-(\partial_{xx}+1-k_y^2)^2 +\mu - 9 r^2(k_y)]A_\mathrm{r},\\
A_{\mathrm{i},t}&=[-(\partial_{xx}+1-k_y^2)^2 +\mu - 3 r^2(k_y)]A_\mathrm{i}.
\end{align*}
The zigzag instability is visible in the imaginary part. Computing the linear spreading speed of instabilities in this equation \cite{holz} one finds 
\begin{equation}\label{e:zzlin}
c_\mathrm{zz}(k_y)=\frac{4}{\sqrt{27}}
\left(2+\sqrt{7}\right) \sqrt{\sqrt{7}-1}  \left(1-k_y^2\right)^{3/2}
,\end{equation}
which is in fact independent of $\mu$. 

The resulting stability boundary is shown in Figure \ref{f:perp-bif}. It intersects the saddle-node bifurcation curve near $k_y=0.920,\, c_x=0.278$, thus marking the smallest wavenumber for which there exists quenching rates for which straight perpendicular stripes can be observed.

In a domain of half the width $y\in(0,\pi)$ with,  Neumann boundary conditions, the zigzag instability is suppressed and perpendicular stripes are stable for smaller values of $k_y$. 	A subsequent instability is visible only in the coupled mode system for $A\rme^{\rmi y}+B\rme^{\rmi x}+c.c$, 
\begin{align*}
A_t&=-(\partial_{xx}+1-k_y^2)^2 A + \mu A - 3A(|A|^2+2|B|^2) +c_x A_x\\
B_t&=4 B_{xx}+ \mu B - 3B(|B|^2+2|A|^2)+c_x B_x.
\end{align*}
Linearizing at $A\equiv \sqrt{\mu-(1-k_y^2)}$, we find a linear operator $4\partial_{xx}+\mu - 6 (\mu-(1-k_y^2))$ which becomes unstable outside of the interval 
\[
\left(k_{y_\mathrm{cr}}^-,k_{y_\mathrm{cr}}^+\right)=\left(\sqrt{1-\sqrt{\mu/2}},\sqrt{1+\sqrt{\mu/2}}\right)\sim (0.80402,1.16342)\quad \text{for } \mu=1/4.
\]
This instability is known as the cross-roll instability. One readily finds an associated spreading speed \cite{holz},
\begin{equation}\label{e:crlin}
c_x^\mathrm{cr}(k_y;\mu)=4 \sqrt{2\left(1-k_y^2\right)^2-\mu}.
\end{equation}
The upper boundary intersects the pitchfork bifurcation curve near $k_y=1.168,\,c_x=0.504$ thus marking the largest wavenumber for which there exist quenching rates for which straight perpendicular stripes can be observed.
The resulting stability boundaries are again depicted in Figure \ref{f:perp-bif}.  Phenomena associated with this instability are shown in Figure \ref{f:sn}\footnote{See also files sh *.m4v in the supplementary materials for movies of solutions}.

Both zigzag and cross-roll instability are supercritical in an appropriate sense and linear spreading speeds give accurate predictions for the associated absolute instability. Both can be eliminated from the Swift-Hohenberg equation in the strip by restricting to even-odd initial conditions.

\paragraph{Blocking perpendicular stripes by amplitude modulations}
The saddle-node curve in Figure \ref{f:perp-bif} terminates at $c_x=0$ for small $k_y$, after reaching a minimal value of $k_y$ for finite $c_x$.  There do not appear to be analytic predictions for either minimal $k_y$-values or the limit at $c_x=0$ as those appear to be global bifurcations even at these limiting points. The fact that the existence of quenched perpendicular stripes is limited can however be understood near $c_x=0$ from an amplitude modulational instability. In fact, inspecting the real amplitude equation \eqref{e:nwsm} with $\rr(x)\equiv \mu$ at a perpendicular stripe, we find, after shifting $A=r(k_y)+u$ so that the perpendicular stripes correspond to $u=0$, a Swift-Hohenberg equation with quadratic nonlinearity, 
\begin{equation}\label{e:shnws}
u_t=-(\partial_{xx}+1-k_y^2)^2+ \mu_\mathrm{eff} u + \gamma u^2 - 3 u^3,\qquad \mu_\mathrm{eff}=\mu - 9r^2(k_y),\ \gamma=-9 r(k_y).
\end{equation}
This equation undergoes a weakly subcritical pattern-forming instability at $k_y=k_{y,\mathrm{a}}$, with selected wavenumber $\ell_\mathrm{a}$, with 
\begin{equation}\label{e:turing}
k_{y,\mathrm{a}} =\sqrt{1-\sqrt{2\mu/3}} , \qquad  \ell_\mathrm{a}=\sqrt[4]{2\mu/3},\qquad (k_{y,\mathrm{a}},\ell_\mathrm{a})\sim (0.7692, 0.6389) \text{ at } \mu=1/4.;
\end{equation}
see Figure \ref{f:perp-bif} for the location of the amplitude modulational instability relative to the saddle-node. Phenomenologically, perpendicular stripes develop amplitude modulations in this instability.

In the quenched problem, $c_x=0$, the perpendicular stripes are hyperbolic equilibria prior to this instability, $k_y>k_{y,\mathrm{a}}$. At the instability, they undergo a Hamiltonian Hopf bifurcation with normal form given by a subcritical Ginzburg-Landau equation, $C_{XX}\pm C+C|C|^2=0$, after suitable scalings. Stable and unstable manifolds of the origin therefore are compact subsets of a small neighborhood of the origin, making an intersection with the stable manifold of $A=0$ at $x=+\infty$ impossible close to criticality.

\subsection{Detaching all stripes}\label{s:5.2}

Increasing the speed further, one eventually sees all stripes detach: the trivial state occupies an increasingly large region in $x<0$, behind the quenching line, as $c_x$ increases, until this region eventually expands linearly in time. The quenching process at this point ceases creating stripes and instead creates an unstable state, which is invaded by a free invasion front in a region well separated from the quenching line. We briefly present predictions for this detachment process and, in particular, consequences for stripe orientation. 

\paragraph{Linear spreading speeds}

Disturbances in the linearized Swift-Hohenberg equation with simple $y$-dependence of the form $\rme^{\rmi k_y y}$ solve
\[
 u_t=-(\partial_{xx}+1-k_y^2)^2 u+\mu u.
\]
Compactly supported initial conditions to this equation spread with the spreading speed 
\begin{equation}\label{e:clinsh}
 c_\mathrm{lin}(k_y)=\left\{
\begin{array}{ll}
\frac{4 \left(2-2 k_y^2+\sqrt{1-2 k_y^2+k_y^4+6 \mu}\right) \sqrt{-1+k_y^2+\sqrt{1-2 k_y^2+k_y^4+6 \mu}}}{3
\sqrt{3}},& 0<k_y<  \sqrt{\frac{2+ \sqrt{3\mu}}{2}}\\
\frac{4 \sqrt{-1+k_y^2-\sqrt{4-8 k_y^2+4 k_y^4-3 \mu}} \left(-2+2 k_y^2+
\sqrt{4-8 k_y^2+4 k_y^4-3 \mu}\right)}{3\sqrt{3}},& \sqrt{\frac{ 2+\sqrt{3\mu}}{2}}<k_y<\sqrt{1+\sqrt{\mu}}.
 \end{array}\right.
\end{equation}
In a frame moving with this speed, one observes oscillations with frequencies $\omega_\mathrm{lin}(k_y)$ which are in $1:1$-resonance with patterns formed at wavenumbers $\omega_\mathrm{lin}(k_y)=c_\mathrm{lin}(k_y)k_\mathrm{lin}(k_y)$, with 
\begin{equation}\label{e:klinsh}
 k_\mathrm{lin}(k_y)=\left\{
 \begin{array}{ll}
 \frac{3 \left(3-3 k_y^2+\sqrt{1-2 k_y^2+k_y^4+6 \mu}\right)^{3/2}}{8 \left(2-2 k_y^2+\sqrt{1-2 k_y^2+k_y^4+6
\mu}\right)},& 0<k_y< \sqrt{\frac{2+ \sqrt{3\mu}}{2}}\\
0, &\sqrt{\frac{2+ \sqrt{3\mu}}{2}}<k_y<\sqrt{1+\sqrt{\mu}}.
 \end{array}\right.
\end{equation}
We refer to \cite{holz} for background and in particular for results that demonstrate that this speed is more generally non-increasing in $|k_y|$. 

The values for $\mu=1/4$ are included in Figure \ref{f:perp-bif} as the upper boundary. The cross-over point $k_y=\sqrt{\frac{2+ \sqrt{3\mu}}{2}}$ distinguishes between $k_x=0$, perpendicular stripes, and $k_x>0$, oblique stripes, selected by the spreading in the leading edge. 

It is worth noticing that, due to the monotonicity $c_x(k_y)\searrow $ in $k_y>0$, parallel rolls always spread fastest and generic initial conditions in a system without (!) parameter step will lead to parallel stripes. 

\paragraph{Quenched stripes near the linear spreading speed}

Nevertheless, we observed oblique stripes in the quenched system and values of $c_x$  up to the linear spreading speed for values of $k_y\leq 0.95$. For larger values of $k_y\geq 1$, we noticed that oblique stripes selected in the quenching process destabilize against parallel stripes well before the linear spreading speed of oblique stripes in what appears to be related to the cross-roll instability. In fact, for $k_y$ close to the cross-over, the perpendicular stripes are unstable against the cross-roll instability and, by continuity of spreading speeds \cite{holz}, oblique stripes would be unstable against such perturbations as well for values near the cross-over point.  We did not attempt a more comprehensive study of stability of oblique stripes far from the transition near perpendicular stripes. 

Near the detachment,  the results in \cite{gs1} establish corrections to the wavenumber based on absolute spectra. Based on these predictions, one concludes in this regime near the linear spreading speed, that the transition from oblique to perpendicular stripes occurs at leading order when the absolute spectrum, computed in the co-moving frame, 
possesses a triple point at $\lambda=0$. Some tedious algebra, solving
\[
 \left\{\begin{array}{l}
  d(0,\nu;c_x, k_y)=0,\\d(0,\nu+\rmi\ell;c_x, k_y)=0,
 \end{array}\right.\qquad \text{ with\ \  } d(\lambda,\nu;c_x,k_y)=-(\nu^2+1-k_y^2)^2+\mu+c_x\nu-\lambda,
\]
as one real and one complex equation for the three real variables  $(\nu,\ell,c_x)$ with parameter $k_y$, leads to the location of this triple point at 
\begin{align}
 c_\mathrm{tr}(k_y)&=\frac{4 \left(2 \sqrt{3}(1-k_y^2)+5 \sqrt{-4(1-k_y^2)^2+7 \mu}\right) \sqrt{-3(1- k_y^2)+\sqrt{-12(1-k_y^2)^2+21 \mu}}}{21 \sqrt{7}}\notag\\
&= 
2
\sqrt{2}\left(\frac{\mu}{3}\right)^{3/4}+
4 \sqrt{2+ \sqrt{3\mu}}\left(\frac{\mu}{3}\right)^{1/4}\Delta k_y+\rmO\left(\Delta k_y^2\right),\label{e:triple}
\end{align}
for $\Delta k_y=k_y-\sqrt{\frac{ 2+\sqrt{3\mu}}{2}}\lesssim 0$; see Figure \ref{f:perp-bif} for a comparison between these asymptotics, the saddle-node of perpendicular stripes, and the pitchfork bifurcation of oblique stripes.

%
%
%
%
%
%
%
%
%

\section{Back to Swift-Hohenberg: organizing stripe formation in the moduli space}\label{s:6}
We present a conceptually simple object, the moduli space, that captures much of the phenomena presented in this work, in particular much of the results from direct simulations as summarized in the parameter landscape of stripe formation, Figure \ref{f:2a}. The moduli space is a variety in $(k_x,k_y,c_x)$-space, which encodes stripe formation at the rate $c_x$ with stripes of wave vector $(k_x,k_y)$  in the wake. We present a more precise definition and describe coarse features of this object in \S\ref{s:6.1}, followed by a more detailed description of numerical strategies used for computing this variety in \S\ref{s:6.2}. The remaining paragraphs zoom in on some of the finer structures of the variety, relating to oblique detachment \S\ref{s:6.3}, perpendicular detachment, \S\ref{s:6.4}, and the interaction of the two detachments, \S\ref{s:6.5}.

\subsection{The moduli space} \label{s:6.1}

Solutions that form stripes in the wake of the quenching step can be stationary in an appropriately co-moving frame, hence solving the elliptic traveling-wave equation
\begin{align}\label{e:shtwc}
0 &=-(\partial_{xx}+k_y^2 \partial_{yy}+1)^2u + \mu(x) u - u^3 + c_x( u_x + k_x u_y),\\
0 &= \lim_{x\rightarrow-\infty} \left(u(x,y) - u_\mathrm{p}(k_x x + y;k)\right),\qquad 0= \lim_{x\rightarrow\infty} u(x,y),\qquad u(x,y)=u(x,y+2\pi).\notag
\end{align}
Note  that the vertical velocity satisfies $c_y = c_x k_x$ since asymptotic patterns are stationary in a stationary frame. We emphasize that this system only captures the \emph{simplest} solutions that form stripes with a given wave vector --- actual stripe formation could possess periodic or even more complex temporal modulations. The system \eqref{e:shtwc} comes with three parameters $(k_x,k_y,c_x)$, and we define the \emph{moduli space} as
\[
 \mathcal{M}=\{(k_x,k_y,c_x)|\text{ ex. solution to \eqref{e:shtwc} } \}.
\]
One can see using Fredholm theory that this moduli space would typically be a two-dimen\-sional surface, except at singularities (or bifurcation points). Practically, this surface encodes the regimes of existence for various orientations of striped patterns formed behind the quenching line. When paired with stability information, it gives a recipe for how to select various orientations of stripes through quenching rates and lateral aspect ratios.

\begin{figure}[h!]\centering
\includegraphics[width=.49\textwidth,trim={0.8cm 0 1.5cm 0},clip]{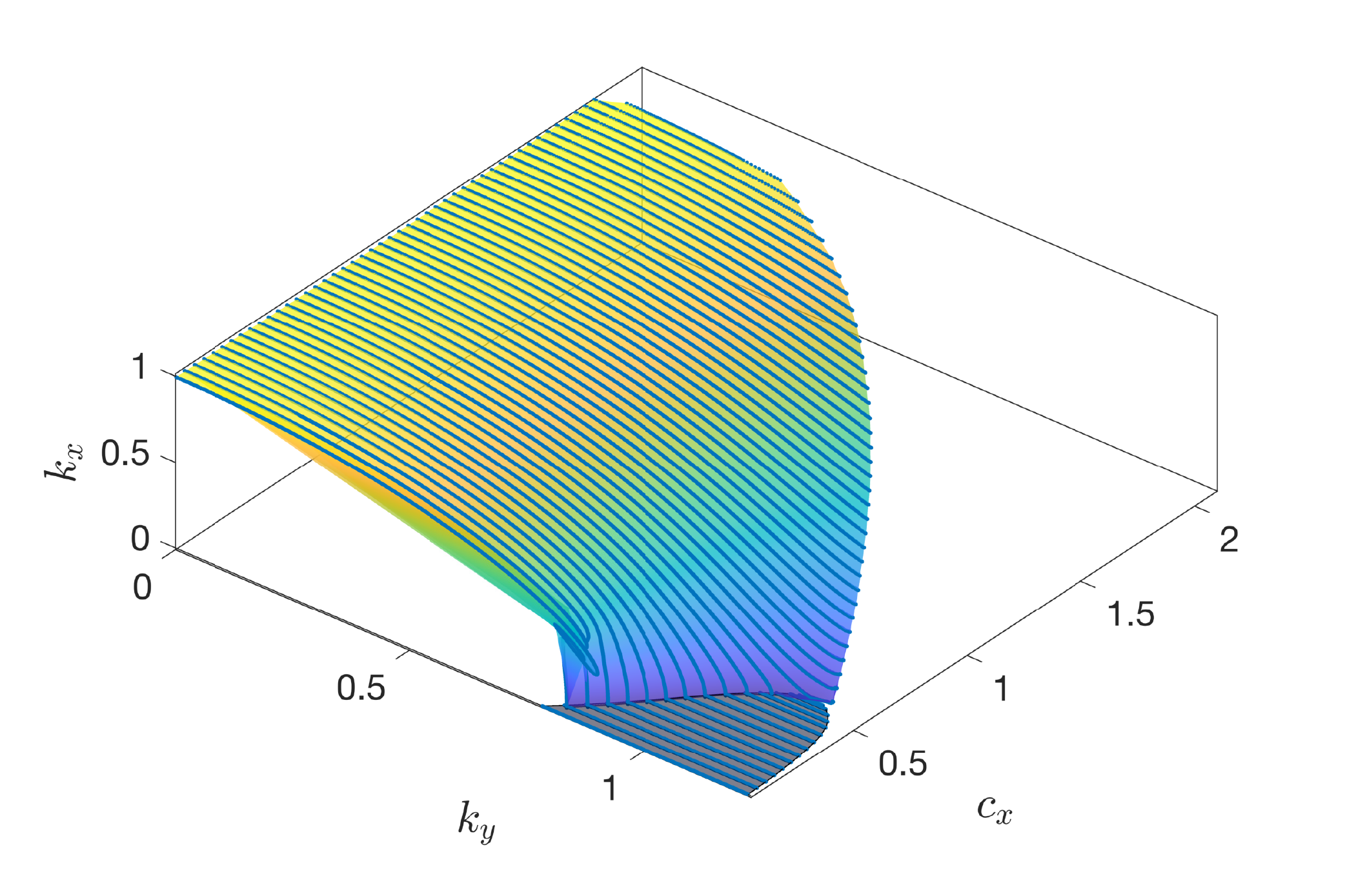}\hfill
\raisebox{0.2in}{\includegraphics[width=.4\textwidth,trim={0.8cm 0 1.5cm 0},clip]{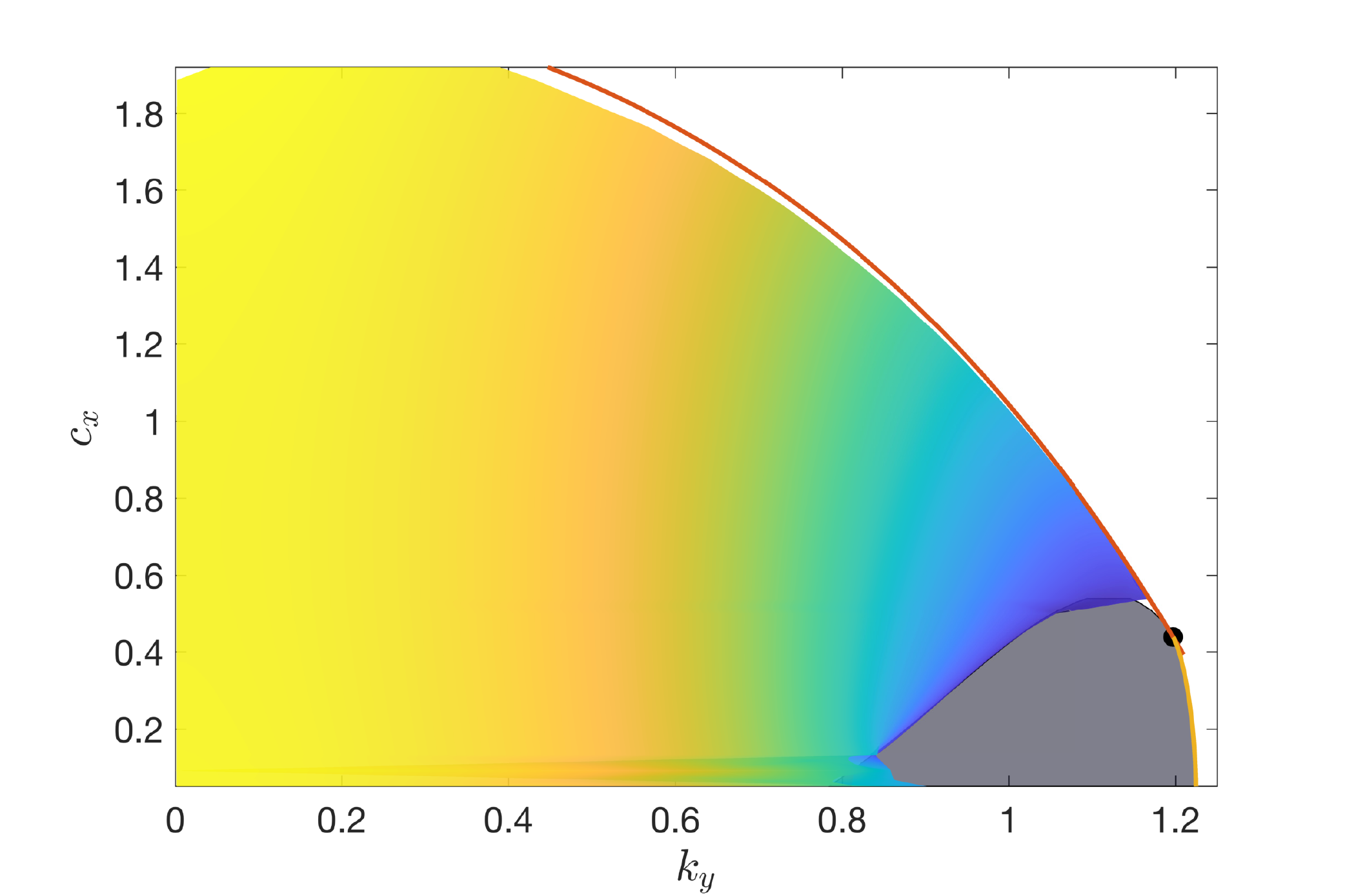}}\quad$ $\\
\includegraphics[width=0.49\textwidth,trim = {.8cm .5cm .8cm 0.5cm},clip]{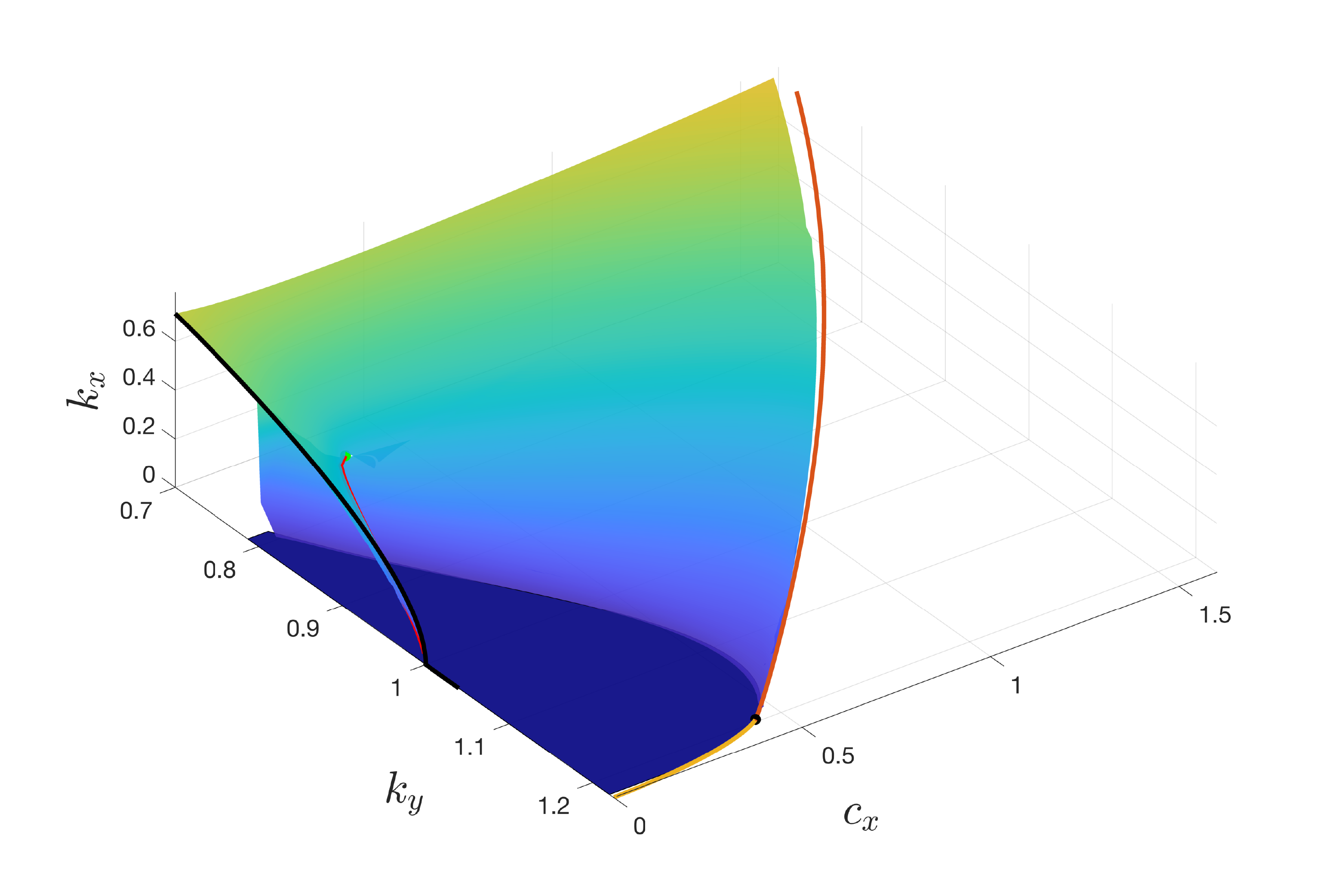}\hfill
\includegraphics[width=0.49\textwidth,trim = {.8cm .5cm .8cm 0.5cm},clip]{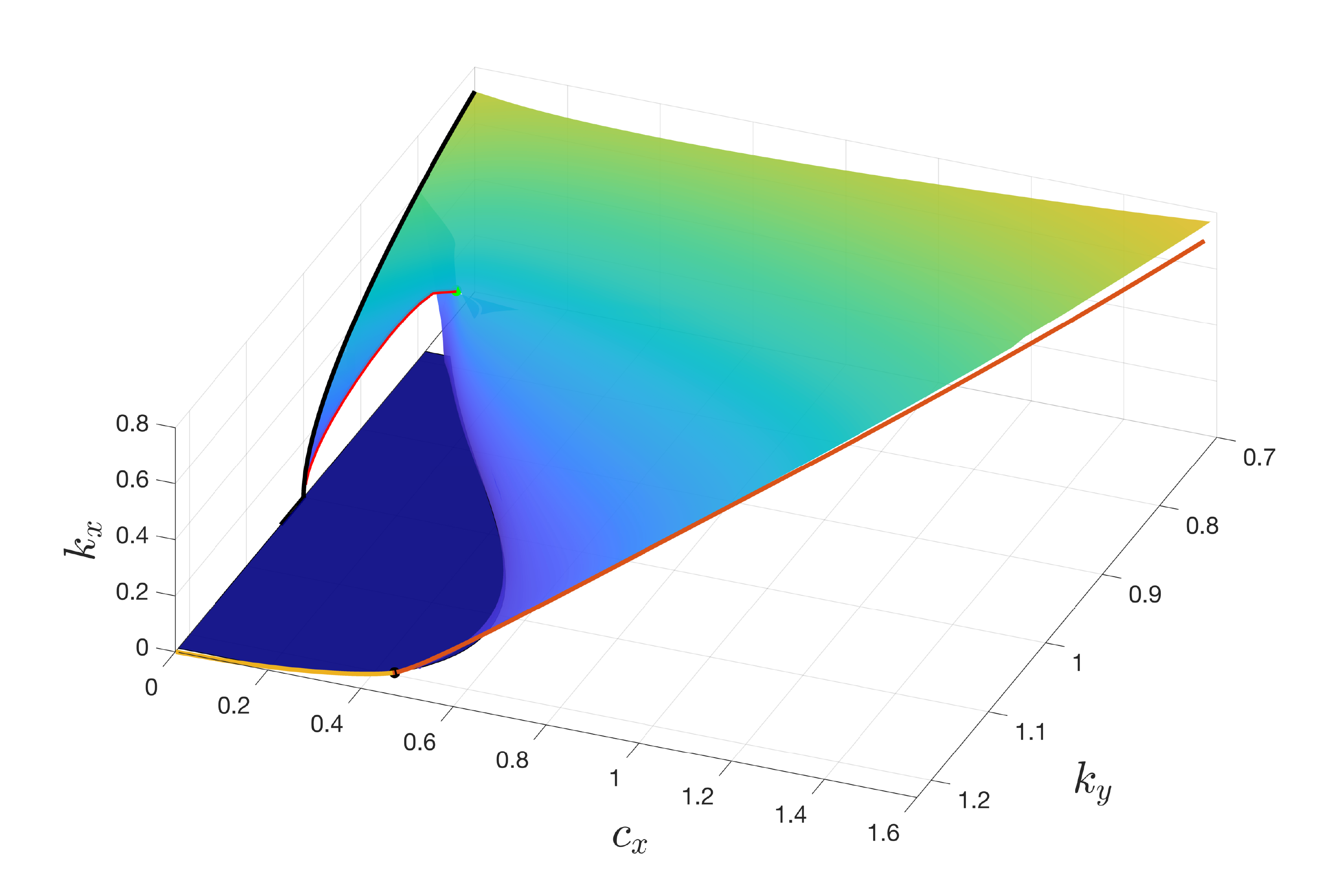}
\caption{Top row: Coarse view of the moduli space computed via continuation in $k_y$ in 3d and top-down view, with gray areas representing perpendicular stripes, color shading $k_x$ which roughly encodes the angle from perpendicular (dark blue) to parallel (yellow). Bottom row: Two views of a zoom into the region $k_y\sim 1$, showing in particular the touch-down near $c_x=0, k_y=k_\mathrm{zz}, k_x=0$ via the kink-dragging bubble which resembles a delicate arch in this view, the perpendicular stripe detachment at finite $c_x^\mathrm{psn}$  where a wing-like surface lifts up above the plane, and the hyperbolic catastrophe where delicate arch and wing meet.  Perpendicular stripes in dark blue, $c_x=0$ as black line, detachment of oblique stripes $c_x^\mathrm{osn}$ at red curves via kink shedding and final detachment.}\label{f:end}
\end{figure}

The moduli space, computed using numerical continuation that we shall explain in the next section, is shown in  Figure \ref{f:end}. The two key organizing elements, the oblique detachment and the perpendicular detachment are both clearly recognizable as a delicate arch near $c_x=0$ and a lift-off to a wing-shaped structure at larger $c_x$. Both collide in the hyperbolic catastrophe\footnote{See  movie \textsc{moduli.m4v} in supplementary materials for an animated $360^\circ$ tour of the moduli space.}. In the following, we explain in more detail the information contained in this surface, how it was obtained, and how it relates to our previous analysis. We encourage, however, at this point, a comparison with the coarse information from the parameter landscape in Figure \ref{f:2a}.

\subsection{Farfield-core continuation and the moduli space}\label{s:6.2}
To explore the moduli space, we use the general approach outlined in \cite{lloydscheel}, as well as in \S\ref{s:3.2}, where heteroclinic profiles are decomposed into a pure asymptotic state cutoff away from negative infinity, and an exponentially localized perturbation which glues the asymptotic state to another asymptotic state at positive infinity. We use a cutoff function $\chi_-(x)$ supported on $x\leq d+1$ with $\chi_-\equiv 1$ on $x\leq d$ to decompose solutions of \eqref{e:shtwc} as
\[
u(x,y) = w(x,y) + \chi_-(x) u_\mathrm{p}(k_x x + y;k);
\]
see Figure \ref{fig:coreff} for a depiction of the various solution components.

\begin{figure} \centering
\hspace*{-0.5in}\includegraphics[width=0.8\linewidth,trim={2.0cm 0.0cm 0.0cm 0.5cm},clip]{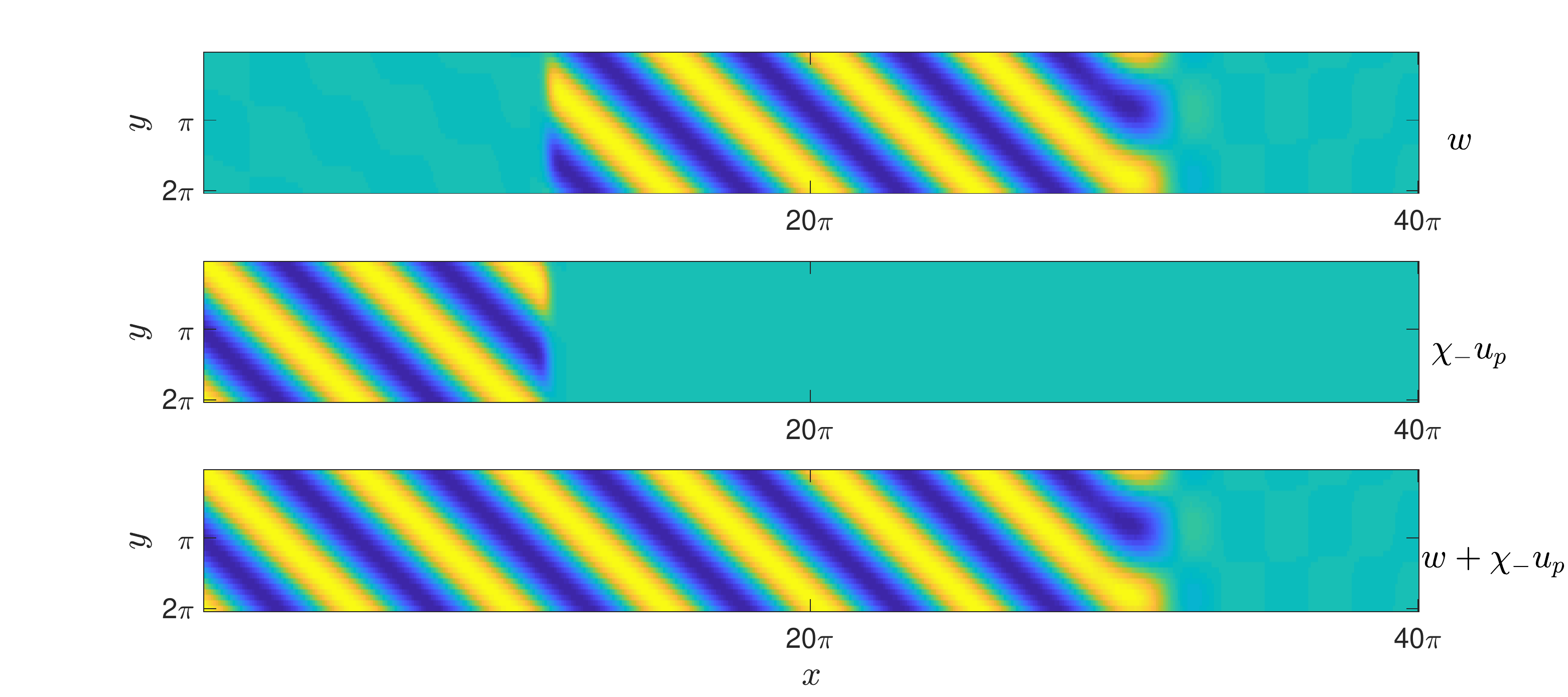}
\caption{Example of a farfield-core decomposition of a traveling wave solution $u = w + \chi_- u_\mathrm{p}$, quenching interface at $x  = 8L_x/10 \sim 100$ and farfield cutoff at $x = 3L_x/10 \sim 37$. }\label{fig:coreff}
\end{figure}

We insert this Ansatz into \eqref{e:shtwc} and use the fact that the stripe solution $u_\mathrm{p}$ is an easily available solution for $\rr\equiv \mu$ to obtain the following nonlinear problem for $(w;k_x,k_y,c_x)$ on a truncated domain
\begin{align}
\mathcal{L} \left( w+ \chi_- u_\mathrm{p} \right) -  \left( w+ \chi_- u_\mathrm{p} \right)^3 &=0,\qquad\quad (x,y)\in(0,L_x)\times(0,2\pi)\\
w = w_{xx} &=0,\qquad\quad(x,y)\in \{0,L_x\}\times(0,2\pi)\\
\partial_y^j w(x,0) - \partial_y^j w(x,2\pi) &=0,\qquad \quad x\in(0,L_x),j = 0,...,3,\\
\int_{x=0}^{2\pi}\int_{y=0}^{2\pi} u'_\mathrm{p}(k_x x + y;k) w(x,y) \rmd y\, \rmd x &=0,\quad\label{e:ffcpc}\\
-\left(k^2\frac{d^2}{d\xi^2} + 1\right)^2u_\mathrm{p} + \mu u_\mathrm{p} - u_\mathrm{p}^3 &=0,\qquad\quad \xi\in (0,2\pi)\\
\frac{d^j}{d\xi^j} u_\mathrm{p}(0) - \frac{d^j}{d\xi^j} u_\mathrm{p}(2\pi)&=0,\quad\qquad j = 0,...,3.
\end{align}
where $\mathcal{L} = -(\partial_{xx}+k_y^2 \partial_{yy}+1)^2 + \rr(x) + c_x( \partial_x + k_x \partial_y)$, and $k^2 = k_x^2 + k_y^2$, and where the parameter jump is now located at $L_\mathrm{q}=8L_x/10$, $\rr(x)\sim \,-\mu\sign(x-L_\mathrm{q})$. 

This decomposition suppresses the continuous family of neutral modes, arising from the asymptotic periodic pattern, in the spectrum of the linearization about a generic traveling wave solution $u(x,y)$ of \eqref{e:shtwc}.   In the $x$-unbounded domain, one imposes exponential weights on the perturbation $w$ to obtain a Fredholm index -1 linearization in $w$ which, after appending the wavenumber parameter $k_x$, yields a Fredholm index 0 problem, with trivial kernel whenever the derivative with respect to $k_x$ does not belong to the range.     
 
Truncating to $x\in [0,L_x]$, we impose Dirichlet boundary conditions in $x$ which are readily found to be transverse to the unstable subspace of the asymptotic stripes and constant state at $x=\pm\infty$. This implies that the truncated problem has the correct Fredholm index in $(w,k_x,k_y,c_x)$ and the perturbation $w$ will be exponentially localized in the domain for generic parameter values. This also implies that truncated solutions converge to the full modulated traveling wave as $L_x \rightarrow \infty.$ Note also that since the quenching interface destroys $x$-translational invariance, we need only one phase condition \eqref{e:ffcpc} to eliminate the multiplicity from the translational mode $\partial_y u_\mathrm{p}$ and fix the vertical phase of the solution. See \cite{lloydscheel,morrissey} for more details about this approach.




Using the above formulation, we implemented an arc-length continuation algorithm in \textsc{matlab2018a}, solving for $(w,k_x)$ and continuing in either $c$ or $k_y$ with the other fixed. We roughly followed the approach outlined in \cite[\S 3]{lloydscheel} and refer the reader there for more details on the implementation. To discretize the problem, we used fourth-order finite differences in $x$ with $L_x = 40\pi$ and approximately 500 grid points. We used a pseudo-spectral discretization in $y$ with 26 collocation points and the far-field periodic patterns $u_\mathrm{p}$ were also computed on a periodic domain $\zeta\in [0,2\pi)$ using a Fourier pseudo-spectral method with 26 collocation points. The quenching interface was placed at $x = 8L_x/10$ and the cutoff-interface was placed at $d = 3L_x/10$; see Figure \ref{fig:coreff} for a depiction of the computational domain and the solution decomposition.  The Jacobian of the discretized system is formed explicitly in $w$ while the derivatives in parameters were approximated using a second-order finite difference. We used the trust-region algorithm in \textsc{matlab}'s \textsc{fsolve} to perform the nonlinear Newton iterations.  Our initial guess for the nonlinear solver consisted of a piece-wise constant stripe solution, rotated to have a specific wavenumber, and cutoff at the quenching interface.  Throughout all of this section we used the onset parameter $\mu = 0.25.$


%
%

We explored the solution space starting from initial guesses along the line $(k_x,k_y) = (1,0)$, keeping $c\in(0,c_\mathrm{lin}(0))$ fixed and continuing in $k_y$ to track oblique solutions as they continuously perturbed from parallel stripes, $k_y = 0$.  For large $k_y$ curves in parameter space either run into the detachment curve predicted by $c_\mathrm{lin}(k_y)$ (roughly in the region $c_x\geq .55$) or the solution transitions, via a pitchfork bifurcation, through the family of perpendicular stripes, to the opposite orientation of stripes with $k_x <0$.  In the former cases, the core solution $w$ loses localization, bleeding into the far-field domain $x\in (0,d)$ and the $L^2$-norm of the full solution $u$ decays to zero as $c_x \rightarrow c_\mathrm{lin}(k_y)$.  

To explore the perpendicular stripe region we started from initial patterns along the line $(k_x,k_y) = (0,1.12)$ for a range of $c_x$.  Continuing in increasing $k_y$ for $c_x < c_\mathrm{lin}\left(\sqrt{\frac{2+\sqrt{3\mu}}{2}}\right) \sim 0.438691 $, the core solution once again loses localization as the detachment curve $c_\mathrm{lin}(k_y)$ is approached.  For larger $c_x\sim 0.5$, continuation in $k_y$ gives an isola bounded by the two saddle-node curves predicted by the Newell-Whitehead-Segel equation; see Figure \ref{f:perp-bif} and \ref{fig:sh-re-comp}.

Figure \ref{f:end} combines these two sets of continuations, oblique/parallel and perpendicular striped, to give an overview of the moduli surface and we find good agreement for all orientations of stripes between the measured detachment points and predictions from the linear spreading speed \eqref{e:clinsh}; see Figure \ref{fig:modsl} for slices of the moduli space for select values of $c_x$ with corresponding solution profiles\footnote{See \textsc{movie\textunderscore c2.948980e-01.m4v} and \textsc{movie\textunderscore c9.081633e-02.m4v} in supplementary materials for movies of these solutions.}.

\begin{figure}\centering
\includegraphics[width=.49\textwidth,trim={0.75cm 0 1.0cm 0},clip]{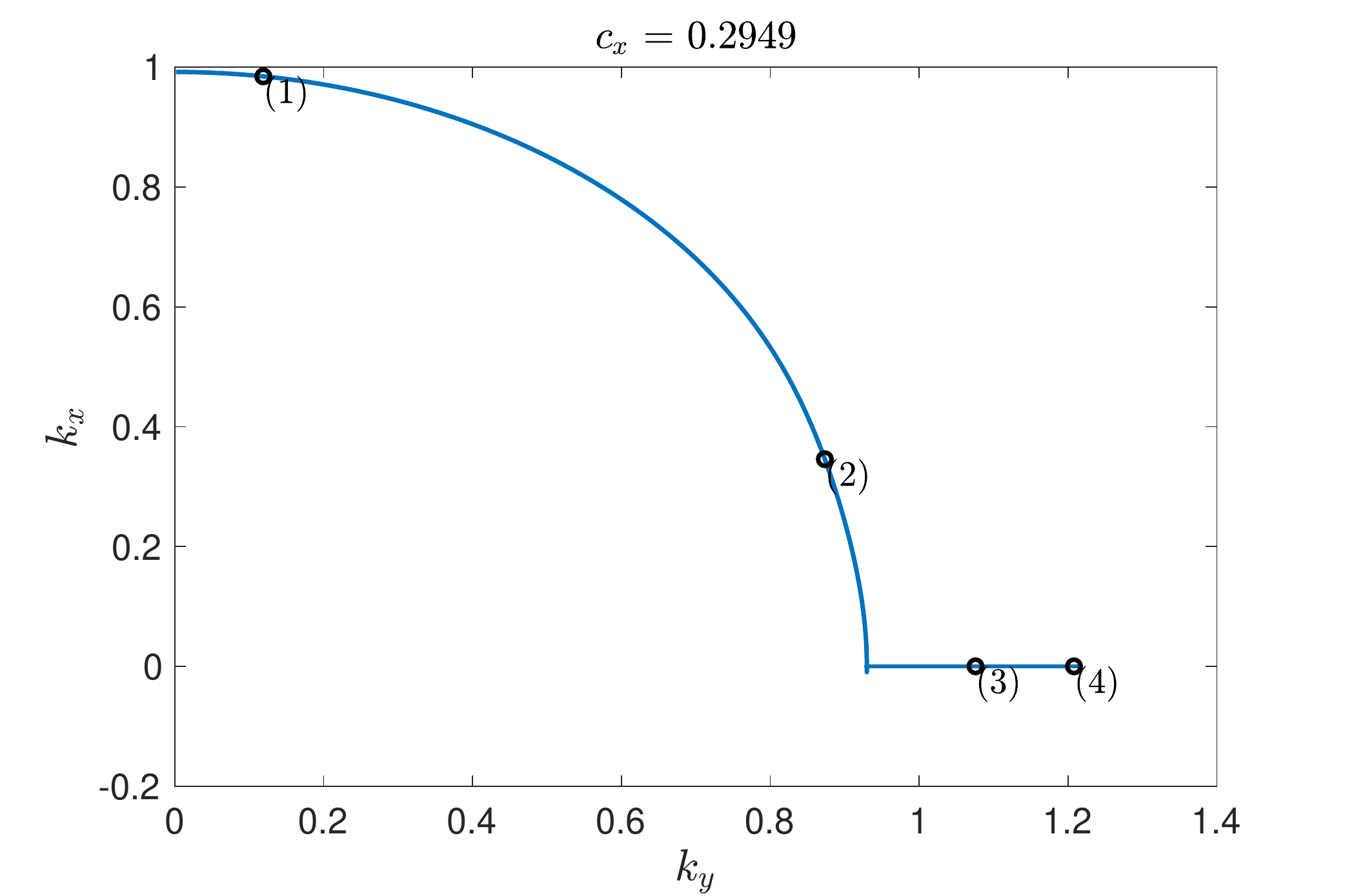}
\includegraphics[width=.48\textwidth,trim={1.5cm 0cm 1.8cm 0},clip]{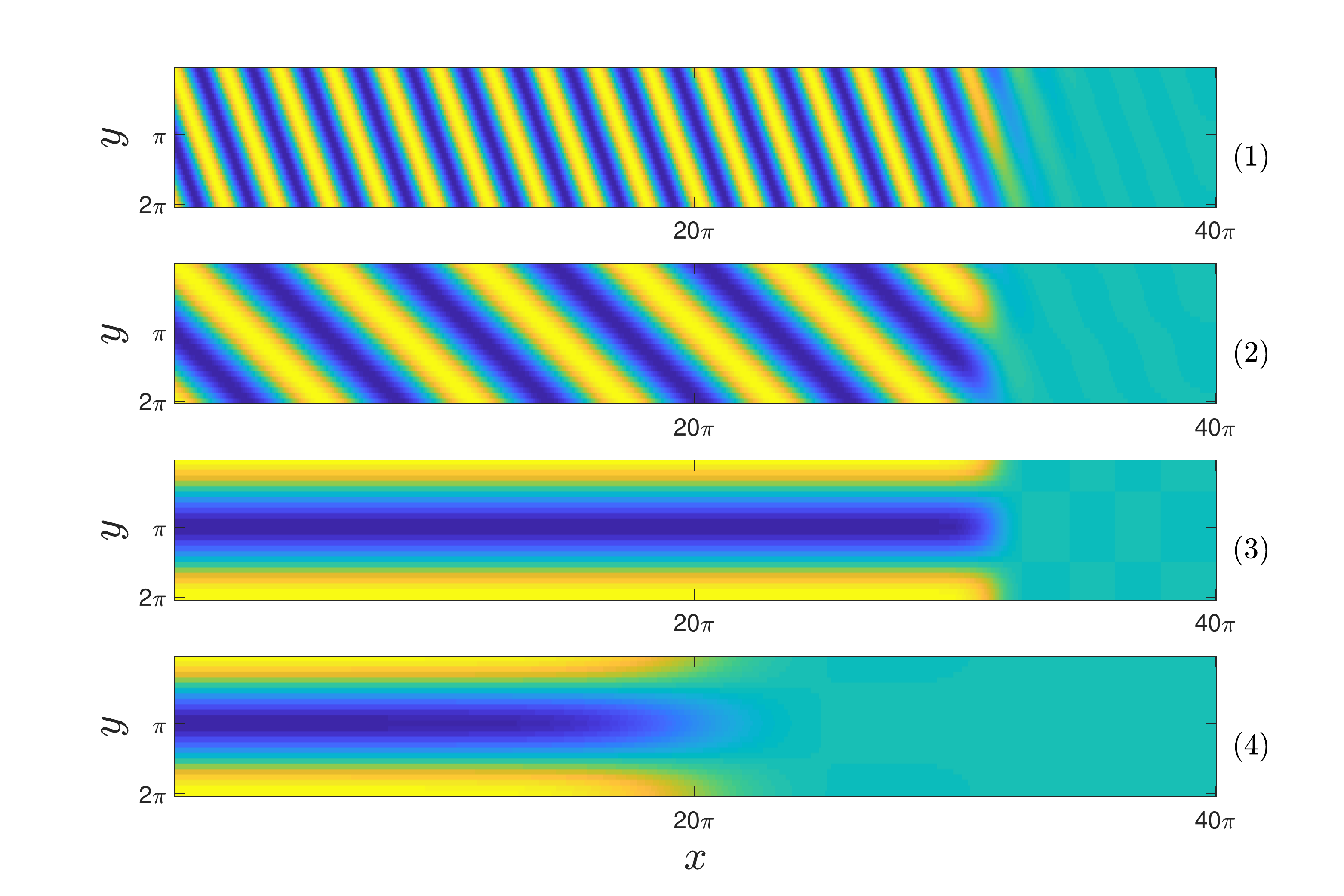}\\
\includegraphics[width=.49\textwidth,trim={0.75cm 0 1.0cm 0},clip]{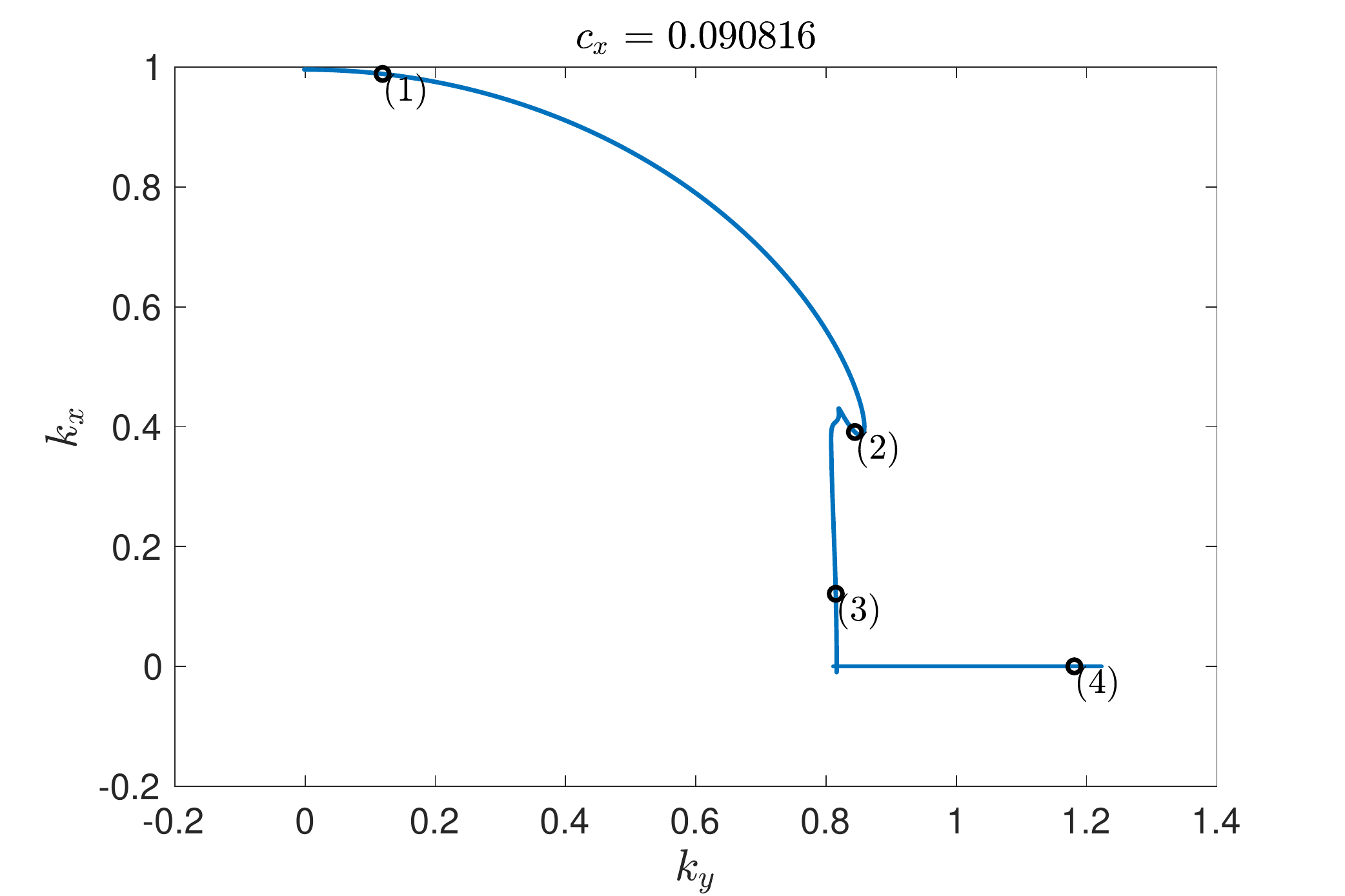}
\includegraphics[width=.48\textwidth,trim={1.5cm 0cm 1.8cm 0},clip]{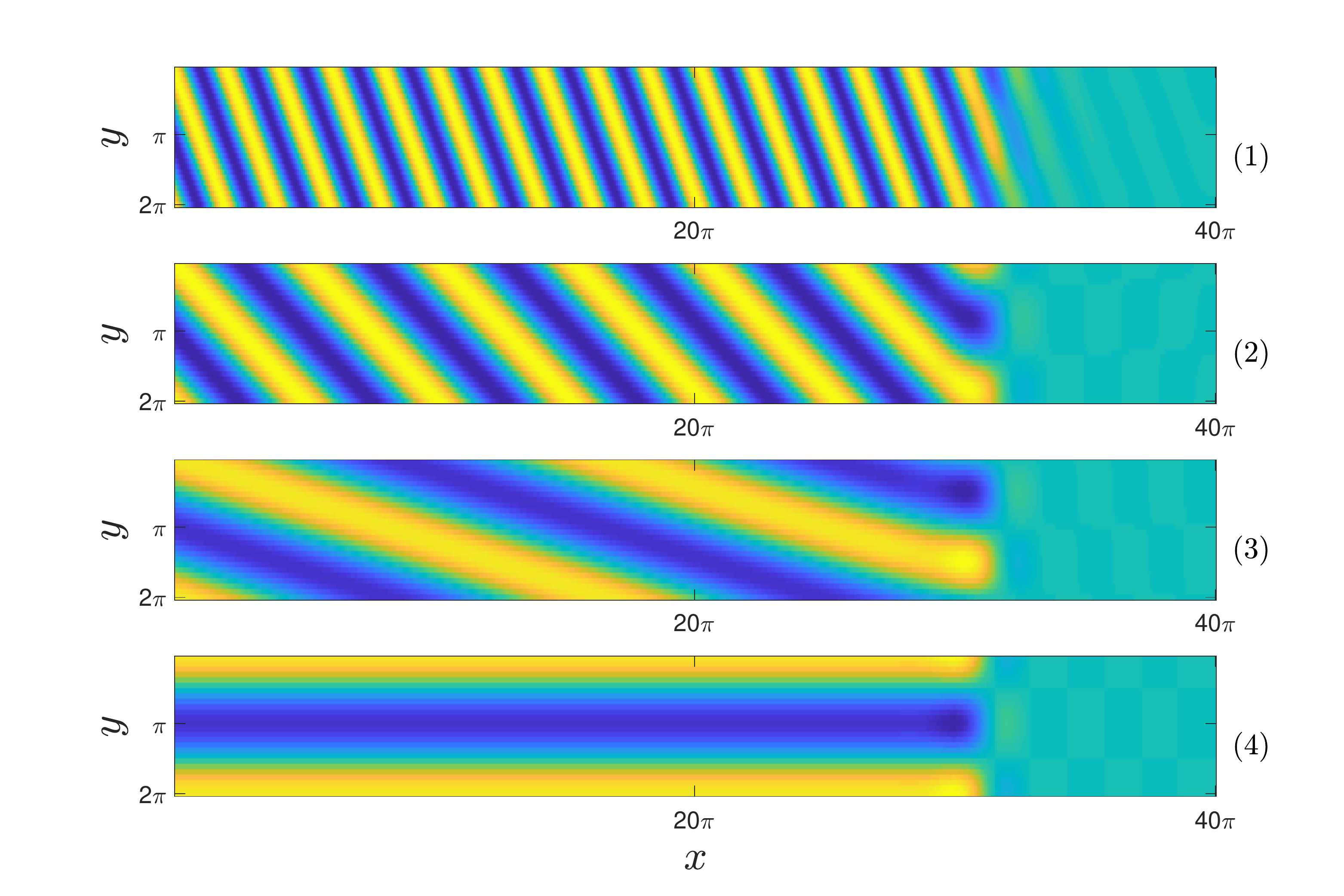}
\caption{Slices of the moduli space with $c_x$ fixed at $0.2949$ (top left) and $0.0908$ (bottom left) with solution profiles at various points along the surface (top and bottom right)}\label{fig:modsl}
\end{figure}




\subsection{Kink-dragging and the delicate arch}\label{s:6.3}
We next present continuation results of the kink-dragging bubble, visible as the ``delicate arch'' in Figure \ref{f:end}. As discussed  in \S\ref{s:2.4}, $y$-dependent patterned solutions with $c_x=0$ must select the critical zigzag curve $\{(k_x,k_y,c_x) \,|\,c_x  = \,0,\,k_\mathrm{zz}^2= k_x^2 + k_y^2\}$. Indeed starting with $k_x =0$ and continuing  in decreasing $k_y$ with $c_x=0$ fixed, oblique stripes bifurcated at the zig-zag critical wavenumber $k_y = k_\mathrm{zz}$ and the resulting curve conserves the bulk wavenumber $k = k_\mathrm{zz}$; see black curve in Figure \ref{fig:sh-zzbubble}.  These solutions were then used as initial guesses to continue solutions in $c_x$ with $k_y$ fixed (blue curves in Figure \ref{fig:sh-zzbubble}, left panel)\footnote{See \textsc{kink\textunderscore ky9.930188e-01.m4v} in supplementary materials for movie solutions along a slice of the bubble. }.

We also used the bifurcation curves obtained in \S\ref{s:4.1} for the kink-dragging bubble in the Cahn-Hilliard system \eqref{e:chtw00} to obtain a prediction for the corresponding bubble in Swift-Hohenberg. Letting $(c_{x,ch},\eta_{ch}(c_{x,ch}))$ denote the bifurcation curves for the speed and angle of stripes in the Cahn-Hilliard system, appropriate scalings yield the Swift-Hohenberg prediction,
 \begin{equation}\label{e:ch-pred}
 k_y = k_\mathrm{zz} - \zeta, \qquad  k_{x} = \sqrt{2}\,k_\mathrm{zz}\,\zeta^{1/2}\, \eta_\mathrm{ch} ,\qquad c_{x} = 8\,\zeta^{3/2}c_{x,\mathrm{ch}};
 \end{equation}
 see Figure \ref{fig:sh-zzbubble}. We find that the numerically predicted saddle-node curve (green) obtained from $(c_{x,\mathrm{ch}}^\mathrm{sn},\eta_{\mathrm{ch}}^\mathrm{sn})(\zeta)$,  asymptotically agrees well with the saddle-node curve in Swift-Hohenberg (red curve).


\begin{figure}[h]\centering
\includegraphics[width=0.48\textwidth,trim={1.5cm 0 1.5cm 0}]{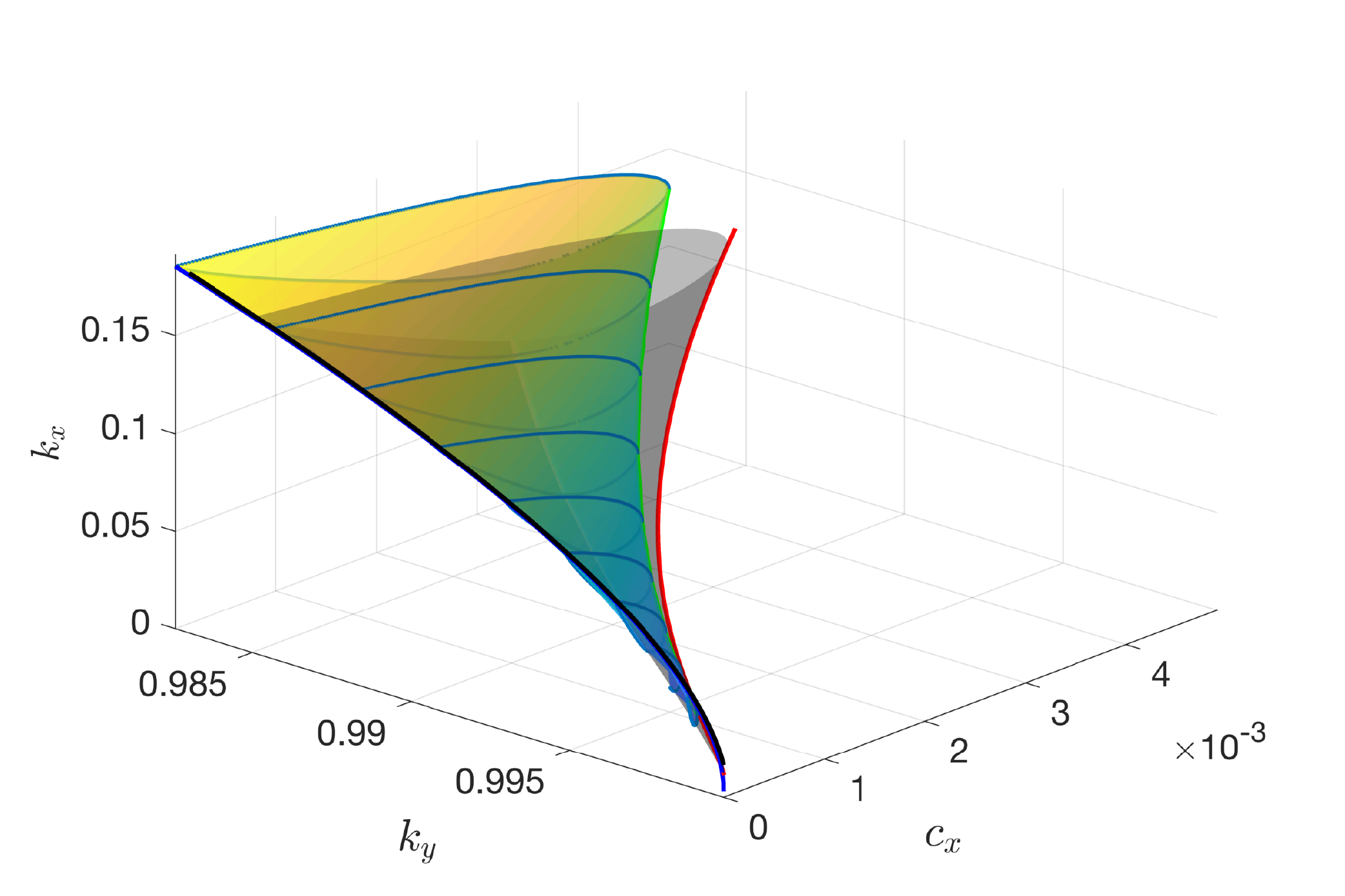}\hfill 
\raisebox{0.12in}{\includegraphics[width=0.37\textwidth,trim={1.5cm 0.5 2.5cm 0}]{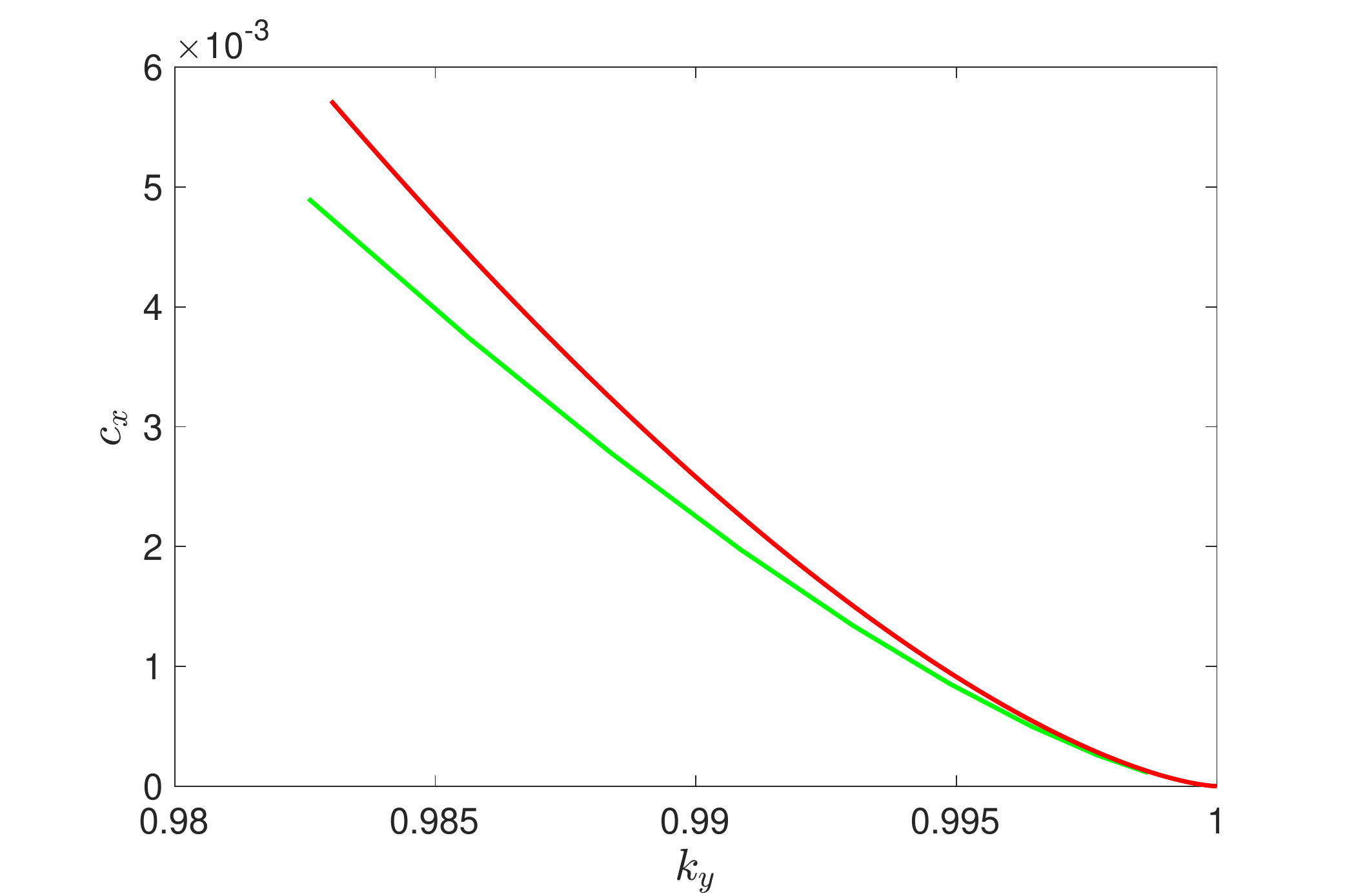}}\quad \qquad $ $\\
\includegraphics[width=0.45\textwidth,trim={1.5cm 0 0.5cm 0}]{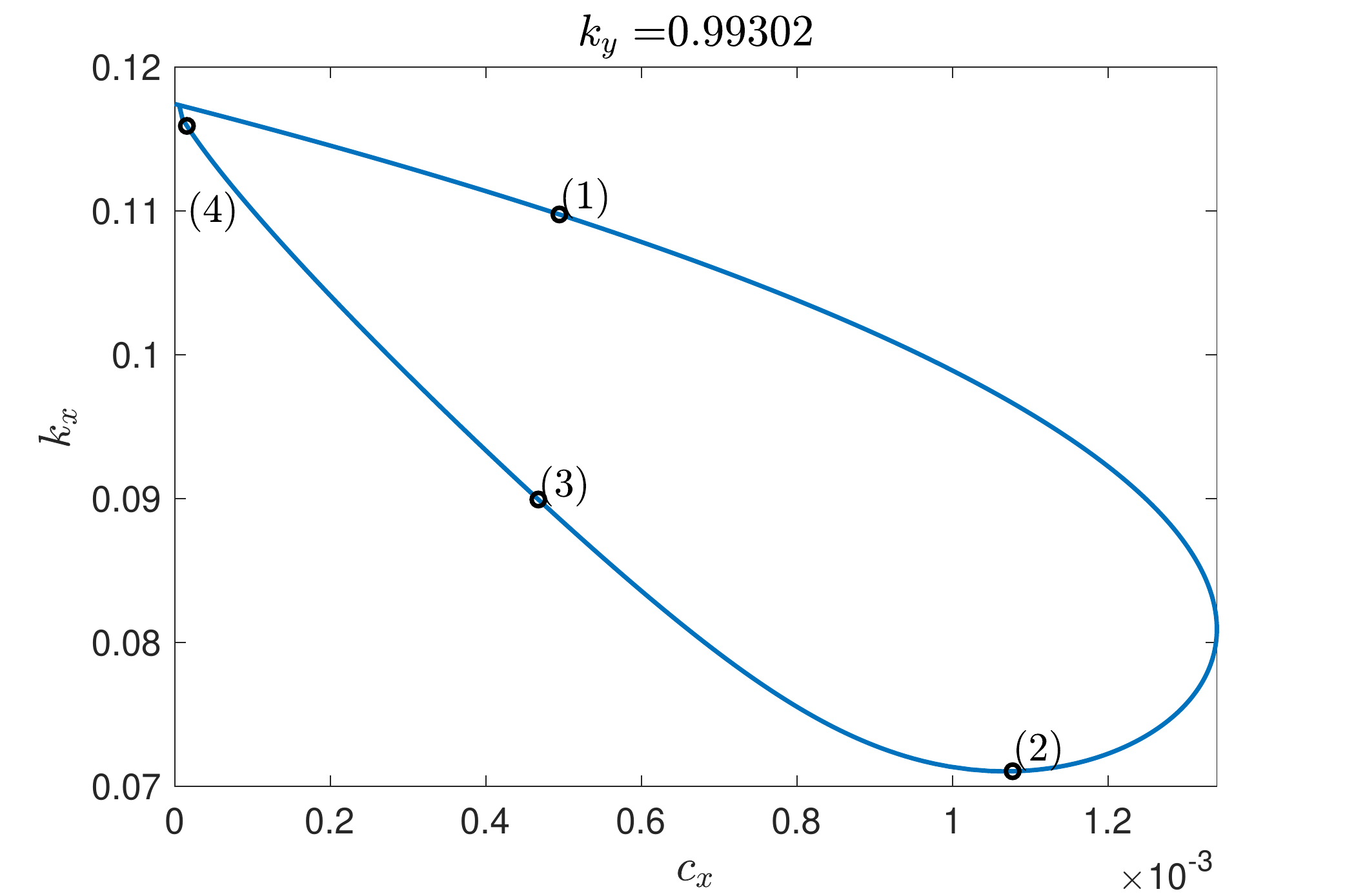}\hfill 
\includegraphics[width=0.5\textwidth,trim={1.5cm 1.5cm 2.5cm 0.5cm}]{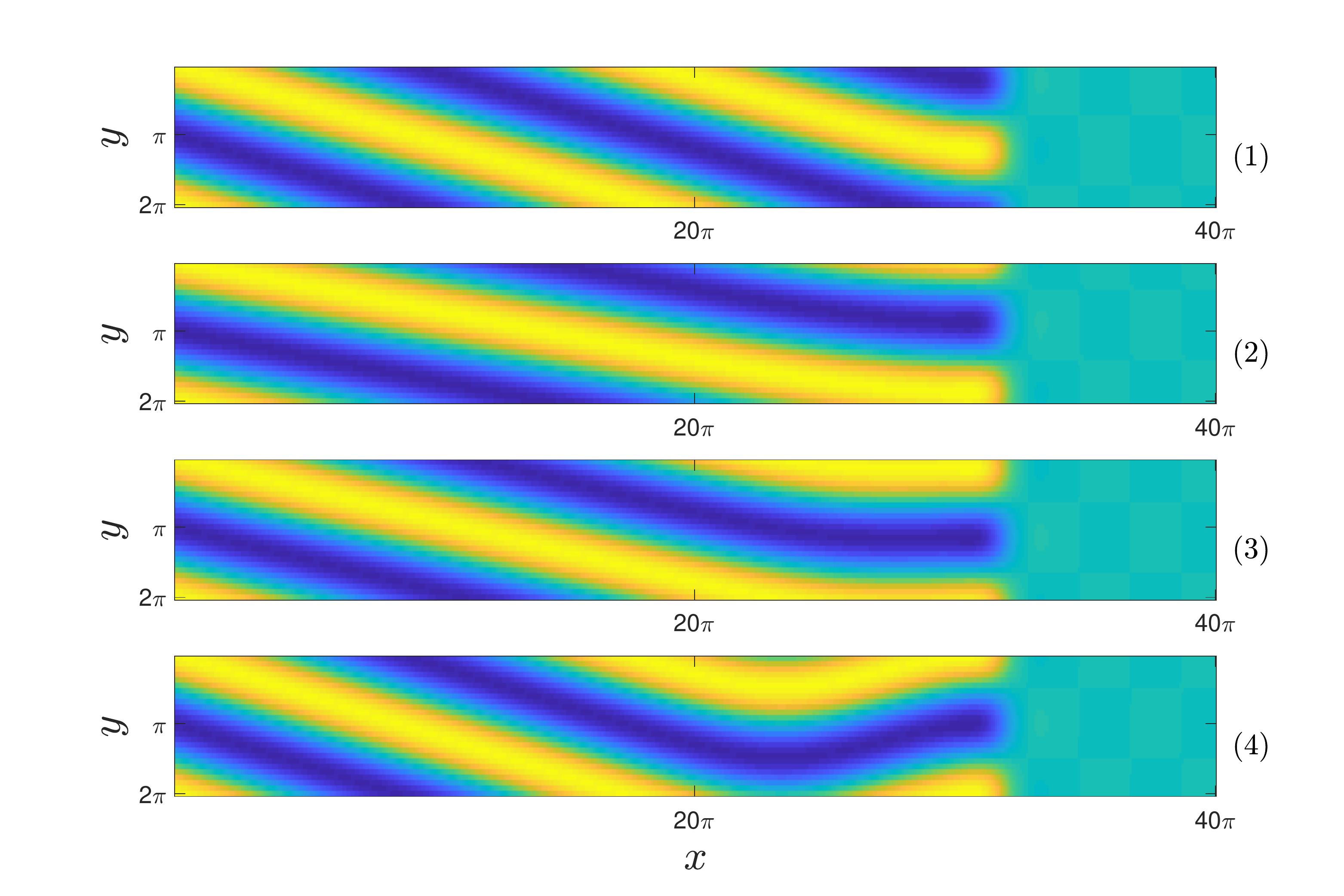}
\caption{Top row: Moduli space of the kink-dragging bubble (left) as interpolated surface (colored) from data obtained from numerical continuation for \eqref{e:shtwc} in $c_x$ (blue dots), compared against Cahn-Hilliard asymptotics \eqref{e:ch-pred} with Dirichlet boundary conditions (grey surface).  Green and red curves denote the fold curve in Swift-Hohenberg and Cahn-Hilliard respectively. Black curve gives the stationary zig-zag critical curve $\{(k_x,k_y,c_x)\,: k_{zz}^2 = k_x^2 + k_y^2,\, c_x = 0\}$ discussed in Sec. \ref{s:2.3}. Also shown, projection of saddle-node curves onto the $(k_y,c_x)$-plane (right). Bottom row: Cross section of kink-dragging bubble, continuation data with $k_y = 0.99302$ fixed (left) and solution profiles for points along kink-dragging curve (right), corresponding to labels in the left figure.}
\label{fig:sh-zzbubble}
\end{figure}


%
%
%

\subsection{Periodic detachment, oblique reattachment, and all-stripe detachment: the wing}\label{s:6.4}

We next compare the predictions given by the Newell-Whitehead-Segel reduction from \S\ref{s:5} with the Swift-Hohenberg moduli surface in the oblique-stripe reattachment regime.
%

\begin{figure}[h!]\centering
\includegraphics[width=0.5\textwidth,trim={1.5cm 0 0.5cm 0}]{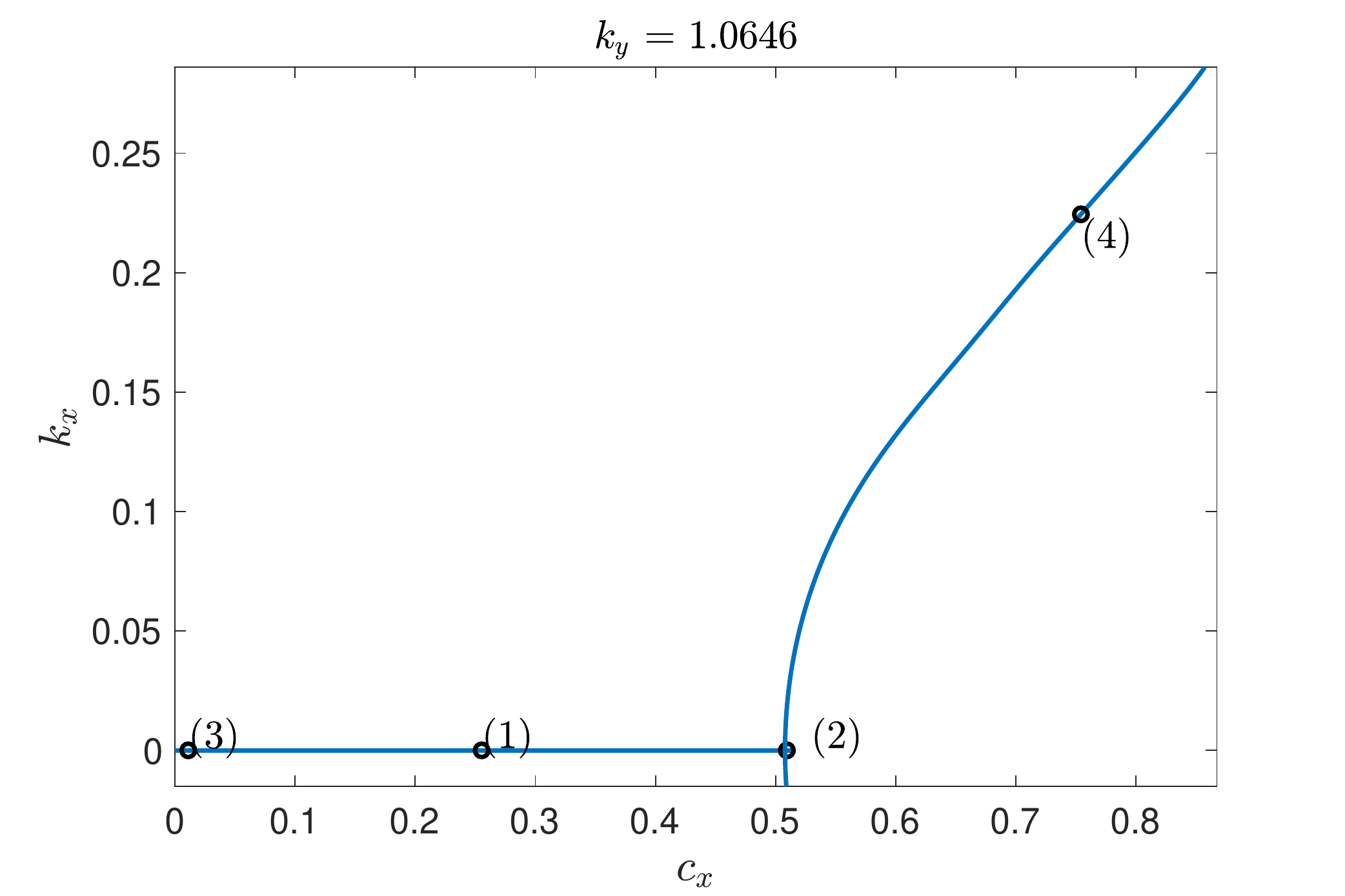}\hfill
\includegraphics[width = 0.49\textwidth,trim={1.5cm 0 2cm 0}]{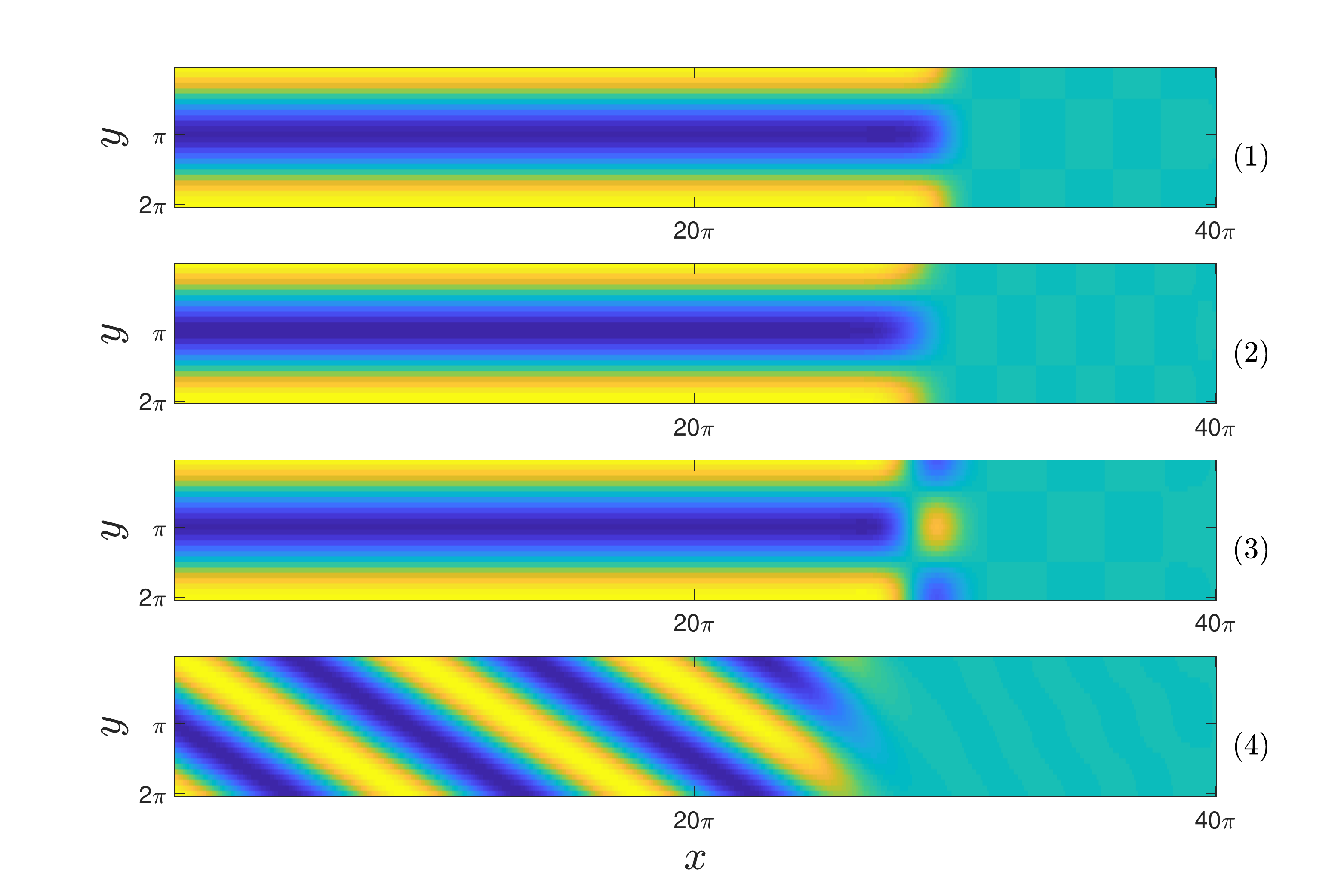}
\caption{Left: Cross-section of the Swift-Hohenberg moduli space for $k_y = 1.0646$ fixed, near the oblique stripe reattachment point. Right: Solution profiles along the perpendicular and oblique curves. Profiles $(1)$ and $(2)$ lie on the stable branch of the fold while profile $(3)$ lies on the unstable branch.  }
\label{fig:sh-obsl}
\end{figure}

\begin{figure}[h!]\centering
\includegraphics[width=0.5\textwidth,trim={0.5cm 1.9cm 0.5cm 0.5cm}]{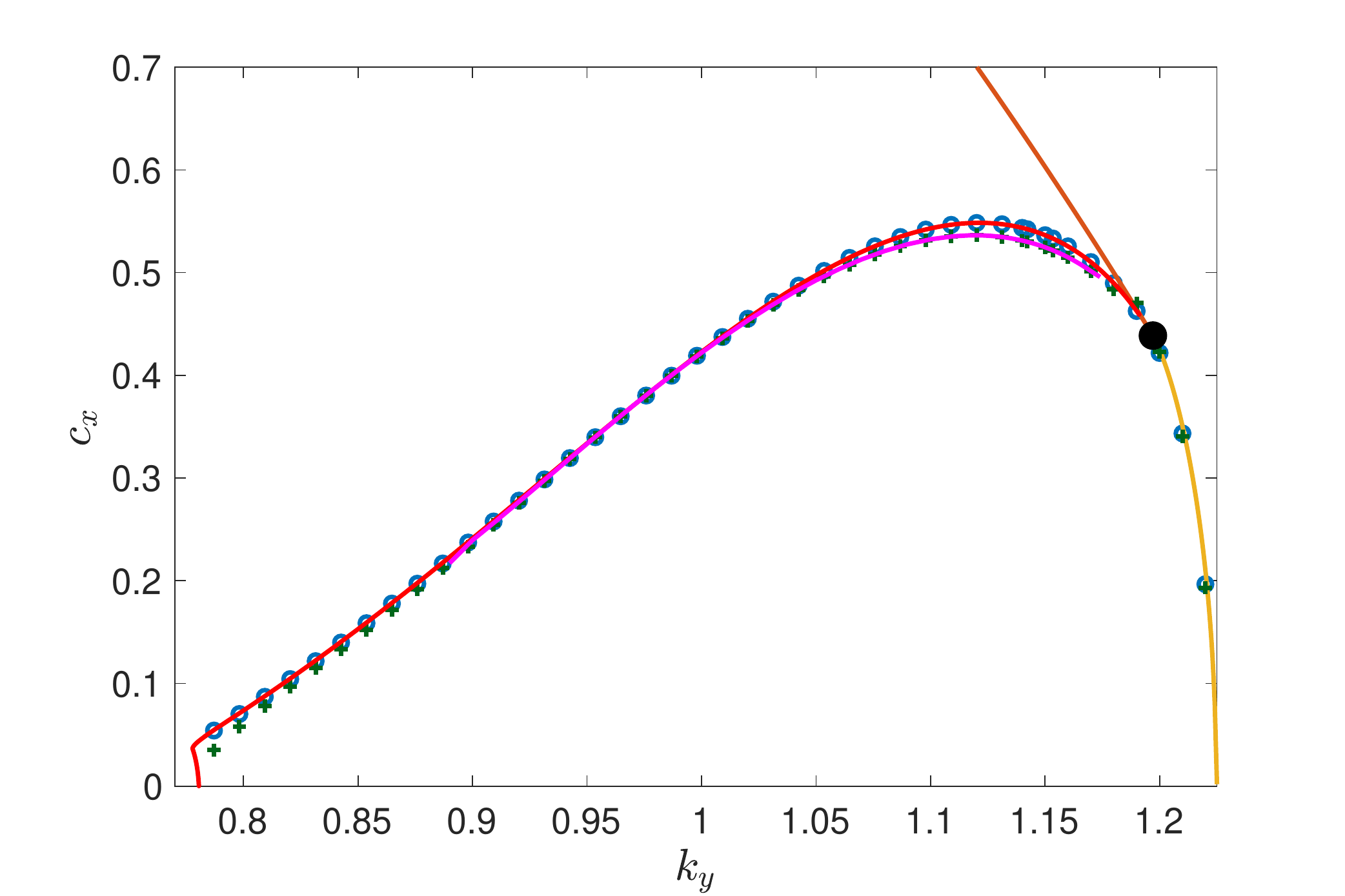}\hfill
\includegraphics[width=0.5\textwidth,trim={0.5cm 1.9cm 0.5cm 0.5cm}]{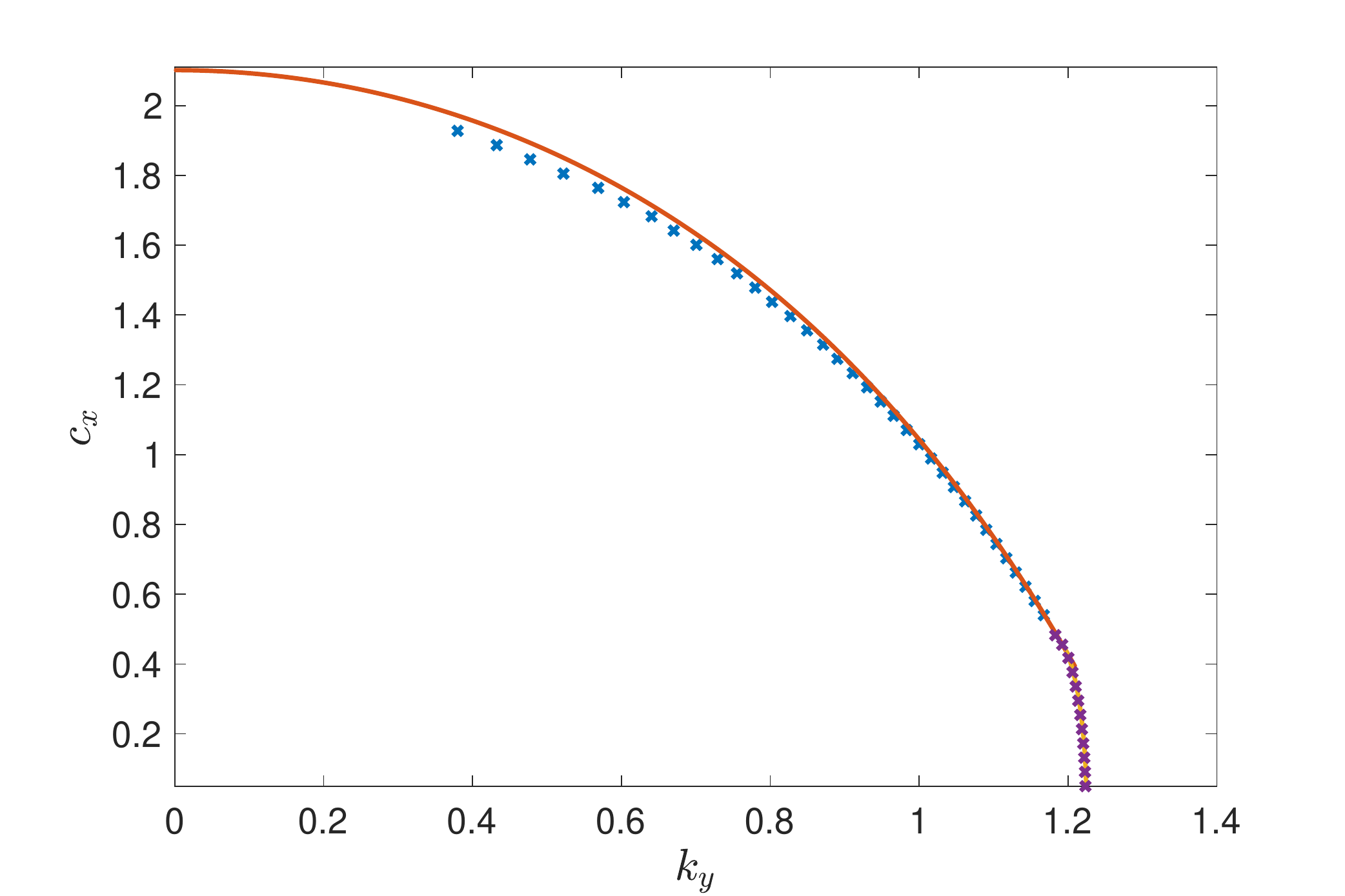}
\caption{Left: Comparison of the pitchfork and saddle-node curves found in Swift-Hohenberg (blue and green points) with predictions from Newell-Whitehead-Segel, Figure \ref{f:perp-bif} (magenta and red curves). Also plotted is the perpendicular stripe detachment curve (yellow) and the projection of the oblique stripe detachment curve.
Right: Plot of numerical detachment points (i.e. where the core solution dropped below a certain max norm near the boundary of the far-field) for both oblique (blue x's) and perpendicular (purple x's) stripes against the corresponding predictions from the linear spreading $c_\mathrm{lin}(k_y)$ for both $k_y<\sqrt{\frac{2+\sqrt{3\mu}}{2}}$ (orange curve) and $k_y>\sqrt{\frac{2+\sqrt{3\mu_0}}{2}}$ (yellow curve).
  }
\label{fig:sh-re-comp}
\end{figure}

As predicted in \S\ref{s:5.1}, continuing perpendicular stripes in $c_x$ for $k_y$ fixed less than $\sqrt{\frac{4+\sqrt{3}}{2}} \sim 1.1971,$ the solution destabilizes in a saddle-node bifurcation, after which the unstable branch undergoes a secondary pitchfork bifurcation from which the oblique stripes bifurcate\footnote{See \textsc{movie\textunderscore ky1.064602e+00.m4v} in supplementary materials for video of how solution varies this slice of
moduli space}; see Figure \ref{fig:sh-obsl}.  The bifurcating oblique stripes then continue up to the detachment curve $(k_x,k_y,c_x) = (k_\mathrm{lin}(k_y),k_y,c_\mathrm{lin}(k_y))$ predicted by the linear spreading speed calculated in \S\ref{s:5.2}.  Note that after the fold bifurcation, the perpendicular stripes develop a phase-kink in the vertical direction (see solution (3) of Figure \ref{fig:sh-obsl}), indicating how the nearby periodic solutions will evolve as shown in Figure \ref{f:sn}. Figure \ref{fig:sh-re-comp} compares the saddle-node, pitchfork, and detachment points, to the corresponding predictions from previous sections. 


\subsection{Hyperbolic and elliptic catastrophes at detachment interaction}\label{s:6.5}

Also, our solutions show that some reattachment curves near the lower $k_y$ boundary of the perpendicular attachment regime actually bend back and connect with the critical zig-zag curve at $c_x = 0,$ with the solution developing a kink at its interface\footnote{See \textsc{Catastrophe\textunderscore ky8.42e-01.m4v} and \textsc{Catastrophe\textunderscore ky8.4712e-01.m4v} in supplementary materials for movies of how solutions vary around catastrophe.}; see Figure \ref{fig:sh-zzsl}.  This happens when the kink-dragging bubble merges with the oblique-stripe reattachment surface, causing the top branch of the kink-dragging bubble to continue to the all-stripe detachment curve and the bottom branch of the bubble to connect with the perpendicular stripes. Locally, the reconnection can be described by Morse theory as a family of \emph{hyperbolas} forming a hyperboloid,  $\delta k_y\sim \kappa_1^2-\kappa_2^2$ where $\kappa_j$ are local coordinates in the $(k_x,c_x)$-plane. 

\begin{figure}[h!]\centering
\includegraphics[width=.48\textwidth,trim={0.5cm 0 1.5cm 0},clip]{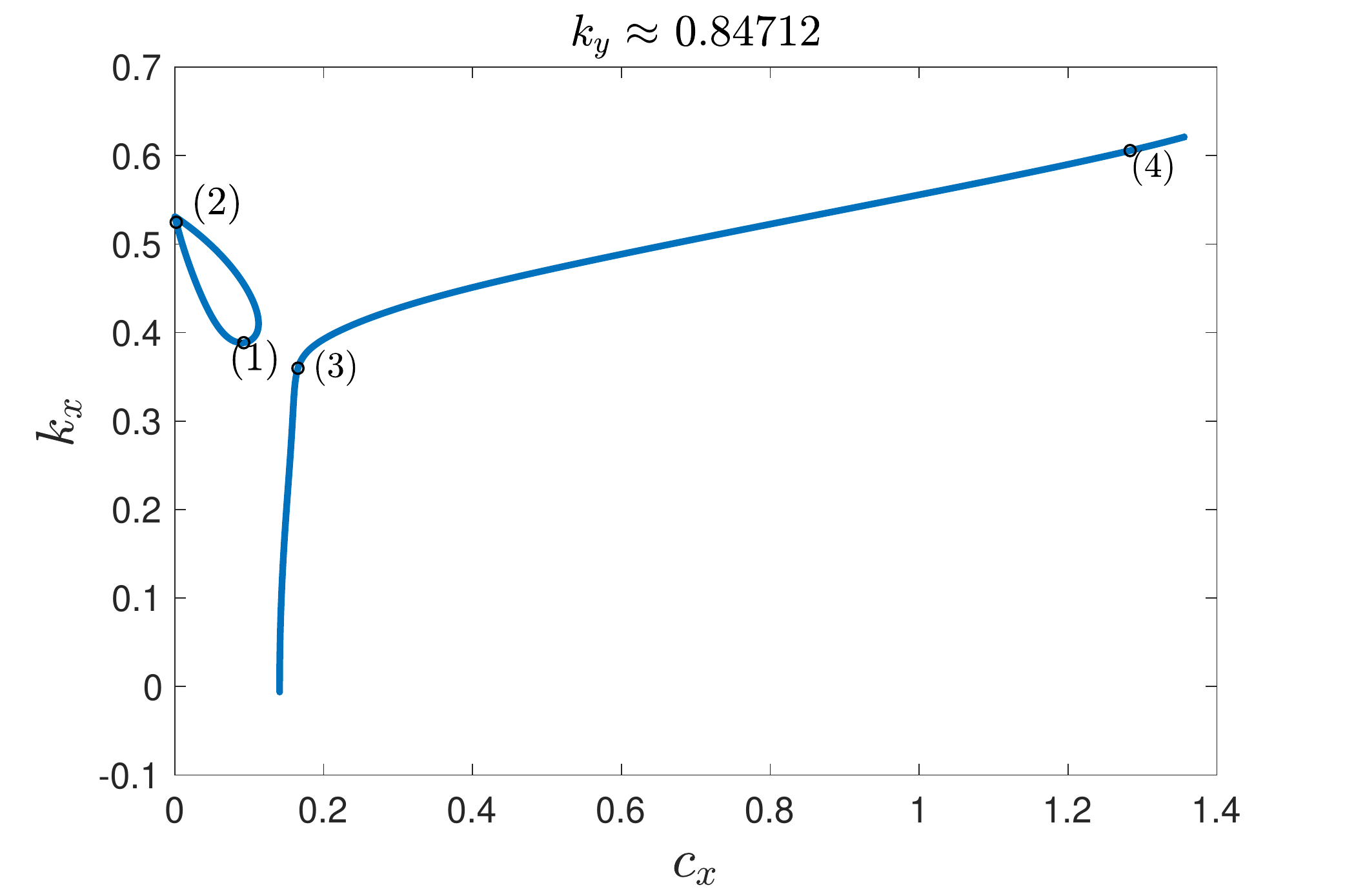}\hfill
\includegraphics[width=.49\textwidth,trim={1.5cm 0 1.8cm 0},clip]{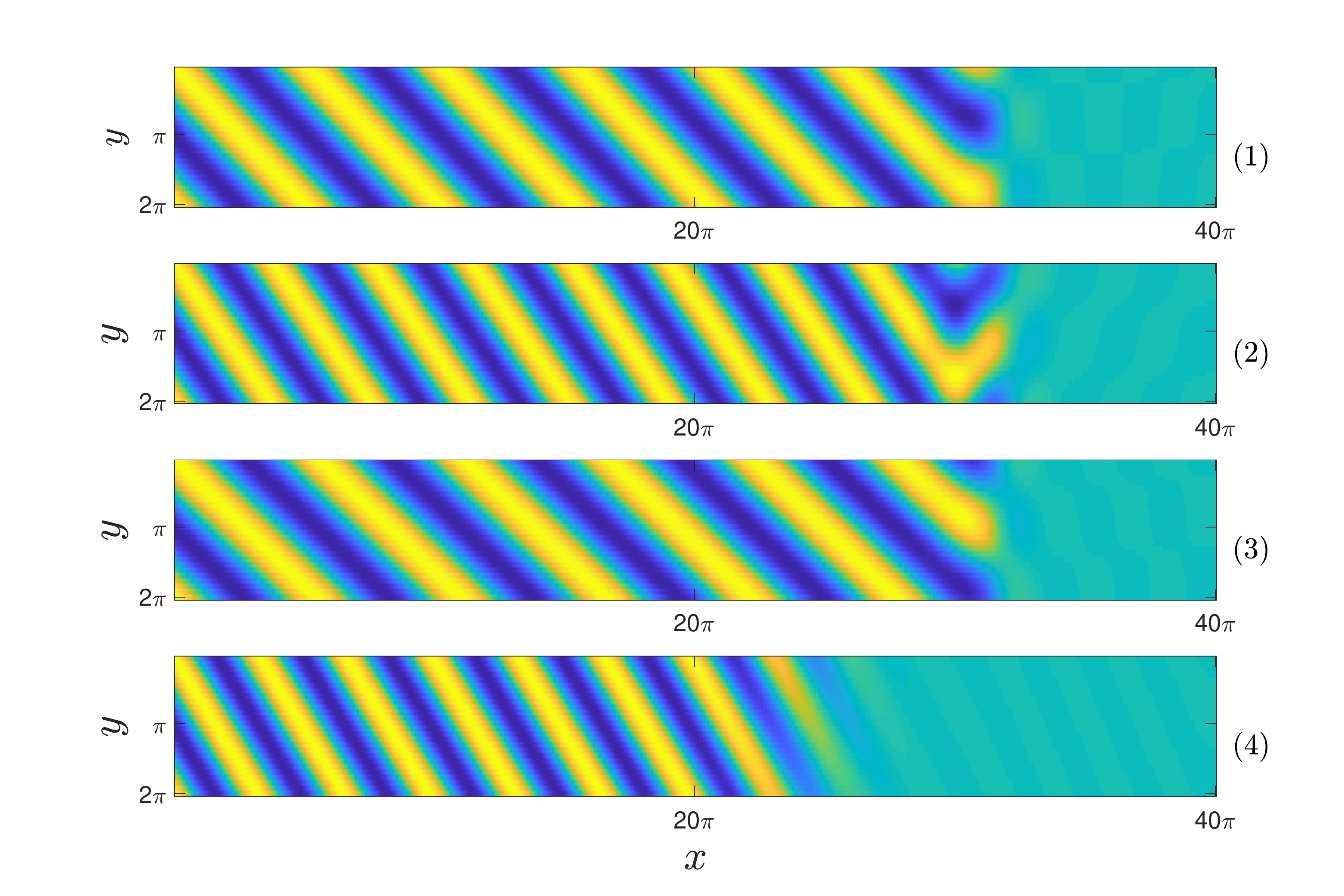}\\
\includegraphics[width=.48\textwidth,trim={0.5cm 0 1.5cm 0},clip]{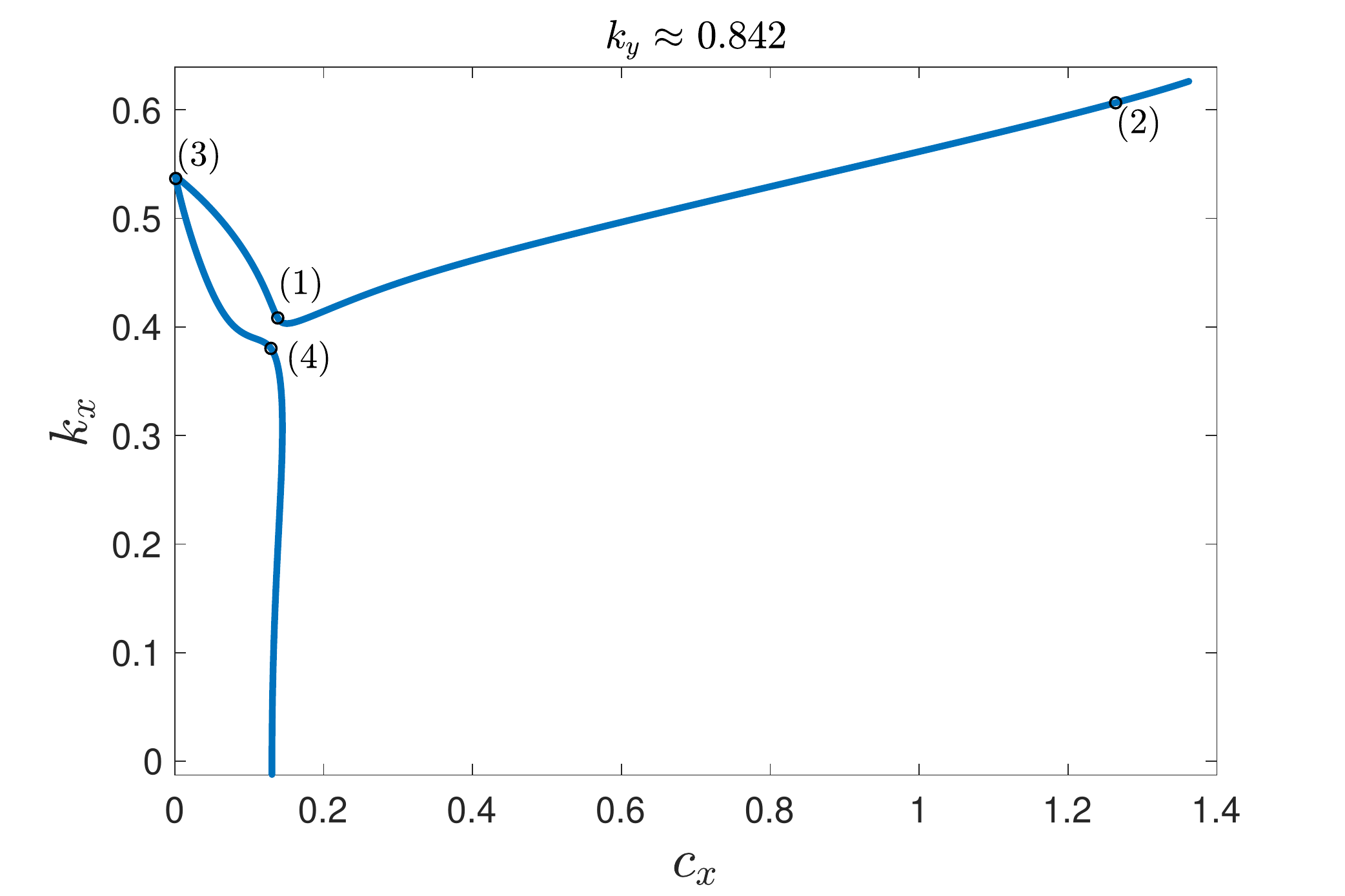}\hfill
\includegraphics[width=.49\textwidth,trim={1.5cm 0 1.8cm 0},clip]{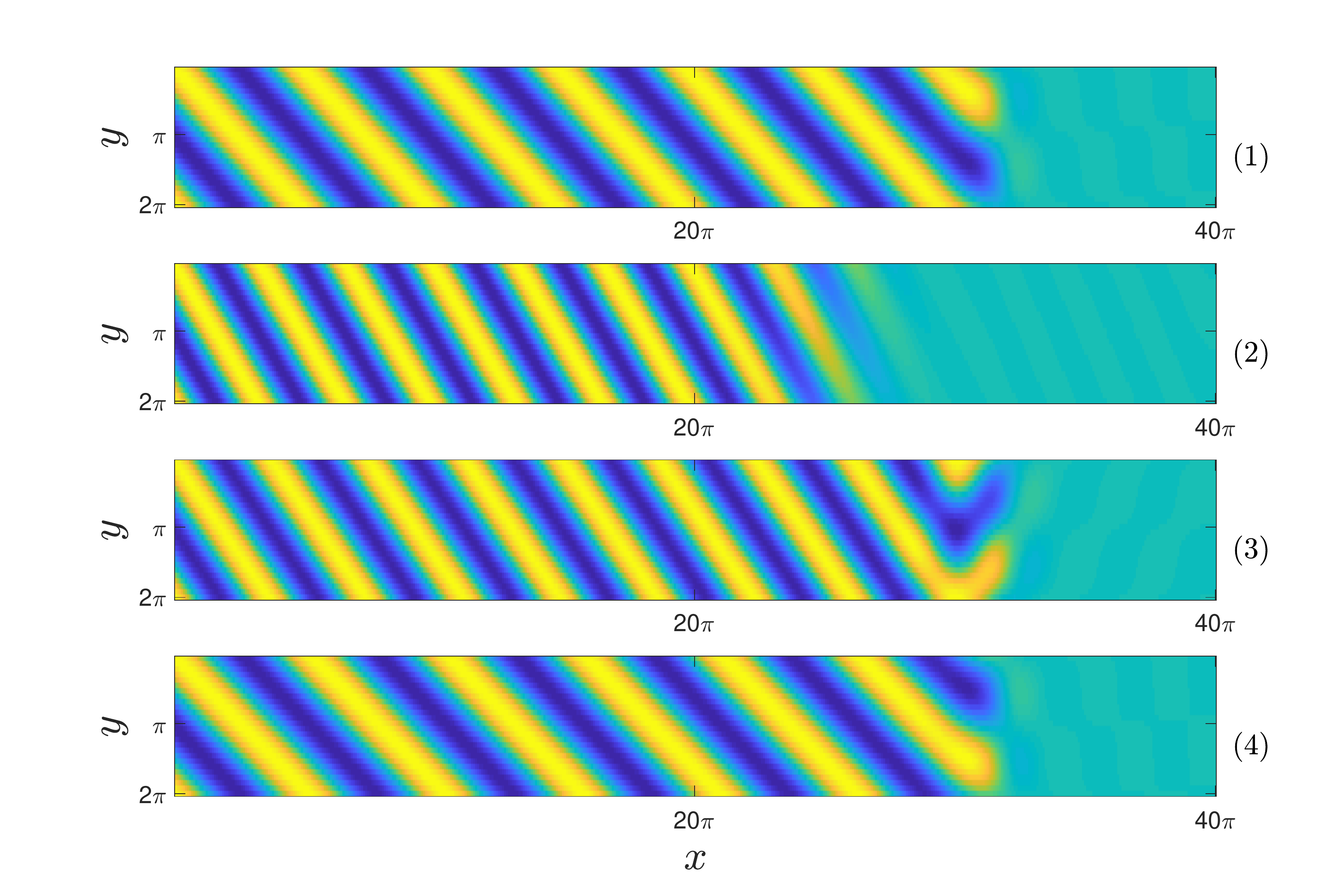}\\
\caption{Catastrophe where kink-dragging bubble merges with the oblique reattachment surface as $k_y$ is decreased. Cross-section of the moduli space on the left, corresponding solution profiles on the right. For larger $k_y$ (top row), the reattached oblique stripes eventually detach; For smaller $k_y$ (bottom row), the reattached oblique stripes turn back and connect with the zig-zag mode at $c_x = 0$.   }
\label{fig:sh-zzsl}
\end{figure}
Meanwhile, the saddle-node curve of perpendicular stripes reaches a minimum  near $k_y\sim 0.777$ before snaking around $c_x=0$; see Figure \ref{f:perp-bif}. Towards these smaller $k_y$-values, perpendicular stripes are confined to a finite interval of $c_x$-values where they form an isola between the two branches of the saddle-node, before disappearing. Near the singularity, the isolas have the shape of \emph{ellipses}, forming a paraboloid,  $\delta k_y\sim \kappa_1^2+\kappa_2^2$ for local coordinates $\kappa_j$ in the $(c_x,\|w\|)$-plane. 

Increasing $k_y$ slightly from this elliptic singularity, $k_y\sim 0.778$, we observe two new saddle-nodes emerging in a cusp singularity and the isolas form  figure-eight shaped curves;  see Figure \ref{fig:sh-isola}.  We also found an a small branch of oblique stripes with $k_x\sim 0$ which bifurcate off of and reattach to the upper branches of the figure-eight isolas\footnote{See \textsc{barba\textunderscore ky7.8e-01.m4v} in supplementary materials for movie of solutions along this isola.}; see Figure \ref{fig:sh-hat}. We suspect that the isolas continue into a more complex scenario of broken up snakes and ladders as observed for instance in \cite{mccalla,makrides}.

\begin{figure}[h!]\centering
\includegraphics[width=.99\textwidth,trim={1.5cm 0 1.5cm 0},clip]{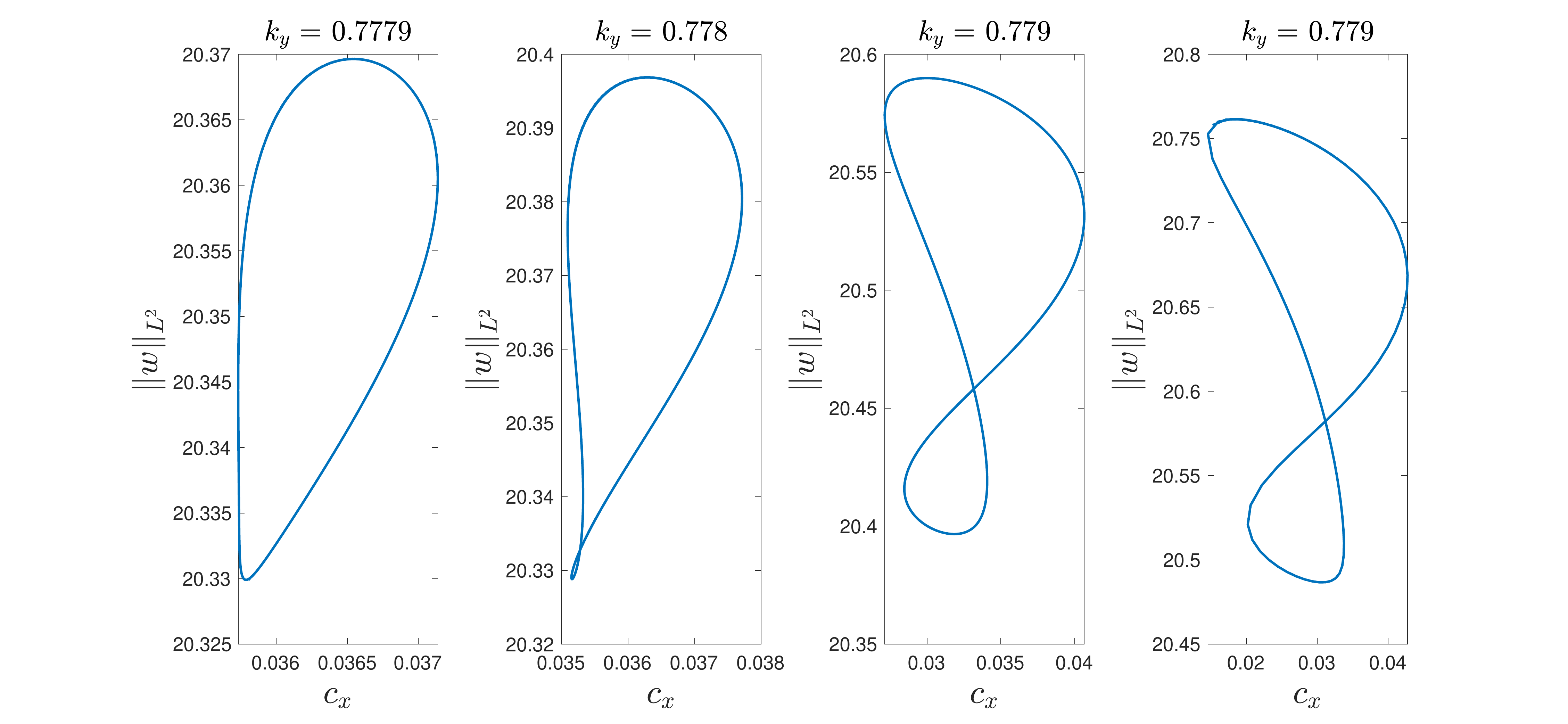}\hfill
\caption{Isola curves of perpendicular stripes near the parabolic catastrophe, plotting $L^2$-norm of the core-solution against $c_x$,  for $k_y = 0.7779, 0.778 ,0.779 ,0.78$ (from left to right). }
\label{fig:sh-isola}
\end{figure}

\begin{figure}[h!]\centering
\includegraphics[width=.49\textwidth,trim={1.5cm 0 1.5cm 0},clip]{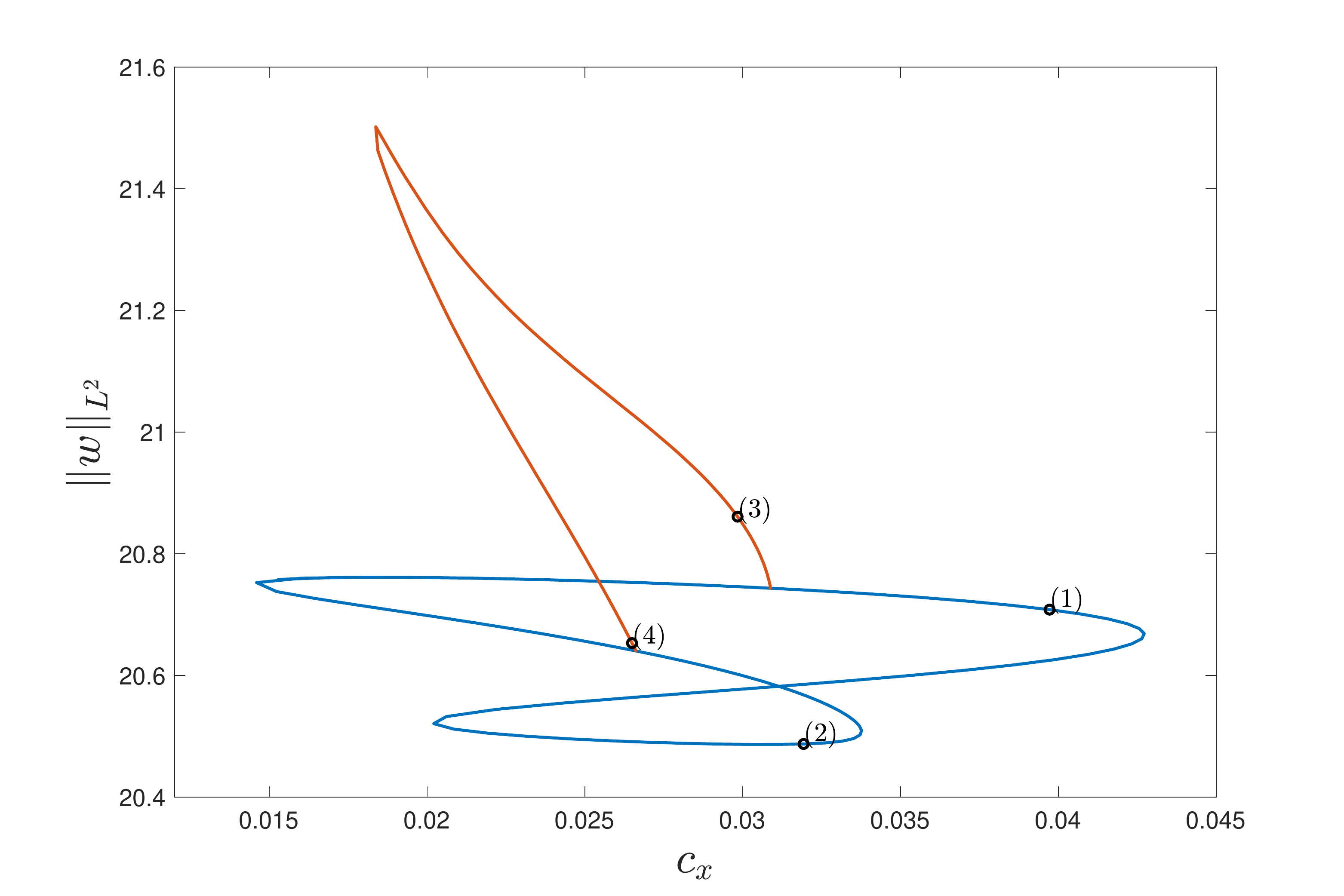}\hfill
\includegraphics[width=.49\textwidth,trim={1.5cm 0 1.8cm 0},clip]{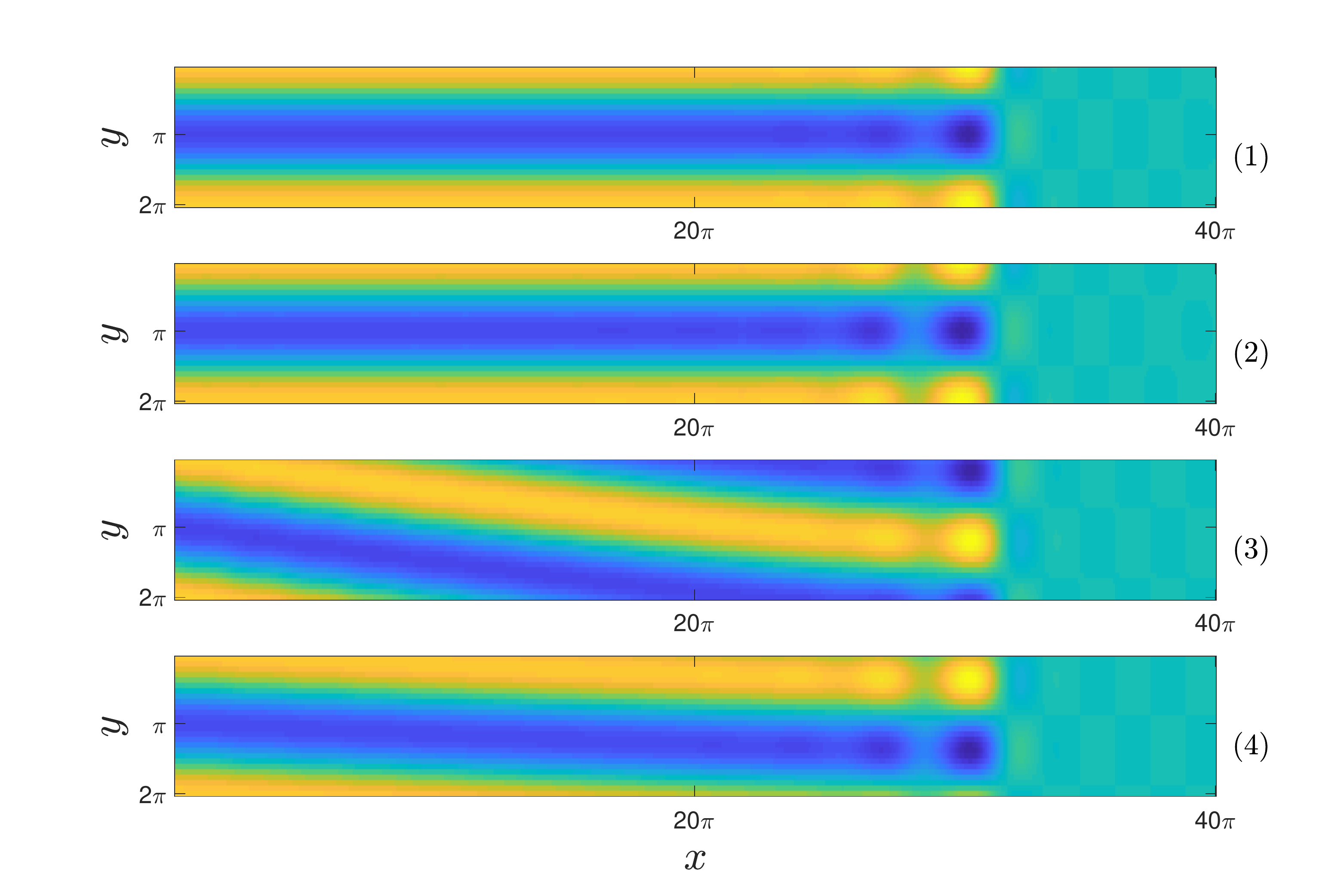}
\caption{Perpendicular figure-eight isola (blue) for $k_y = 0.78$, where a small branch of oblique stripes ($\max k_x \sim 0.12$) bifurcate as secondary pitchfork bifurcations (orange) . }
\label{fig:sh-hat}
\end{figure}


%
%

\section{Discussion}\label{s:d}

We analyzed the formation of stripes formed in the wake of a directional quenching process, relating in particular orientation and wavelength to the speed of propagation of the quenching line. Our work here focused on a region of transition between stripes formed perpendicular to the quenching line, and stripes formed at a small oblique angle. Our major findings illustrate that this transition is in fact quite subtle, organized by a variety of bifurcations of solutions, often accompanied by essential spectra. 

Among the key organizing recurring features are saddle-node bifurcations. Taking a perspective of increasing the rate of quenching $c_x$ for a fixed lateral wavenumber $k_y$, we can continue both perpendicular and oblique stripes until they undergo a saddle-node bifurcation. In both cases, the saddle-node bifurcation marks the release of a defect (or kink) from the quenching line. This kink is both visible in the continuation of the unstable branch from the saddle-node bifurcation point, as its distance from the quenching line increases with decreasing quenching speed $c_x$, and in the profiles visible for speeds past the critical saddle-node speed. Both saddle-node bifurcations are of independent theoretical interest. First, the bifurcation to periodic orbits resembles in many ways a saddle-node on a limit cycle, with many caveats reflected in the asymptotics of period and in the lack of an actual homoclinic orbit connecting the saddle-node equilibrium. Second, it would be very desirable to gain a theoretical understanding of the saddle-node that would ideally predict its location in the $(c_x,k_y)$-plane, possibly also explain the presence of the additional nearby pitchfork bifurcation for perpendicular stripes. Our analysis only gives such predictions near a detachment point and for small speeds, only. 

Beyond the transition from perpendicular to oblique stripes, we have analyzed detachment \cite{gs1,gs2}, $k_y\sim 0$ \cite{gs3}, and $k_y=0$, $c_x\sim 0$ \cite{beekie} in prior work. From this ``completist'' perspective, the major challenge appears to be a description of the moduli space in a vicinity of  $k_y,c_x=0$.

Our results are somewhat universal. Small speed predictions near the zigzag transition should hold quite universally for systems with such an instability. Moderate speed predictions should hold near onset, where amplitude equation approximations are valid. We did notice however subtle differences varying $\mu$ or, more generally, setting $\rr(x)=\mu_\pm$ for $\pm x>0$.  From this perspective, we hope that our computational approach in \S\ref{s:6} will help compare different systems systematically and quantitatively, in particular in regimes where subtle bifurcations and multi-stability make a direct mapping of parameter space through direct simulations unreliable. Interesting extensions here would include different types of parameter triggers, boundary conditions at $x=0$ rather than quenching, and non-variational effects, but also systems such as reaction-diffusion models from morphogenesis, possibly far from onset. 

Our results on perpendicular stripes in \S\ref{s:5} predicted transitions in direct simulations very well. In particular, the far-field instabilities of perpendicular stripes appeared to be the only limitations on observability other than the saddle-node and pitchfork bifurcations. We did not attempt such a stability analysis for oblique stripes, which would be both algebraically and computationally more involved, but would clearly complement and to some extent complete the analysis, here. 

Within the context of pattern formation, a very natural next question would point towards growth patterns when spots, in particular on hexagonal lattices, are the preferred states. Many of the tools here, in particular computational recipes from \S\ref{s:6} and amplitude equation approximations would still be available in this context. Changes in orientation of hexagonal lattices with respect to the quenching line are however subject to more complex pinning effects, as the lateral period for the creation of ideal energy-minimizing hexagons would be subject to a wealth of resonances as the relative angle varies; we refer to \cite{phyllo} for a study of hexagonal patterns formed in the wake of interfaces, with emphasis on periodicities and orientation of lattices in the example of phyllotaxis. 

%
%

 
\end{document}